%% file: paper.tex
\documentclass[letterpaper,11pt]{article}

\usepackage{amsmath,amsthm,amssymb,amsfonts}

\usepackage[mono=false]{libertine}
\usepackage{libertinust1math}
\usepackage[T1]{fontenc}

\usepackage[margin=1in]{geometry}
\usepackage{subfiles}
\usepackage{hyperref}
\hypersetup{colorlinks=true, citecolor=blue, linkcolor=red, urlcolor=blue}
\usepackage{enumerate}
\usepackage[shortlabels]{enumitem} % Provides the ability to easily modify the amount of white space in a list
\usepackage{graphicx}
\graphicspath{{figures/}}
\usepackage{subcaption}
\usepackage{titling}
\usepackage{multirow}
\usepackage{array}
\usepackage{booktabs}
\usepackage{wasysym} % for the command $\Square$ rather than just $\square$

\newtheorem{theorem}{Theorem}[section]
\newtheorem*{theorem*}{Theorem}
\newtheorem{corollary}[theorem]{Corollary}
\newtheorem{lemma}[theorem]{Lemma}
\newtheorem*{lemma*}{Lemma}

\theoremstyle{definition}

\newtheorem*{definition*}{Definition}
\theoremstyle{remark}
\newtheorem{remark}{Remark}[section]
\newtheorem*{notation}{Notation}

\numberwithin{equation}{section}

\newcommand{\figref}[1]{Figure~\ref{#1}}
\newcommand{\secref}[1]{Section~\ref{#1}}
\newcommand{\thmref}[1]{Theorem~\ref{#1}}
\newcommand{\lemref}[1]{Lemma~\ref{#1}}

\newcommand{\corref}[1]{Corollary~\ref{#1}}
\newcommand{\tabref}[1]{Table~\ref{#1}}

\newcommand{\sub}[1]{\raisebox{-.1\baselineskip}{\includegraphics[height=0.7\baselineskip, width=0.7\baselineskip, keepaspectratio]{sub_#1}}}
\newcommand{\subinmatrix}[1]{\raisebox{-.1\baselineskip}{\includegraphics[height=\baselineskip, width=\baselineskip, keepaspectratio]{sub_#1}}}

\date{}

\begin{document}

\title{Approximability of the Eight-vertex Model}
\author{
{Jin-Yi Cai}\thanks{Department of Computer Sciences, University of Wisconsin-Madison. Supported by NSF CCF-1714275. \texttt{jyc@cs.wisc.edu}}
\and {Tianyu Liu}\thanks{Department of Computer Sciences, University of Wisconsin-Madison. Supported by NSF CCF-1714275. \texttt{tl@cs.wisc.edu}}
\and
{Pinyan Lu}\thanks{ITCS, Shanghai University of Finance and Economics. \texttt{lu.pinyan@mail.shufe.edu.cn}}
\and
{Jing Yu}\thanks{School of Mathematical Sciences, Fudan University. \texttt{jyu17@fudan.edu.cn}}
}

\maketitle

\begin{abstract}
We initiate a study of
the classification of approximation complexity of the eight-vertex model
defined over 4-regular graphs.
The eight-vertex model, together with its special case the six-vertex model,
is one of the most extensively studied models in statistical physics,
and can be stated as a problem of counting weighted orientations in graph theory.
Our result concerns the approximability of the partition function on
all 4-regular graphs,
classified according to the parameters of the model.
Our complexity results conform to the phase transition phenomenon from physics.

We introduce a \emph{quantum decomposition} of the eight-vertex model and prove a set of \emph{closure properties} in various
regions of the parameter space.
Furthermore, we show that there are extra closure properties
on 4-regular planar graphs.
These regions of the parameter space are concordant with
the phase transition threshold.
Using these closure properties,
we derive polynomial time approximation algorithms
 via \emph{Markov chain Monte Carlo}.
We also show that the eight-vertex model is NP-hard
to approximate on the other side of the phase transition threshold.
\end{abstract}

%\clearpage
%\setcounter{page}{1}

\section{Introduction}\label{sec:intro}
\subfile{intro}

\section{Preliminaries}\label{sec:prelim}
\subfile{prelim}

\section{Closure Properties}\label{sec:properties}
\subfile{properties}

\section{FPRAS}\label{sec:fpras}
\subfile{fpras}

\section{Hardness}\label{sec:hardness}
\subfile{hardness}

\bibliography{reference}{}
\bibliographystyle{alpha}

%%\clearpage
%\appendix
%\section*{Appendix}\label{sec:appendix}
%\subfile{appendix}

\end{document}

%% file: intro.tex
\documentclass[paper]{subfiles}

Let us consider the following natural orientation problem
which is called 
 the eight-vertex model in statistical physics.
 Given a 4-regular graph $G$,
we consider all orientations of the edges
such that there is an even number of arrows into (and out of) each vertex. 
Such a configuration is called an \emph{even orientation}.
In the \emph{unweighted} case,  the problem is to
count the number of even orientations of $G$, and this
is computable in polynomial time~\cite{DBLP:journals/corr/CaiF17}.
In the general case
of the eight-vertex model there are 
\emph{weights} associated with 
local configurations,
and the problem is to compute a weighted sum called
the partition function. This becomes an  interesting
and challenging problem, and the complexity
picture becomes  more intricate~\cite{DBLP:journals/corr/CaiF17}.

\renewcommand{\thesubfigure}{-\arabic{subfigure}}
\begin{figure}[h!]
\centering
\begin{subfigure}[b]{0.12\linewidth}
\centering\includegraphics[width=\linewidth]{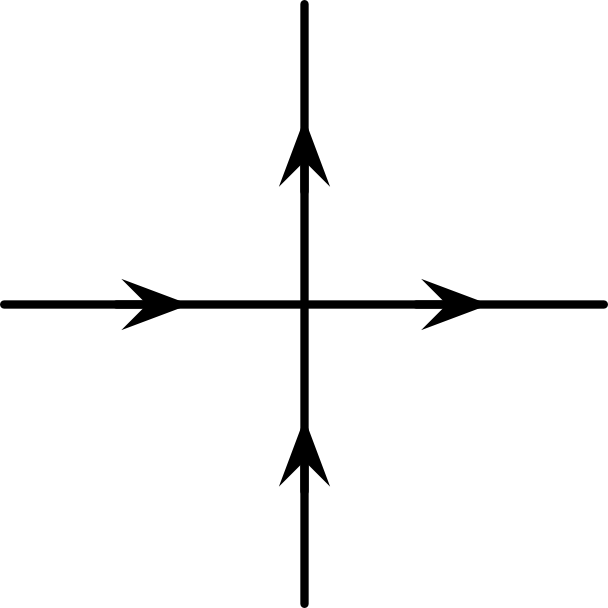}\caption{$1$}
\label{fig:orientations_1}
\end{subfigure}
\begin{subfigure}[b]{0.12\linewidth}
\centering\includegraphics[width=\linewidth]{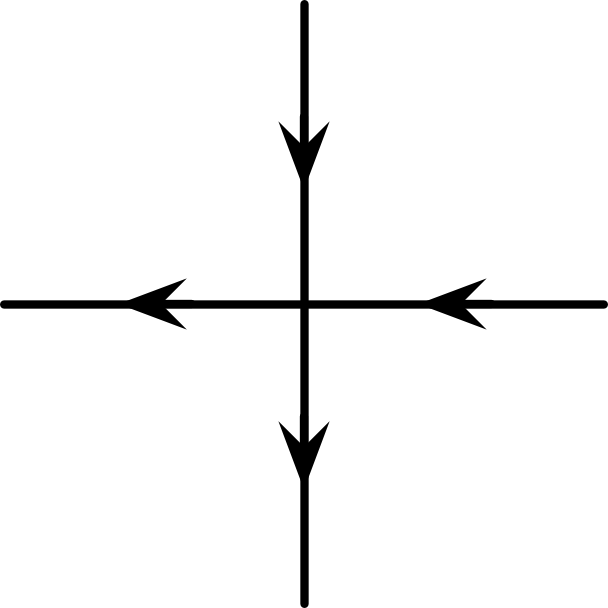}\caption{$2$}
\label{fig:orientations_2}
\end{subfigure}
\begin{subfigure}[b]{0.12\linewidth}
\centering\includegraphics[width=\linewidth]{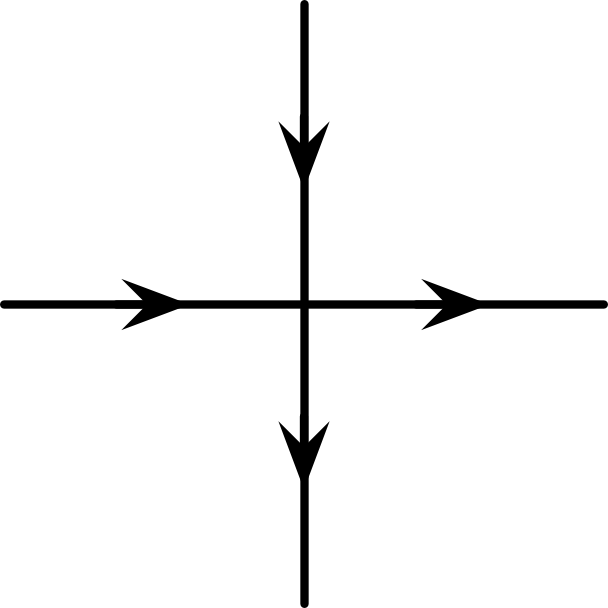}\caption{$3$}
\label{fig:orientations_3}
\end{subfigure}
\begin{subfigure}[b]{0.12\linewidth}
\centering\includegraphics[width=\linewidth]{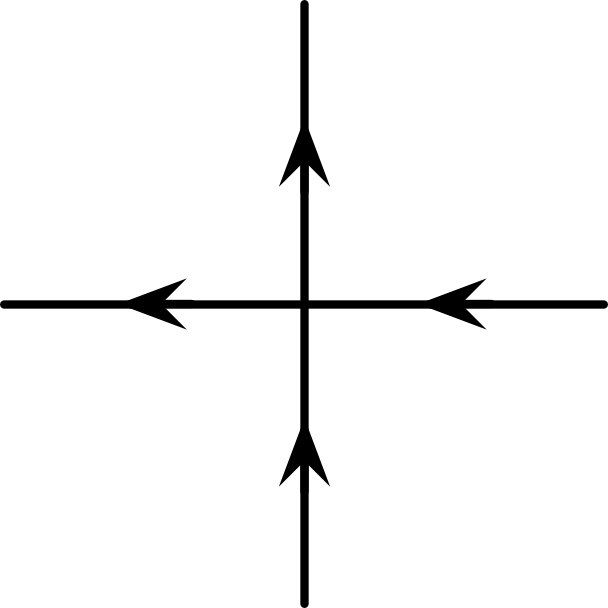}\caption{$4$}
\label{fig:orientations_4}
\end{subfigure}
\begin{subfigure}[b]{0.12\linewidth}
\centering\includegraphics[width=\linewidth]{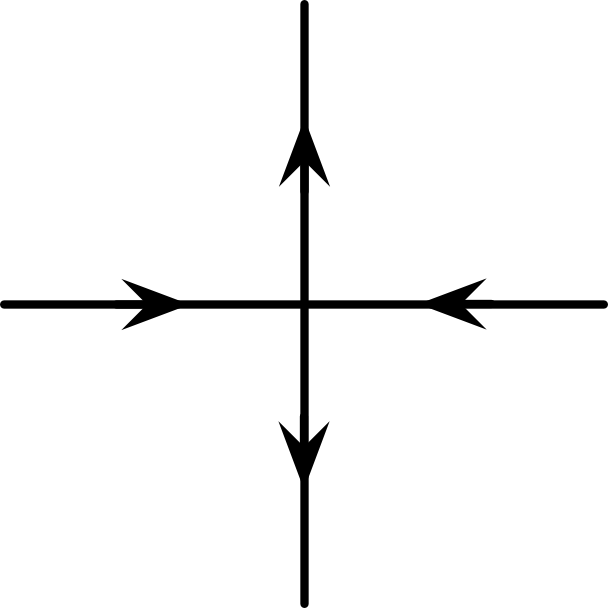}\caption{$5$}
\label{fig:orientations_5}
\end{subfigure}
\begin{subfigure}[b]{0.12\linewidth}
\centering\includegraphics[width=\linewidth]{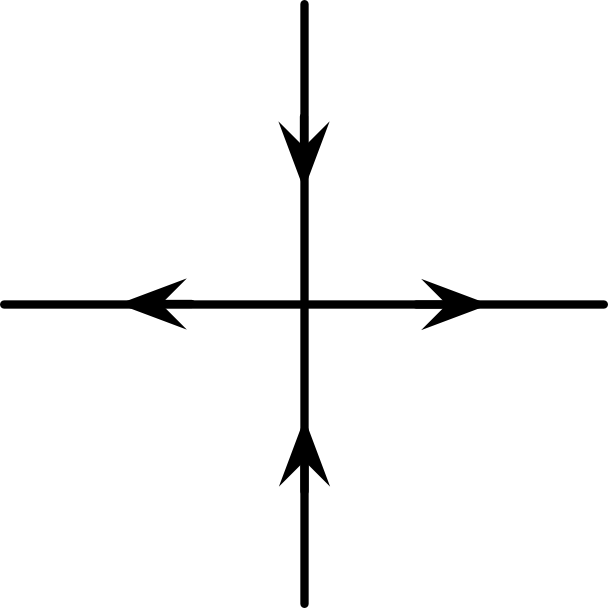}\caption{$6$}
\label{fig:orientations_6}
\end{subfigure}
\begin{subfigure}[b]{0.12\linewidth}
\centering\includegraphics[width=\linewidth]{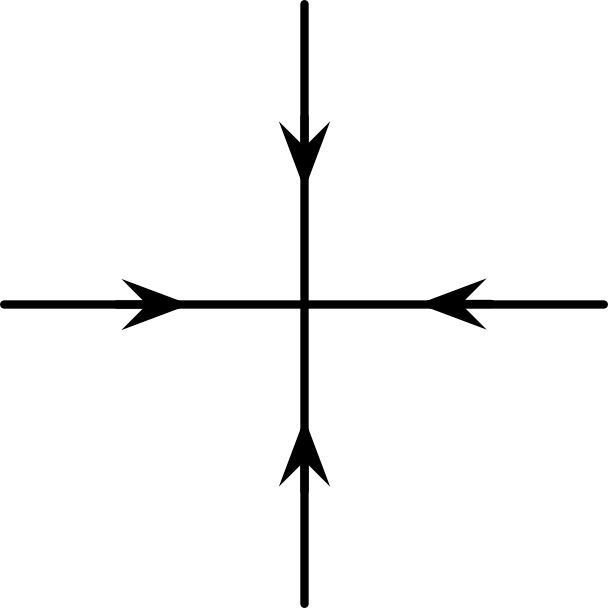}\caption{$7$}
\label{fig:orientations_7}
\end{subfigure}
\begin{subfigure}[b]{0.12\linewidth}
\centering\includegraphics[width=\linewidth]{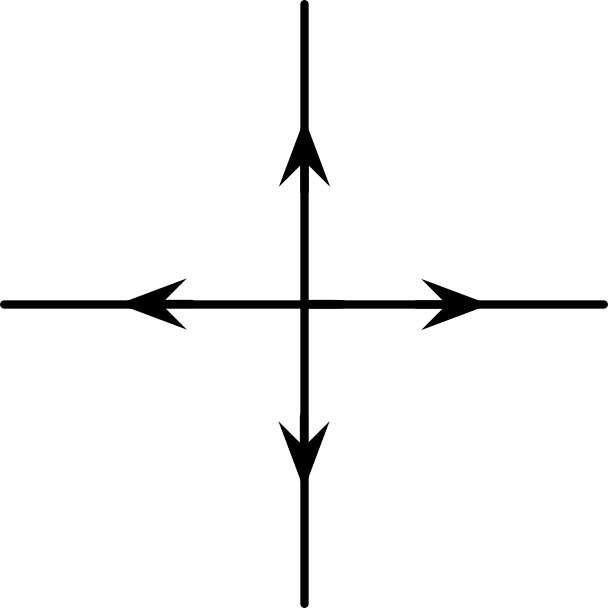}\caption{$8$}
\label{fig:orientations_8}
\end{subfigure}
\caption{Valid configurations of the eight-vertex model.}\label{fig:orientations}
\end{figure}

Classically, the eight-vertex model
is defined by statistical physicists on a square lattice region where each vertex of the lattice is connected by an edge to four nearest neighbors.
There are eight permitted types of local configurations around a vertex---hence the name eight-vertex model (see \figref{fig:orientations}).
In general, the eight configurations 1 to 8 in \figref{fig:orientations}
are associated with eight possible weights $w_1, \ldots, w_8$.
By physical considerations, the total weight of a state remains unchanged
if  all arrows are flipped,
assuming there is no external electric field.
%Thus one may assume 
% without loss of generality that
In this case we write
$w_1 = w_2 = a$, $w_3 = w_4= b$, $w_5 = w_6 = c$, and $w_7 = w_8 = d$.
This complementary invariance is known as \emph{arrow reversal symmetry} or \emph{zero field assumption}.
In this paper, we make this assumption
and further assume that
 $a, b, c, d \ge 0$, as is the case in \emph{classical} physics.
Given a  4-regular graph $G$, we label
four incident edges of each vertex
from 1 to 4.
The \emph{partition function} of the eight-vertex model with parameters
$(a, b, c, d)$ on  $G$
is defined as
\begin{equation}\label{Z-defn}
Z(G; a, b, c, d) = \sum_{\tau \in \mathcal{O}_{\bf e}(G)}a^{n_1 + n_2}b^{n_3 + n_4}c^{n_5 + n_6}d^{n_7 + n_8},
\end{equation}
where $\mathcal{O}_{\bf e}(G)$ is the set of all even orientations of $G$,
and $n_i$ is the number of vertices in type  $i$  in $G$ ($1 \le i \le 8$,
locally depicted as in     
Figure~\ref{fig:orientations}) 
under an even orientation $\tau \in \mathcal{O}_{\bf e}(G)$.

If only six local arrangements 1 to 6 are permitted around a vertex (i.e. $d=0$), then the configurations are \emph{Eulerian orientations} of the underlying 4-regular graph.
This is called the \emph{six-vertex model} which is the antecedent of the eight-vertex model.
The latter was first introduced in 1970 by Sutherland~\cite{doi:10.1063/1.1665111}, and Fan and Wu~\cite{PhysRevB.2.723} as a generalization of the six-vertex model for certain
% JYC: was just "its"
 more desirable properties on the square lattice. However in contrast to the six-vertex model which has been ``exactly solved'' (in the physics sense, a good understanding in the thermodynamic limit on the square lattice) under various parameter settings and external fields~\cite{PhysRev.162.162, PhysRevLett.18.1046, PhysRevLett.19.108, PhysRevLett.19.103, PhysRevB.2.723}, the eight-vertex model was ``exactly solved'' only in the zero-field case~\cite{PhysRevLett.26.832, BAXTER1972193}.
This model is enormously expressive even in the zero-field setting:
its special case when $d=0$, the zero-field six-vertex model, has sub-models such as the ice, KDP, and Rys $F$ models; some other important models such as the dimer and zero-field Ising models can be reduced to it.
%It includes a bewildering array of
%sub-models as special cases, and those that can be
%reduced to it. This includes the zero-field six-vertex model
%which itself has sub-models such as the ice, KDP, and Rys $F$ models,
%and also the dimer and zero-field Ising models as special cases 
%on the square lattice~\cite{Baxter:book}.
Therefore, insight to the eight-vertex model is much sought-after
in statistical physics.

Not until recently did we fully understand the exact computational complexity of the eight-vertex model on 4-regular graphs. In \cite{DBLP:journals/corr/CaiF17}, a complexity dichotomy is given for the eight-vertex model for all eight parameters.
This is studied in the context of a classification program for the complexity of counting problems, where the eight-vertex model serves as important basic cases for Holant problems defined by not necessarily symmetric constraint functions. It is shown that
every setting is either P-time computable (and some are surprising) or \#P-hard.  
However, most cases for P-time tractability 
are due to nontrivial cancellations.
In our setting where $a, b, c, d$ are nonnegative, the problem of computing the partition function of the eight-vertex model is \#P-hard unless: (1) $a = b = c = d$ (this is equivalent to the unweighted case); (2) at least three of $a, b, c, d$ are zero; or (3) two of $a, b, c, d$ are zero and the other two are equal.
In addition, on planar graphs it is also P-time computable for parameter settings $(a, b, c, d)$ with $a^2 + b^2 = c^2 + d^2$, using the \emph{FKT algorithm}. We note that the classification of the exact complexity for the eight-vertex model on planar graphs is still open.

Since exact computation is hard in most cases,
one natural question is what is the approximate complexity of counting and sampling of the eight-vertex model. To our best knowledge, there is only one
previous result in this regard due to Greenberg and Randall.
They showed that on square lattice regions
 a specific Markov chain (which flips the orientations of 
all four edges along a uniformly picked face at each step) is torpidly mixing
 when $d$ is 
large~\cite{Greenberg2010}. It means that when sinks and sources have
large weights, this particular chain cannot be used to approximately sample eight-vertex configurations on the square lattice according to the Gibbs measure.

In this paper we initiate a study toward a classification of the approximate complexity of the eight-vertex model on 4-regular graphs in terms of the parameters. Our results conform to phase transitions in physics.

Here we briefly describe the phenomenon of phase transition of the 
zero-field eight-vertex model (see Baxter's book~\cite{Baxter:book} for more details).
On the square lattice 
in the thermodynamic limit:
\vspace{-2mm}
\setlist[enumerate]{itemsep=-1mm}
\begin{enumerate}[(1)]
\item When  $a > b + c + d$ (called the ferroelectric phase, or FE for short)
any finite region tends to be frozen into one of the two configurations 
where either all 
arrows point up or to the right~(Figure~\ref{fig:orientations_1}), or
alternatively all point down or 
to the left~(Figure~\ref{fig:orientations_2}).
\item Symmetrically when $b > a + c + d$ (also called FE) either all 
arrows point down or to the right~(Figure~\ref{fig:orientations_3}), or
alternatively all point up or to the left~(Figure~\ref{fig:orientations_4}).
\item When
$c > a + b + d$ (AFE: anti-ferroelectric phase)
configurations in Figure~\ref{fig:orientations_5} and Figure~\ref{fig:orientations_6} alternate.
\item When
$d > a + b + c$ (also AFE)
configurations in Figure~\ref{fig:orientations_7} and Figure~\ref{fig:orientations_8} alternate.
\item When $a < b + c + d$, $b < a + c + d$, $c < a + b + d$ and $d < a + b + c$, the system is disordered 
(DO: disordered phase)
 in the sense that all correlations decay to zero
with increasing distance.
\end{enumerate}

For convenience in presenting our theorems and proofs, we adopt the following notations assuming $a, b, c, d \in \mathbb{R}^+$.
\vspace{-2mm}
\setlist[itemize]{itemsep=-1mm}
\begin{itemize}
\item
$\mathcal{F}_{\le^2} := \{(a,b,c,d) \; | \;
a^2 \le b^2 + c^2 + d^2, ~~b^2 \le a^2 + c^2 + d^2,~~ c^2 \le a^2 + b^2 + d^2,~~d^2 \le a^2 + b^2 + c^2\}$;
\item
$\mathcal{F}_> := \{(a,b,c,d) \; | \;
a > b + c + d ~\text{ or }~ b > a + c + d ~\text{ or }~ c > a + b + d ~\text{ or }~ d > a + b + c \text{ where at least two of } a, b, c, d > 0\footnote{If at most one of $a, b, c, d$ is nonzero, computing the partition function is poly-time tractable.}\}$;
\item
$\mathcal{A}_\le := \{(a,b,c,d) \; | \; a+d \le b+c\}$, $\mathcal{B}_\le := \{(a,b,c,d) \; | \; b+d \le a+c\}$, $\mathcal{C}_\le := \{(a,b,c,d) \; | \; c+d \le a+b\}$, $\mathcal{C}_\ge := \{(a,b,c,d) \; | \; c+d \ge a+b\}$, $\mathcal{C}_= := \{(a,b,c,d) \; | \; c+d = a+b\}$.
\end{itemize}

\begin{remark}
We have
$\mathcal{F}_{\le^2} \subset \overline{\mathcal{F}_>}$,
and $\mathcal{A}_\le \bigcap \mathcal{B}_\le \bigcap
 \mathcal{C}_\le  \subset \overline{\mathcal{F}_>}$.
Clearly $\mathcal{C}_=  = \mathcal{C}_\le \bigcap \mathcal{C}_\ge$.
But $\mathcal{A}_\le \bigcap \mathcal{B}_\le \bigcap 
\mathcal{C}_\ge \not\subseteq \overline{\mathcal{F}_>}$. 
\end{remark}

\begin{theorem}\label{thm:main}
There is an
FPRAS\footnote{Suppose $f: \Sigma^* \rightarrow \mathbb{R}$ is a function mapping problem instances to real numbers. A \textit{fully polynomial randomized approximation scheme (FPRAS)} \cite{Karp:1983:MAE:1382437.1382804} for a problem is a randomized algorithm that takes as input an instance $x$ and $\varepsilon > 0$, running in time polynomial in $n$ (the input length) and $\varepsilon^{-1}$, and outputs a number $Y$ (a random variable) such that
\(\operatorname{Pr}\left[(1 - \varepsilon)f(x) \le Y \le (1 + \varepsilon)f(x)\right] \ge \frac{3}{4}.\)}
%%%%%%
for $Z(a, b, c, d)$ if $(a, b, c, d) \in \mathcal{F}_{\le^2} \bigcap \mathcal{A}_\le \bigcap \mathcal{B}_\le \bigcap \mathcal{C}_\le$; there is no FPRAS for $Z(a, b, c, d)$ if $(a, b, c, d) \in \mathcal{F}_>$ unless RP = NP.
%\mathcal{A}_{\le^2} \bigcap \mathcal{B}_{\le^2} \cap \mathcal{C}_{\le^2} \cap \mathcal{D}_{\le^2}
In addition, for planar graphs there is an FPRAS for $Z(a, b, c, d)$ if $(a, b, c, d) \in \mathcal{F}_{\le^2} \bigcap \mathcal{A}_\le \bigcap \mathcal{B}_\le \bigcap \mathcal{C}_\ge$.
\end{theorem}

\begin{remark}
The relationship of these regions denoted by $\mathcal{F}_{\le^2}$, $\mathcal{F}_>$, $\mathcal{A}_\le$, $\mathcal{B}_\le$, $\mathcal{C}_\le$, $\mathcal{C}_\ge$, and $\mathcal{C}_=$ may not
 be easy to visualize, since they reside in 4-dimensional space. 
See \figref{fig:landscape} (where we normalize $d=1$)\footnote{Some 3D renderings of the parameter space can be found at \url{https://skfb.ly/6C9LE} and \url{https://skfb.ly/6C9MS}.}.
The roles of $a$, $b$, $c$, and $d$ are not all symmetric in the eight-vertex model. In particular, $d$ is the weight of sinks and sources and has a special role (e.g. see \cite{Greenberg2010}).
%
%%% JYC is there any strong reason this should be a separate para?
If $(a, b, c, d) \in  \mathcal{A}_\le \bigcap \mathcal{B}_\le \bigcap \mathcal{C}_\le$ then $d \le a, b, c$. So our algorithm 
(for general, i.e., not necessarily planar, graphs)
works only when the weight on sinks and sources is relatively not large.
The restriction of $(a, b, c, d) \in \mathcal{A}_\le \bigcap \mathcal{B}_\le$ is equivalent to $c - d \ge |a - b|$. Therefore, for planar graphs even when sinks and sources have weights larger than the weights of the first four 
configurations in \figref{fig:orientations}, FPRAS can still exist.
\end{remark}

\renewcommand{\thesubfigure}{\alph{subfigure}}
\captionsetup[subfigure]{labelformat=parens}
\begin{figure}[h!]
\centering
\begin{subfigure}[t]{0.48\linewidth}
\centering\includegraphics[width=0.7\linewidth]{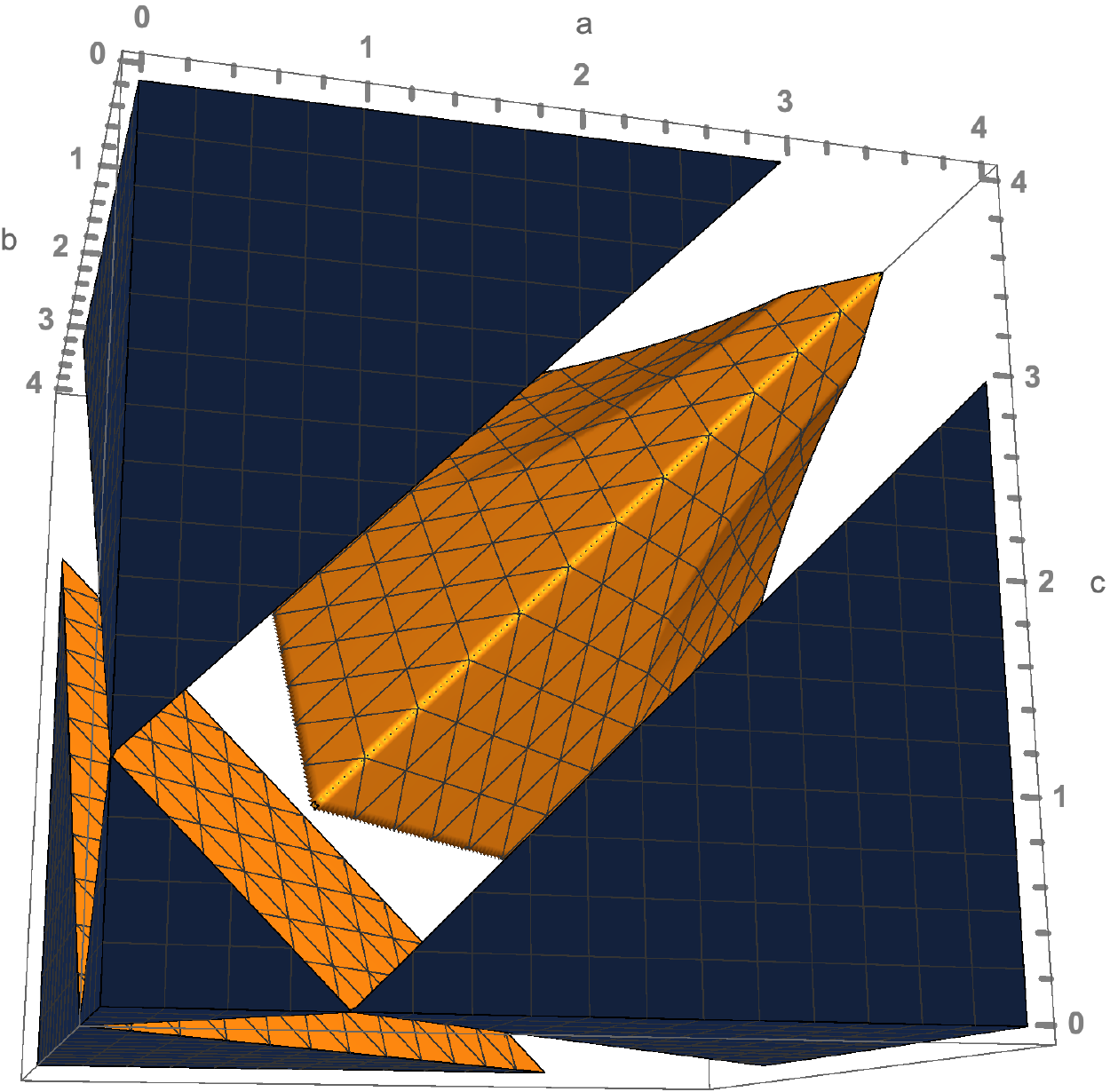}
\caption{Regions of known complexity in the eight-vertex model.
The four corner regions constitute $\mathcal{F}_>$.
The non-corner region depicted is
%one
% can we make it  some diff color?
%is the tractable region  
$\mathcal{F}_{\le^2} \bigcap \mathcal{A}_\le \bigcap \mathcal{B}_\le \bigcap \mathcal{C}_\le$.}
\label{fig:landscape_nonplanar}
\end{subfigure}
\hfill
\begin{subfigure}[t]{0.48\linewidth}
\centering\includegraphics[width=0.7\linewidth]{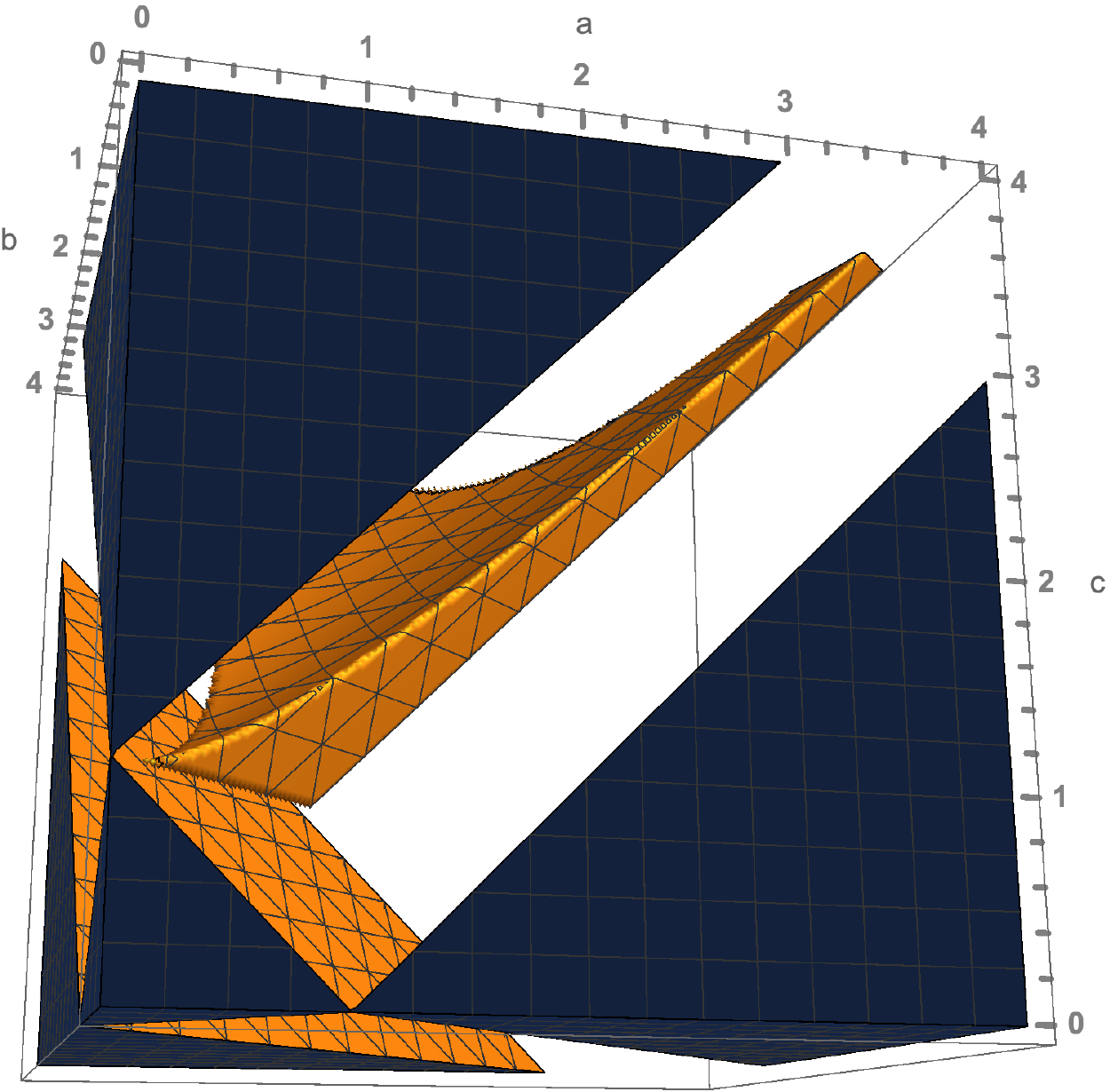}
\caption{An extra region that admits FPRAS on planar graphs.}
\label{fig:landscape_planar}
\end{subfigure}
\caption{}\label{fig:landscape}
\end{figure}

To prove the FPRAS result in \thmref{thm:main}, our most important contribution is a set of \emph{closure properties}. 
We prove these
closure properties  for the eight-vertex model
 in \secref{sec:properties}.
We then use these closure properties to show that a Markov chain designed for the six-vertex model can be adapted to provide our FPRAS.
%Our FPRAS is achieved by characterizing the structural nature, called \emph{closure properties}, of the eight-vertex model and applying the method of \emph{Markov chain Monte Carlo}.
The Markov chain we adapt is the \emph{directed-loop algorithm} which was invented by Rahman and Stillinger~\cite{doi:10.1063/1.1678874} and
is widely used for the six-vertex model (e.g., \cite{YANAGAWA1979329, Barkema98montecarlo, PhysRevE.70.016118}). The state space of our Markov chain for the eight-vertex model consists of even orientations and near-even orientations, which is an extension of the space of valid configurations; the transitions of this algorithm are composed of creating, shifting, and merging of two ``defective'' edges.
%One can interpret this process as the creation of two defects of charge q = ±1 moving from site to site until they recombine. Loops are the natural excitations in vertex models and can be though of as the creation and annihilation of two defects.
A formal description of the directed-loop algorithm is given in \secref{sec:fpras}.

This leads to a  \emph{Markov chain Monte Carlo} approximate counting algorithm
by \emph{sampling}. To prove that this is an FPRAS, we show that
(1) the above Markov chain is rapidly mixing via a \emph{conductance argument}~\cite{doi:10.1137/0218077, Dyer:1991:RPA:102782.102783, Sinclair92improvedbounds, Jerrum-book},
(2) the valid configurations take a non-negligible proportion in the state space, and (3) there is a (not totally obvious) self-reduction (to reduce the computation of the partition function of a graph to that of a ``smaller'' graph)~\cite{JERRUM1986169}.
All three parts depend on the closure properties.
Specifically, we show that when $(a, b, c, d) \in \mathcal{F}_{\le^2}$, the conductance of the Markov chain can be polynomially bounded \emph{if} the ratio of near-even orientations over even orientations can be polynomially bounded; when $(a, b, c, d) \in \mathcal{A}_\le \bigcap \mathcal{B}_\le \bigcap \mathcal{C}_\le$, this ratio is indeed polynomially bounded according to the closure properties. Finally a self-reduction whose success in $\mathcal{A}_\le \bigcap \mathcal{B}_\le \bigcap \mathcal{C}_\le$
 requires an additional closure property.
Therefore, there is an FPRAS in the intersection of $\mathcal{F}_{\le^2}$ and $\mathcal{A}_\le \bigcap \mathcal{B}_\le \bigcap \mathcal{C}_\le$.

The closure properties are keys to our FPRAS.
We use the term \emph{a 4-ary construction} to denote a
4-regular graph $\Gamma$ having four ``dangling'' edges,
and consider all configurations on the edges of $\Gamma$ where every vertex satisfies the even orientation rule and has arrow reversal symmetry. 
We can prove that this $\Gamma$ defines
a \emph{constraint function} of arity 4 that also satisfies
the even orientation rule and has arrow reversal symmetry.
If we imagine the graph $\Gamma$ is shrunken to a single point
 except the 4 dangling edges,
 then
 a 4-ary construction can be viewed as a virtual vertex with  
parameters  $(a', b', c', d')$ in the 
eight-vertex model, for some $a', b', c', d' \ge 0$.

In \thmref{thm:property_1} we show that the set of 4-ary constraint functions in $\mathcal{A}_\le \bigcap \mathcal{B}_\le \bigcap \mathcal{C}_\le$ is closed under 4-ary constructions. This is achieved by inventing a \emph{``quantum decomposition''} of even-orientations.
In~\cite{DBLP:journals/corr/abs-1712-05880}  
 a special case of \thmref{thm:property_1} when $d = 0$ is
proved for the six-vertex model using
a 
%similar 
decomposition of Eulerian orientations.
%With a similar decomposition of Eulerian orientations, the first three authors proved the special case of \thmref{thm:property_1} when $d = 0$ in \cite{DBLP:journals/corr/abs-1712-05880} where the approximate complexity of the six-vertex model on 4-regular graphs is studied.
Given $G = (V, E)$, every Eulerian orientation defines a set of $2^{|V|}$ directed 
%\emph{circuit partitions}
\emph{Eulerian partitions}
by pairing up the four edges around a vertex in one of two ways such that each pair of edges satisfies ``1-in-1-out''. However, such a decomposition does not exist when sinks and sources appear in the eight-vertex model.

In order to overcome this difficulty, we introduce a quantum decomposition where each vertex has a ``signed'' pairing. 
Given an even orientation,
 a \emph{plus pairing} groups the four edges around a vertex into two pairs such that both pairs satisfy ``1-in-1-out'';
% when orientations are embedded; 
a \emph{minus pairing} groups the four edges around a vertex into two pairs such that both pairs independently satisfy either ``2-in'' or ``2-out''.
With weights, this gives rise to a
quantum decomposition of  $3^{|V|}$ \emph{``annotated'' circuit partitions}.
%In this way, we can decompose each even orientation on a 4-regular graph $G = (V, E)$ into $3^{|V|}$ \emph{``annotated'' circuit partitions} so that an orientation on a circuit partition could comply with a corresponding even orientation by changing the direction at each minus pairing.
(Details are in \secref{sec:properties}.)
Although the idea of ``pairings'' and decompositions of Eulerian orientations 
%occasionally find their traces in research
have been used before~\cite{VERGNAS1988367,Jaeger1990,Mihail1996},
% in theoretical computer science and mathematics, our
the idea of a signed pairing and the associated quantum decomposition of
even orientations into annotated circuit partitions is new.
% brand new.
Just as statistical physicists introduce the eight-vertex model on the square lattice for certain desirable 
% its more satisfiable 
properties and better universality over the six-vertex model, in approximate complexity on 4-regular graphs our technique that gives FPRAS for the eight-vertex model extends significantly beyond those for the six-vertex model.

Not only more sophisticated techniques are needed,
% required by studying the eight-vertex model need to be more sophisticated, 
the landscape of approximate complexity for
the eight-vertex model is also richer.
In the six-vertex model we have $d=0$. Then it follows that $\mathcal{F}_{\le^2} \subset \mathcal{A}_\le \bigcap \mathcal{B}_\le \bigcap \mathcal{C}_\le$ which means whenever the conductance of the directed-loop algorithm can be bounded by the ratio of near-even orientations over even orientations, there is an FPRAS. In the eight-vertex model, however, there are parameter settings in $\mathcal{F}_{\le^2}$ where the ratio can be exponentially large. This indicates that the current MCMC method is unable to give FPRAS for the whole region $\mathcal{F}_{\le^2}$, even though
% in which 
there is a nice upper bound for the conductance of this Markov chain.

Moreover, in the eight-vertex model we can give more positive results for planar graphs than for general graphs, unlike in the six-vertex model 
whenever we have an FPRAS for planar graphs we also have one for general graphs
for the same parameters.
For planar graphs, in \thmref{thm:property_2} and \corref{cor:property_4} we show that the extra regions $\mathcal{A}_\le \bigcap \mathcal{B}_\le \bigcap \mathcal{C}_\ge \bigcap \overline{\mathcal{F}_>}$ and $\mathcal{A}_\le \bigcap \mathcal{B}_\le \bigcap \mathcal{C}_=$ also enjoy closure properties.
%It turns out that (only) on planar graphs we have an FPRAS when the parameter
%setting is in the intersection of this region and $\mathcal{F}_{\le^2} \subset \overline{\mathcal{F}_>}$.
%In other words, for planar graphs there is an FPRAS for $Z(a, b, c, d)$ if $(a, b, c, d) \in \mathcal{F}_{\le^2} \bigcap \mathcal{A}_\le \bigcap \mathcal{B}_\le$.
%%%JYC edited slightly...
It turns out that (only) on planar graphs we have an FPRAS when the parameter
setting is in the intersection of  $\mathcal{A}_\le \bigcap \mathcal{B}_\le \bigcap \mathcal{C}_\ge \bigcap \overline{\mathcal{F}_>}$ and $\mathcal{F}_{\le^2}$.
And since $\mathcal{F}_{\le^2} \subset \overline{\mathcal{F}_>}$,
combined with the FPRAS on general graphs,
we get an FPRAS for $\mathcal{F}_{\le^2} \bigcap \mathcal{A}_\le \bigcap \mathcal{B}_\le$ for all planar graphs.
This tolerance of dropping off the requirement $c+d \le a+b$ is in perfect accordance with the special role that ``saddle'' configurations (Figure~\ref{fig:orientations_5} and Figure~\ref{fig:orientations_6}) play on planar graphs.
Although this region is disjoint from the hard regions on general graphs, we find this region is FPRASable on planar graphs but its approximate complexity is unknown for general graphs.
Considering the fact that the exact complexity for the eight-vertex model on planar graphs is not even understood, this is one of the very few cases where research on approximate complexity has advanced beyond that on exact complexity.

The NP-hardness of approximation in FE\&AFE regions is shown by reductions from the problem of computing the maximum cut on a 3-regular graph. For the eight-vertex models not included in the
six-vertex model ($d \neq 0$), both the reduction source and the ``gadgets'' we employ to prove the hardness are substantially different from those that were used in the hardness proof of the six-vertex model~\cite{DBLP:journals/corr/abs-1712-05880}.
We note that the parameter settings in~\cite{Greenberg2010}
where torpid mixing is proved are contained in our  NP-hardness region.

In addition to the complexity result, we show that there is a fundamental difference in the behavior on the two sides separated by the phase transition threshold, in terms of closure properties. In \thmref{thm:property_3}, we show that the set of 4-ary constraint functions lying in the complement of $\mathcal{F}_>$
is closed under 4-ary constructions.
We prove  in this paper that approximation is hard on $\mathcal{F}_>$.
It is not known if the eight-vertex model in the full region of $\overline{\mathcal{F}_>}$ admits FPRAS or not.

%Another important motivation of our work is that 
The eight-vertex model fits into the wider class of Holant problems and 
%actually 
serves as important basic cases for the latter.
Previous results in approximate counting are mostly about spin systems and the present paper, together with \cite{DBLP:journals/corr/abs-1712-05880}, are probably the first fruitful attempts in the Holant literature to make connections to phase transitions. While there is still a gap in the complexity picture for the 
six-vertex and eight-vertex models, we believe the framework set in this paper gives a starting point for studying the approximation complexity of a broader class of counting problems.

%% file: prelim.tex
\documentclass[paper]{subfiles}

Given a 4-regular graph $G = (V, E)$,
the \emph{edge-vertex incidence graph}  $G' = (U_E, U_V, E')$
is a bipartite graph where
$(u_e, u_v) \in U_E \times U_V$ is an edge in $E'$ iff 
$e \in E$ in $G$  is incident to $v \in V$.
We model an orientation ($w \rightarrow v$)
 on an edge $e = \left\{w, v\right\} \in E$
 from $w$ into $v$ in $G$ by assigning
 $1$  to $(u_e, u_w) \in E'$  and  $0$ to $(u_e, u_v) \in E'$
in $G'$.
A configuration of the eight-vertex model on $G$
is  an  \emph{edge 2-coloring} on $G'$,
namely $\sigma: E' \rightarrow \{0, 1\}$,
where for  each  $u_e \in U_E$ its two incident edges are
assigned 01 or 10, and  for  each $u_v \in U_V$ the  sum of
 values $\sum_{i=1}^4 \sigma(e_i) \equiv 0  \pmod 2$,
over the four incident edges of $u_v$.
Thus 
we model the even orientation rule of $G$ on all $v \in V$ by requiring ``two-0-two-1/four-0/four-1'' locally at
each vertex $u_v \in U_V$.

The ``one-0-one-1'' requirement on the two edges incident to a vertex in $U_E$ is a binary {\sc Disequality} constraint, denoted by $(\neq_2)$.
The values of a 4-ary \emph{constraint function} $f$  can be listed in a matrix $M(f) = \left[\begin{smallmatrix} f_{0000} & f_{0010} & f_{0001} & f_{0011} \\ f_{0100} & f_{0110} & f_{0101} & f_{0111} \\ f_{1000} & f_{1010} & f_{1001} & f_{1011} \\ f_{1100} & f_{1110} & f_{1101} & f_{1111}\end{smallmatrix}\right]$,
called the \emph{constraint matrix} of $f$. For the eight-vertex model 
satisfying the even orientation rule and arrow reversal symmetry, the constraint function $f$ at every vertex $v \in U_V$ in $G'$ 
has the form $M(f) = \left[\begin{smallmatrix} d & 0 & 0 & a \\ 0 & b & c & 0 \\ 0 & c & b & 0 \\ a & 0 & 0 & d \end{smallmatrix}\right]$, if we locally index the left, down, right, and up edges incident to $v$ by 1, 2, 3, and 4, respectively according to \figref{fig:orientations}.
Thus computing the partition function $Z(G; a, b, c, d)$ is equivalent to evaluating
(the Holant sum in the framework for Holant  problems)
\[Z'(G'; f) := \sum_{\sigma:E'\rightarrow\left\{0,1\right\}}\prod_{u\in U_E}(\neq_2)\left(\sigma |_{E'(u)}\right) \prod_{u\in U_V}f\left(\sigma |_{E'(u)}\right).\]
where $E'(u)$ denotes the incident edges of $u \in U_E \cup U_V$.

When every vertex in $G$ 
has the same constraint function $f$ with $M(f) = \left[\begin{smallmatrix} d & & & a \\ & b & c & \\ & c & b & \\ a & & & d \end{smallmatrix}\right]$,
we write the partition function $Z(a, b, c, d)$ as $Z(f)$,
 and denote by $Z(\mathcal{F})$ when
each vertex is assigned some constraint function from a set $\mathcal{F}$
consisting of constraint functions of this form.

%% file: properties.tex
\documentclass[paper]{subfiles}
\begin{theorem}\label{thm:property_3}
The set of constraint functions in $\overline{\mathcal{F}_>}$ 
is closed under 4-ary constructions,
i.e., the constraint function of any 4-ary construction
using constraint functions from the set  $\overline{\mathcal{F}_>}$
also belongs to the same set.
\end{theorem}

\begin{theorem}\label{thm:property_1}
The set of constraint functions in $\mathcal{A}_\le \bigcap \mathcal{B}_\le \bigcap \mathcal{C}_\le$ is closed under 4-ary constructions.
%i.e., the constraint function of any 4-ary construction
%using constraint functions from the set 
%$\mathcal{A}_\le \bigcap \mathcal{B}_\le \bigcap \mathcal{C}_\le$
%also belongs to the same set.
\end{theorem}

\begin{theorem}\label{thm:property_2}
The set of constraint functions in $\mathcal{A}_\le \bigcap \mathcal{B}_\le \bigcap \mathcal{C}_\ge \bigcap \overline{\mathcal{F}_>}$ is closed under 4-ary \emph{plane} constructions.
\end{theorem}

%\begin{theorem}\label{thm:property_4}
\begin{corollary}\label{cor:property_4}
The set of constraint functions in $\mathcal{A}_\le \bigcap \mathcal{B}_\le \bigcap \mathcal{C}_=$ is closed under 4-ary \emph{plane} constructions.
\end{corollary}
%\end{theorem}

%\begin{remark}
%$\mathcal{A}_\le \bigcap \mathcal{B}_\le \bigcap \mathcal{C}_\le \subset \overline{\mathcal{F}_>}$, $\mathcal{A}_\le \bigcap \mathcal{B}_\le \bigcap \mathcal{C}_\ge \bigcap \overline{\mathcal{F}_>} \subset \overline{\mathcal{F}_>}$, and $\mathcal{A}_\le \bigcap \mathcal{B}_\le \bigcap \mathcal{C}_= \subset \overline{\mathcal{F}_>}$.
%\end{remark}

In order to prove the above closure properties, we introduce a quantum decomposition for the eight-vertex model, in which every even orientation of a 4-regular graph $G = (V, E)$ is a 
``superposition'' of $3^{|V|}$ annotated circuit partitions
 (to be defined shortly).

Let $v$ be a vertex of $G$, and $e_1, e_2, e_3, e_4$ the four
labeled edges incident to $v$.
A \emph{pairing} $\varrho$ at $v$ is a partition of $\{e_1, e_2, e_3, e_4\}$ into two pairs. There are exactly three distinct pairings at $v$~(\figref{fig:pairings}) which we denote by three special symbols: $\sub{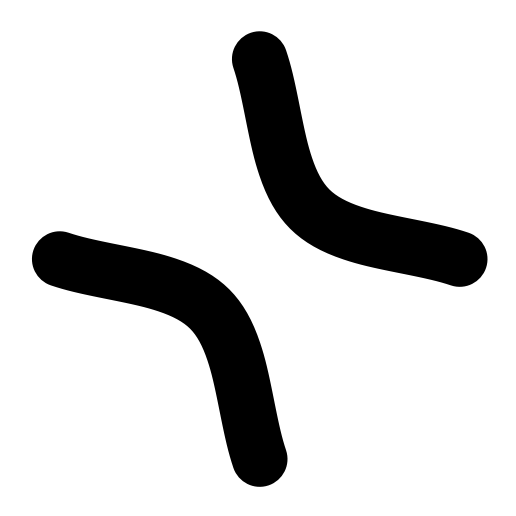}, \sub{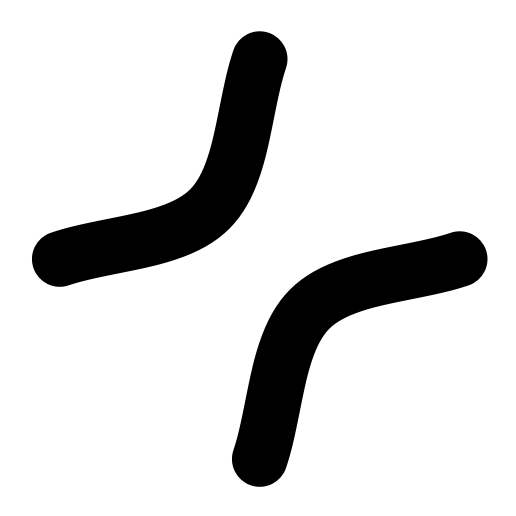}, \sub{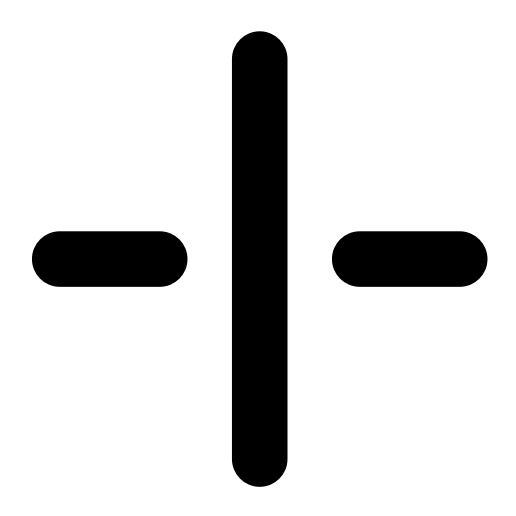}$, respectively.
A \emph{circuit partition} of a graph $G$ is a partition of the edges of $G$ into edge-disjoint circuits (in such a circuit
 vertices may repeat but edges may not).
It is in 1-1 correspondence with a family of pairings $\varphi = \{\varrho_v\}_{v \in V}$, where $\varrho_v \in \{\sub{1}, \sub{2}, \sub{3}\}$ is a pairing at $v${\textemdash}once the pairing at each vertex is fixed, then the two edges paired together at each vertex is also adjacent in the same circuit.

\begin{figure}[h!]
\centering
\begin{subfigure}[b]{0.3\linewidth}
\centering\includegraphics[width=0.5\linewidth]{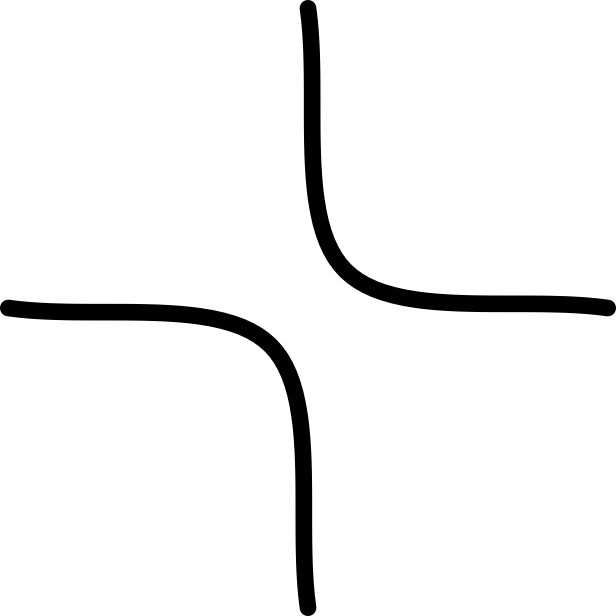}\caption{$\sub{1}$}
\end{subfigure}
\begin{subfigure}[b]{0.3\linewidth}
\centering\includegraphics[width=0.5\linewidth]{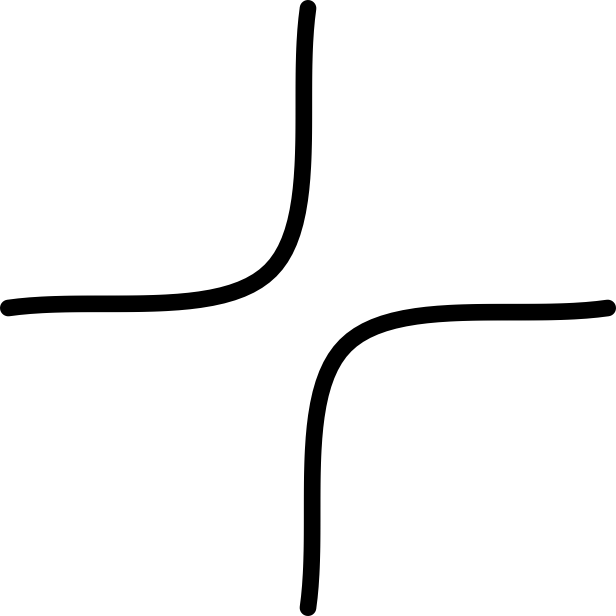}\caption{$\sub{2}$}
\end{subfigure}
\begin{subfigure}[b]{0.3\linewidth}
\centering\includegraphics[width=0.5\linewidth]{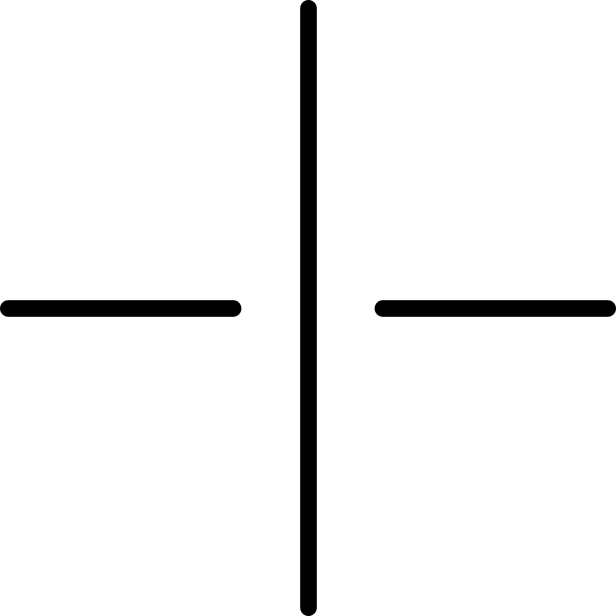}\caption{$\sub{3}$}
\end{subfigure}
\caption{(Unsigned) pairings at a degree 4 vertex.}\label{fig:pairings}
\end{figure}

%In this way, we can decompose each even orientation on a 4-regular graph $G = (V, E)$ into $3^{|V|}$ \emph{``annotated'' circuit partitions} so that an orientation on a circuit partition could comply with a corresponding even orientation by changing the direction at each minus pairing.

A \emph{signed pairing} $\boldsymbol{\varrho}_v$ at $v$ is a pairing with a sign,
either \emph{plus} ($+$) or \emph{minus} ($-$). In other words, it is an element in $\{\sub{1}, \sub{2}, \sub{3}\} \times \{+,-\}$.
We denote a signed pairing by $\varrho_+$ or $\varrho_-$ if the pairing is $\varrho$ and the sign is plus or minus, respectively.
An \emph{annotated circuit partition} of $G$, or
{\sl acp} for short, is  a circuit partition of  $G$ together with
a  map $V \rightarrow \{+,-\}$ such that
along every circuit one encounters an even number of $-$ 
(a repeat vertex with $-$ counts twice on the circuit).
Thus, it is in  1-1 correspondence with
 a family of signed pairings for all $v \in V$, 
with the restriction that there is an \emph{even} number of $-$ along each circuit.
Each circuit $C$ in an {\sl acp} has exactly two directed \emph{states}---starting at an arbitrary edge in $C$ with one of the two orientations on this edge, one can
uniquely orient every edge in $C$ such that for every vertex $v$
on $C$, two edges incident at $v$ paired up by $+$
 have consistent orientations  at $v$
(i.e., they form ``1-in-1-out'' at $v$), 
whereas two edges paired up by $-$  
have contrary orientations at $v$
(i.e., they form ``2-in'' or ``2-out'' at $v$).
%contradicting
%%% JYC I still don't like the term contradicting
%%% how about contrary?
%( orientations.
These two directed states of $C$
%  the circuit 
are well-defined because cyclically
the direction of edges along $C$
%the circuit 
changes an even
number of times, precisely at the minus signs.
A \emph{directed
annotated circuit partition} ({\sl dacp}) is an {\sl acp}
with each circuit in a directed state.
If an {\sl acp} has $k$ circuits, then it
defines $2^k$ {\sl dacp}'s.

Next we describe an association
between even orientations and {\sl acp}'s as well as {\sl dacp}'s.
%to directed annotated circuit partitions.
%Recall that at a vertex $v$ in $G$ there are three
%possible pairings: $\sub{1}, \sub{2}, \sub{3}$.
Given an even orientation $\tau$ of $G$,
every local configuration of  $\tau$ at a vertex
defines exactly three signed pairings at this vertex
according to \tabref{tab:correspondence}.
 %each of the three pairings at $v$ has a sign, plus or minus depending on the local configuration at $v$ (same up to arrow reversal symmetry) according to \tabref{tab:correspondence}.
%If is easy to verify that the 
%local orientation  of the edges by  $\tau$ around a vertex
%is such that two paired-up edges have consistent (or contrary) 
%orientations if and only if it has the sign $+$ (or $-$).
Note that, given  $\tau$ and a pairing at  a vertex $v$,
the two pairs have {either} \emph{both} consistent {or} \emph{both} contrary
orientations.
Thus the same sign, $+$ or $-$,  works for both pairs, although
this depends on the pairing at $v$.

\begin{table}[h!]
\newcolumntype{C}{ >{\centering\arraybackslash} m{3cm} }
\newcolumntype{W}{ >{\centering\arraybackslash} m{1cm} }
\newcolumntype{S}{ >{\centering\arraybackslash} m{1cm} }
  \begin{center}
    \caption{Map from eight local configurations to signed pairings.}
    \label{tab:correspondence}
    \begin{tabular}{@{}CcccSSS@{}}%{m{5em} | m{5em} | m{5em} | m{5em} | m{5em}}%
    \toprule
      \textbf{Configurations} & \phantom{abc} & \textbf{Weight} & \phantom{abc} & \multicolumn{3}{c}{\textbf{Sign}}\\
      \cmidrule{5-7}
      & & & & \includegraphics[width=0.06\textwidth]{phi_1} & \includegraphics[width=0.06\textwidth]{phi_2} & \includegraphics[width=0.06\textwidth]{phi_3}\\
      \midrule
      \includegraphics[width=0.075\textwidth]{epsilon_1} \includegraphics[width=0.075\textwidth]{epsilon_2} && $a$ && - & + & + \\
      \includegraphics[width=0.075\textwidth]{epsilon_3} \includegraphics[width=0.075\textwidth]{epsilon_4} && $b$ && + & - & + \\
      \includegraphics[width=0.075\textwidth]{epsilon_5} \includegraphics[width=0.075\textwidth]{epsilon_6} && $c$ && + & + & - \\
      \includegraphics[width=0.075\textwidth]{epsilon_7} \includegraphics[width=0.075\textwidth]{epsilon_8} && $d$ && - & - & - \\
    \bottomrule
    \end{tabular}
  \end{center}
\end{table}

In this way, every even orientation $\tau$ 
defines $3^{|V|}$ {\sl acp}'s, denoted by $\Phi(\tau)$.
See \tabref{tab:decomposition_a} and \tabref{tab:decomposition_c} for two examples.
Moreover, for any {\sl acp} $\varphi \in \Phi(\tau)$, every circuit in $\varphi$ 
is in one of the two well-defined directed states
under the  orientation $\tau$.
Thus each even orientation $\tau$ defines $3^{|V|}$ {\sl dacp}'s.

\begin{table}[h!]
\newcolumntype{O}{ >{\centering\arraybackslash} m{5cm} }
\newcolumntype{A}{ >{\centering\arraybackslash} m{2.2cm} }
  \begin{center}
    \caption{An even orientation and its quantum decomposition into {\sl acp}'s.}
    \label{tab:decomposition_a}
    \begin{tabular}{@{}O@{\phantom{abcdefg}}A@{\phantom{ab}}A@{\phantom{ab}}A@{}}
    \toprule
      $\tau$ & \multicolumn{3}{c}{$\Phi(\tau)$}\\
      \midrule
      \\[-0.5em]
      \multirow{3}{*}{\includegraphics[width=0.3\textwidth]{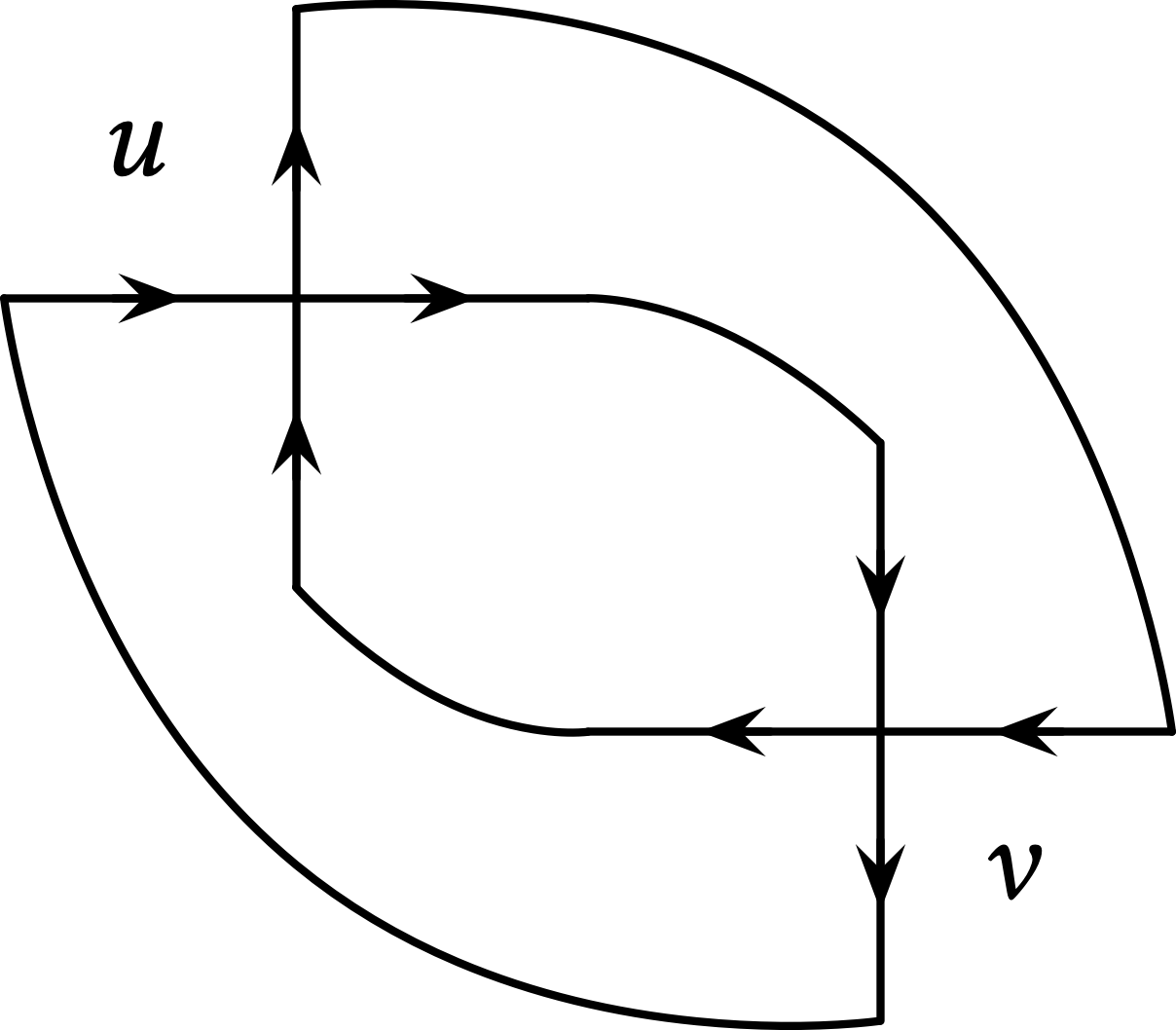}}
      & \includegraphics[width=0.13\textwidth]{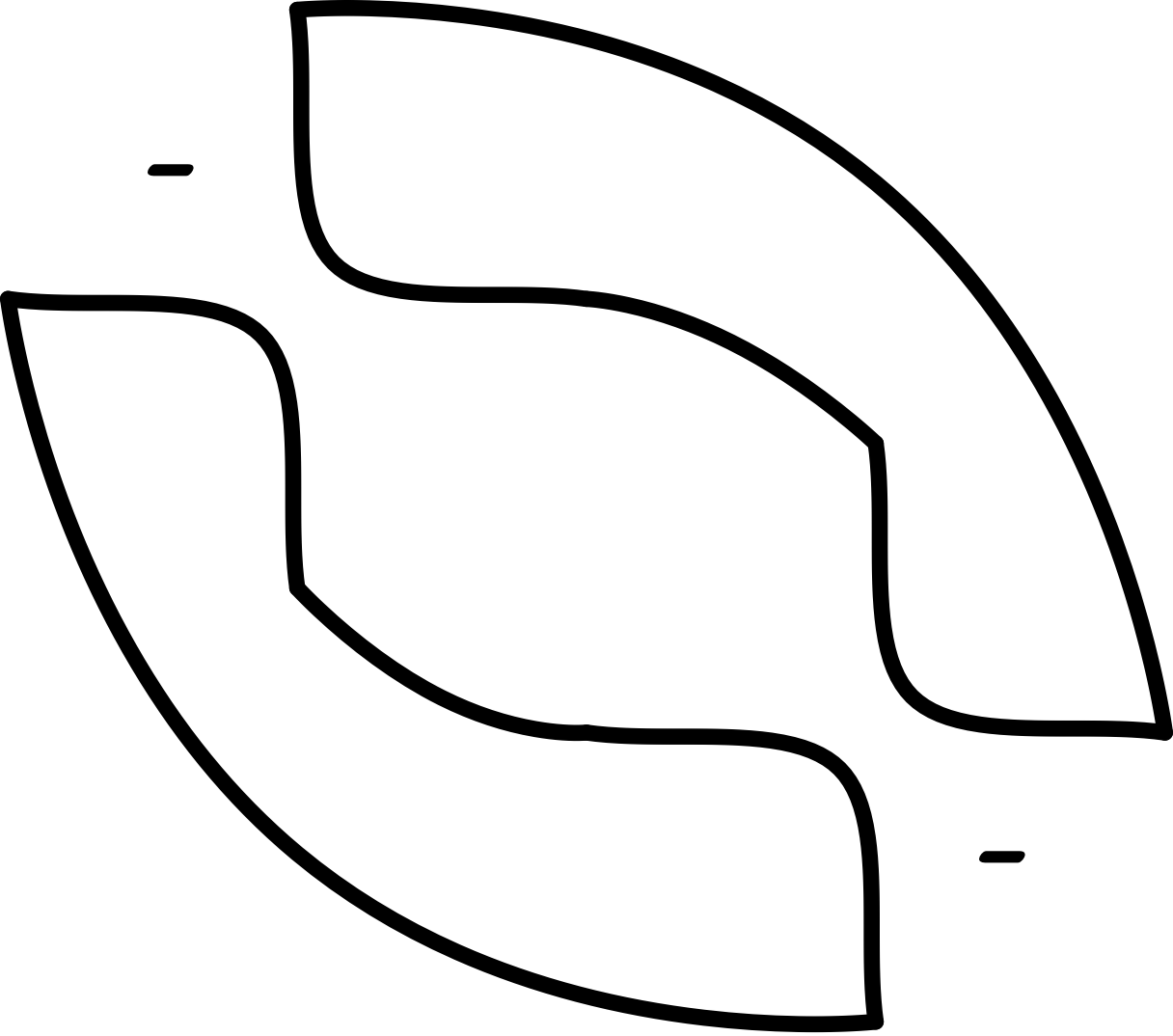} &
             \includegraphics[width=0.13\textwidth]{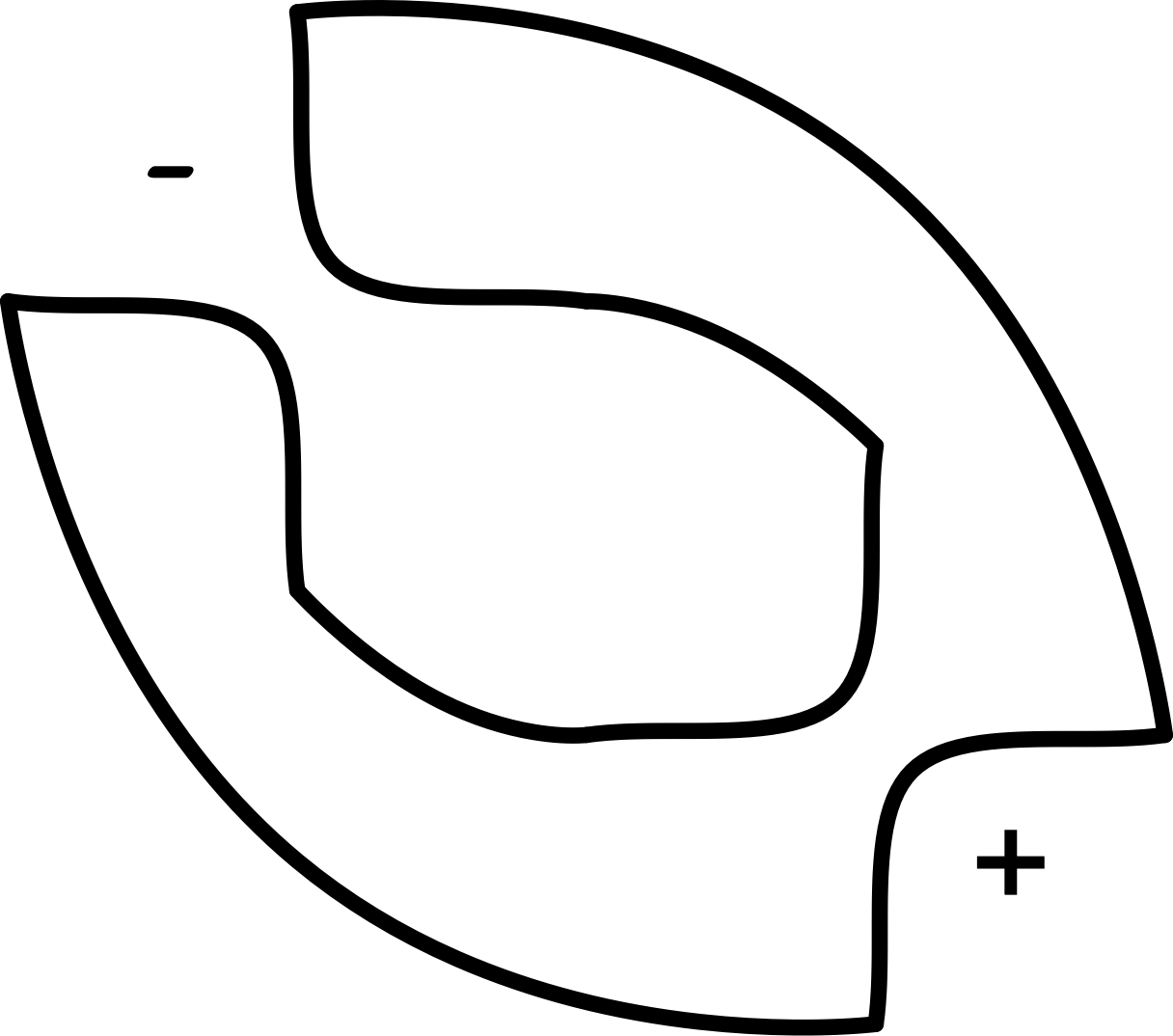} &
       \includegraphics[width=0.13\textwidth]{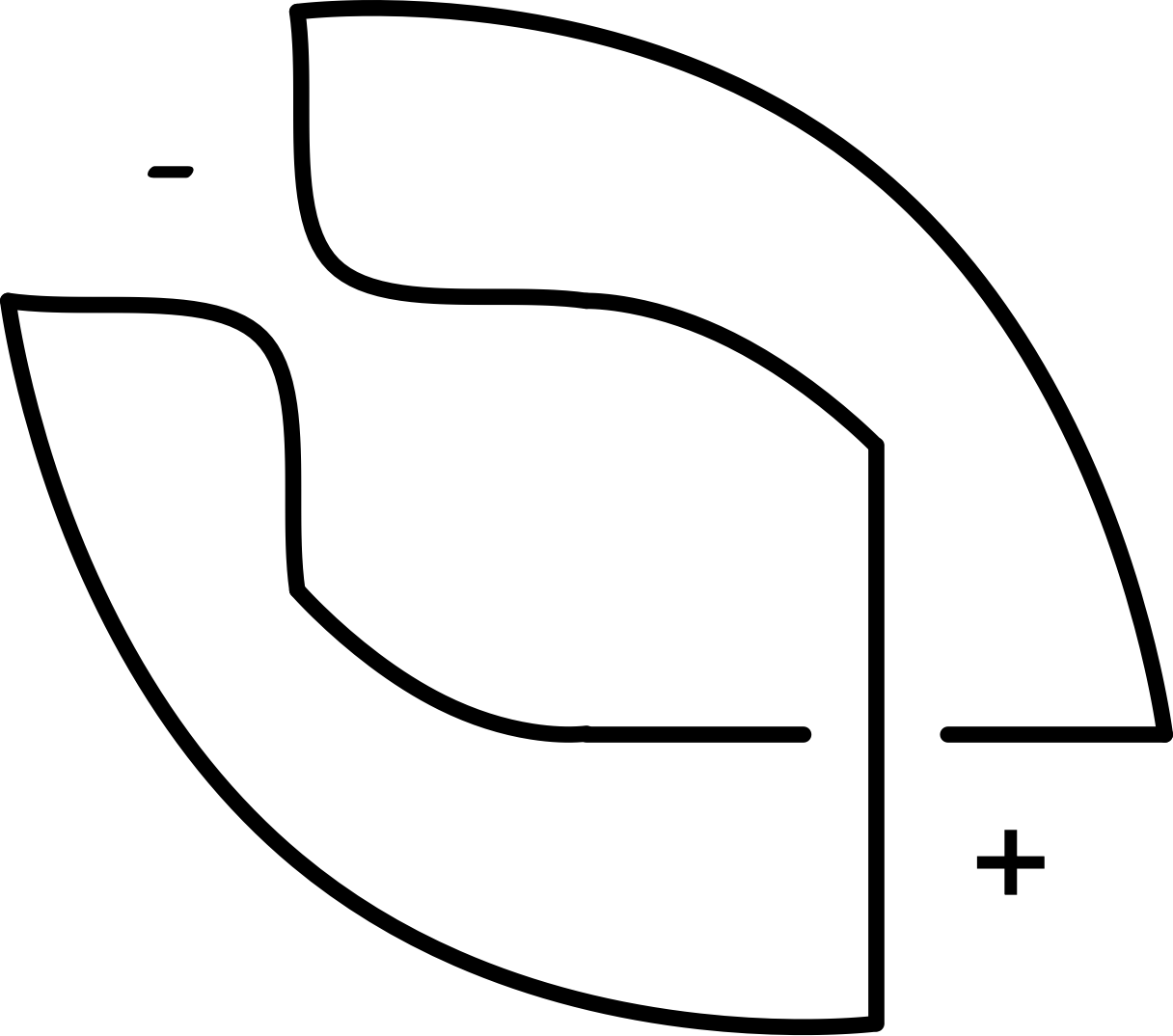} \\
       \\[-0.5em]
      & \includegraphics[width=0.13\textwidth]{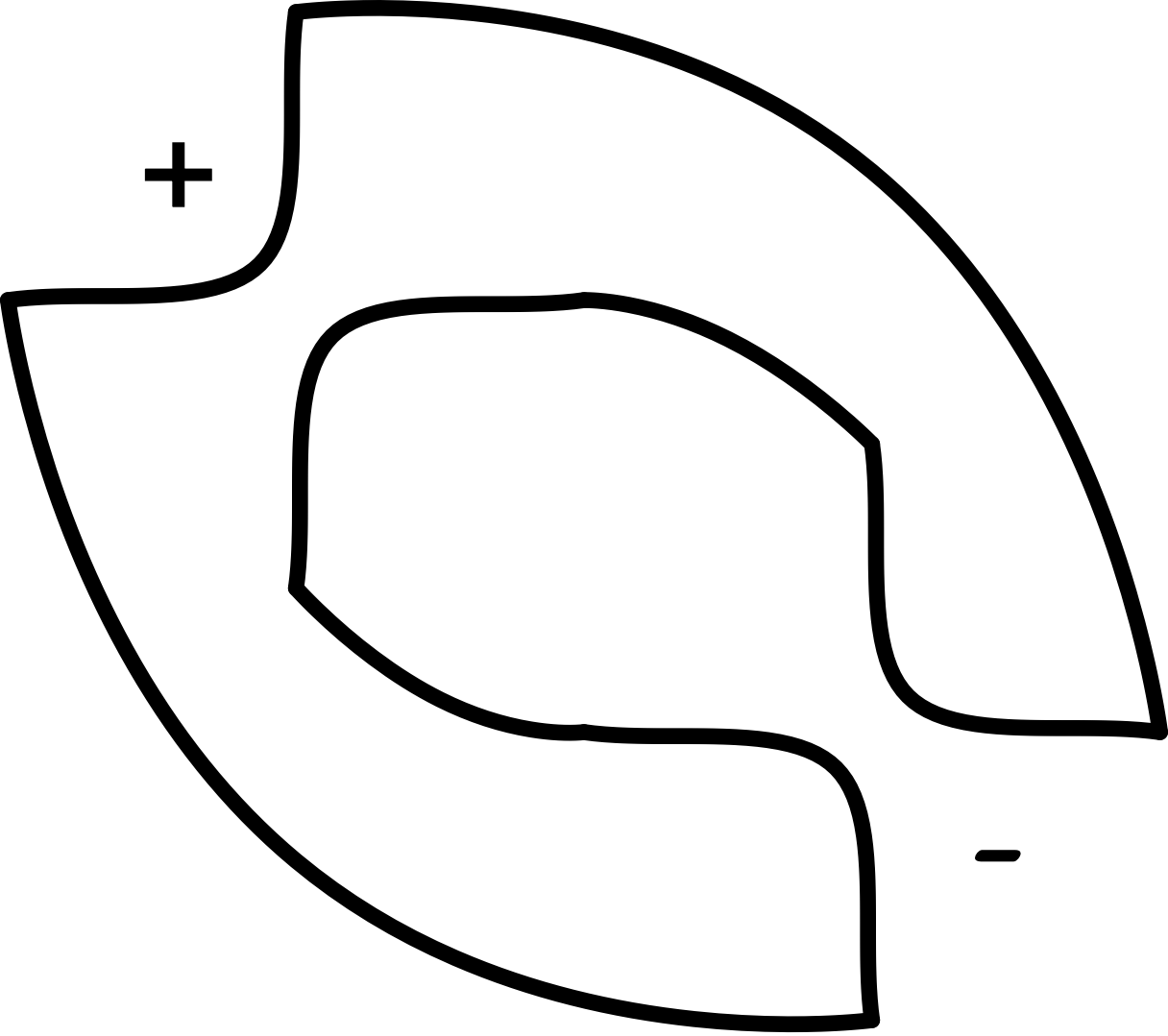} &
       \includegraphics[width=0.13\textwidth]{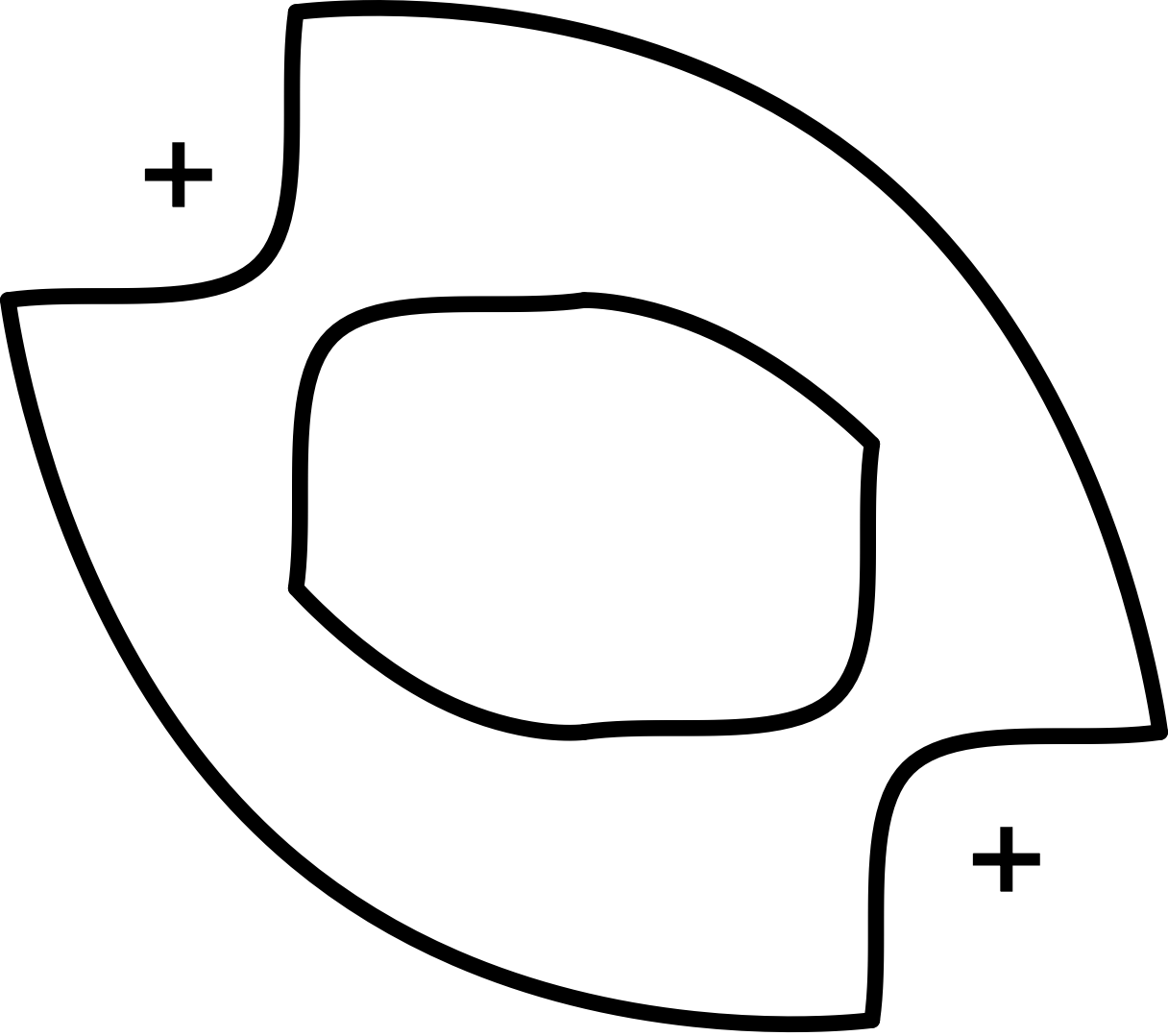} &
       \includegraphics[width=0.13\textwidth]{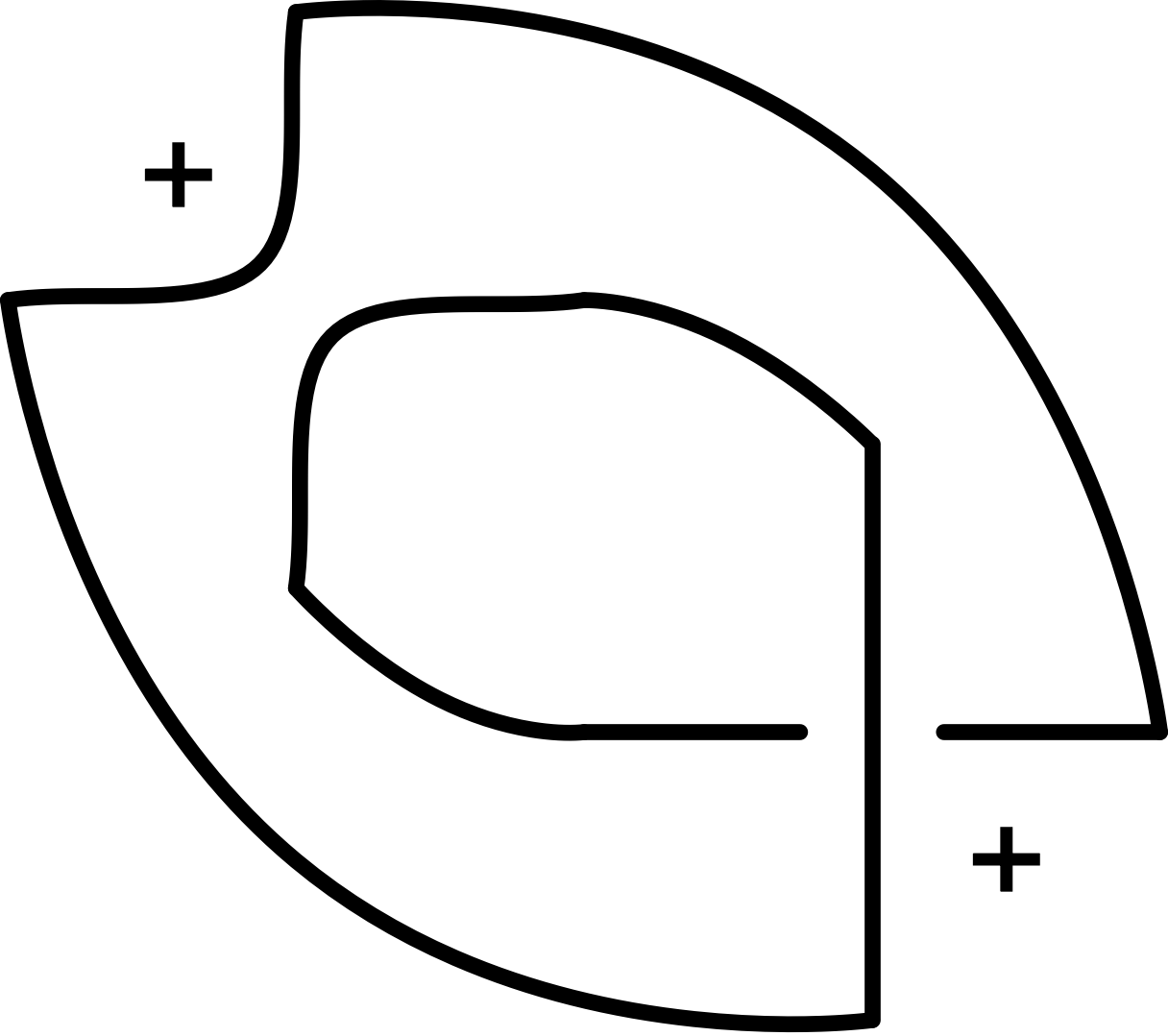} \\
       \\[-0.5em]
      & \includegraphics[width=0.13\textwidth]{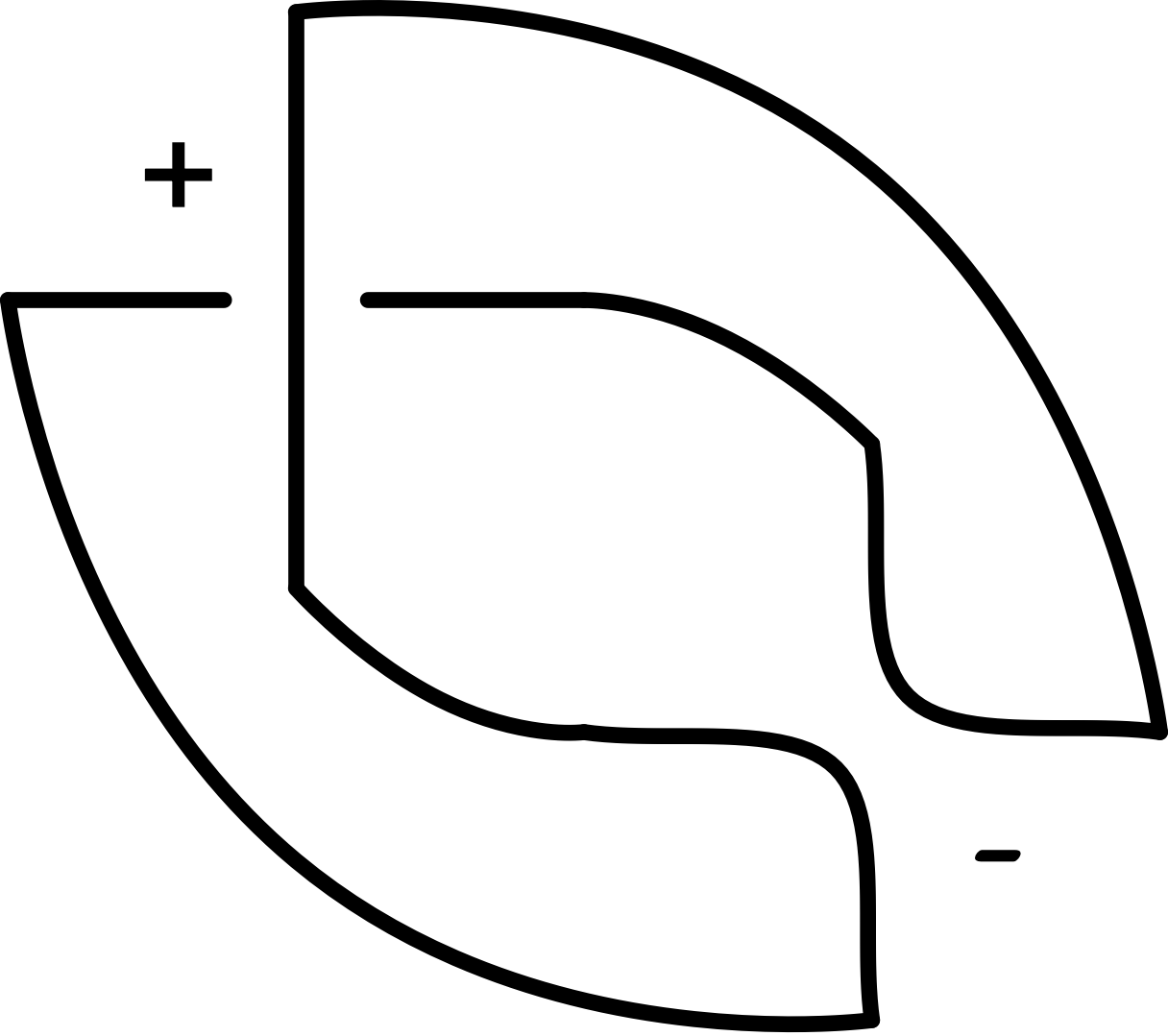} &
       \includegraphics[width=0.13\textwidth]{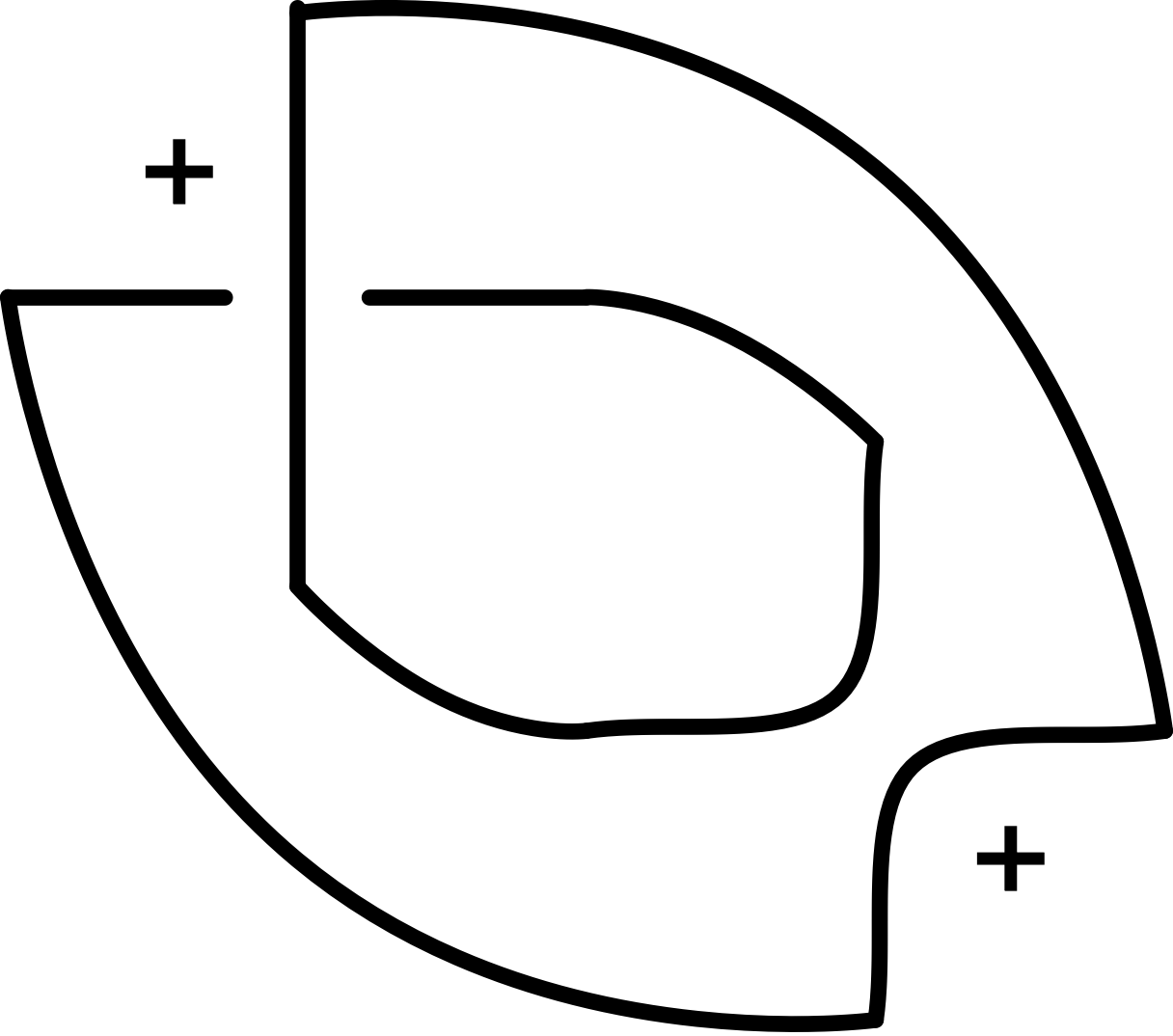} &
       \includegraphics[width=0.13\textwidth]{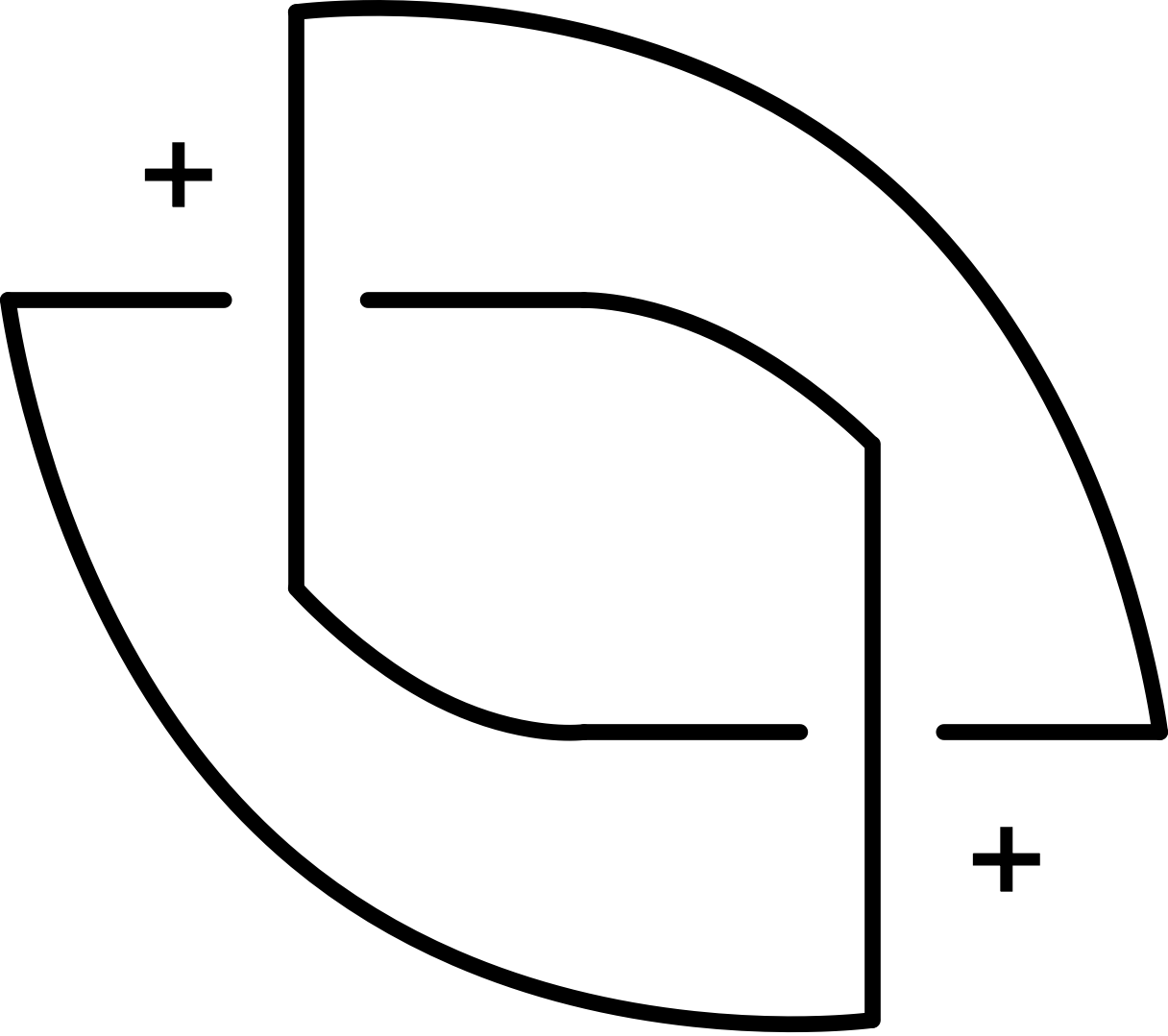} \\
       \\[-0.5em]
    \bottomrule
    \end{tabular}
  \end{center}
\end{table}

\begin{table}[h!]
\newcolumntype{O}{ >{\centering\arraybackslash} m{5cm} }
\newcolumntype{A}{ >{\centering\arraybackslash} m{2.2cm} }
  \begin{center}
    \caption{Another even orientation and its quantum decomposition into {\sl acp}'s.}
    \label{tab:decomposition_c}
    \begin{tabular}{@{}O@{\phantom{abcdefg}}A@{\phantom{ab}}A@{\phantom{ab}}A@{}}
    \toprule
      $\tau$ & \multicolumn{3}{c}{$\Phi(\tau)$}\\
      \midrule
      \\[-0.5em]
      \multirow{3}{*}{\includegraphics[width=0.3\textwidth]{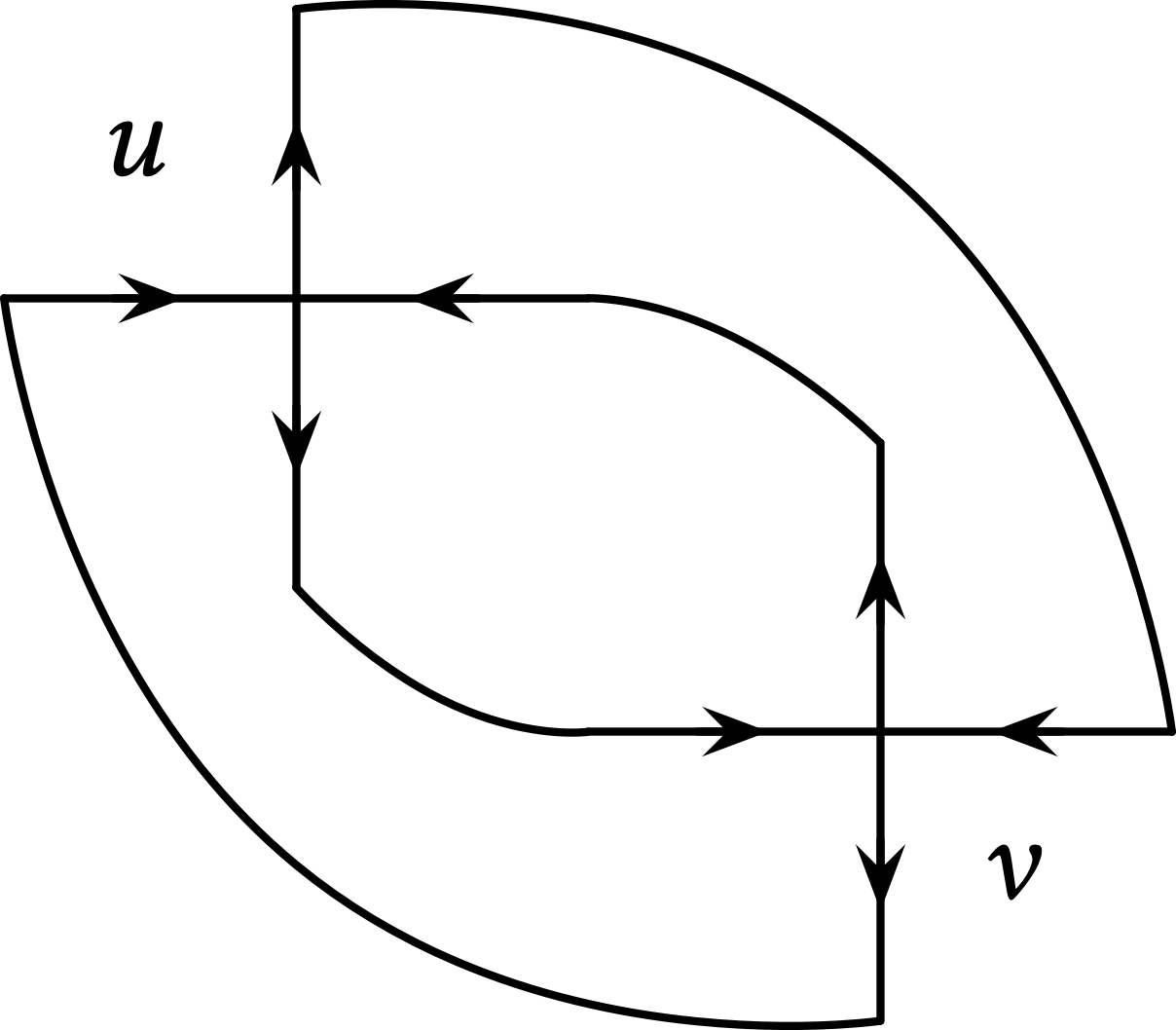}}
      & \includegraphics[width=0.13\textwidth]{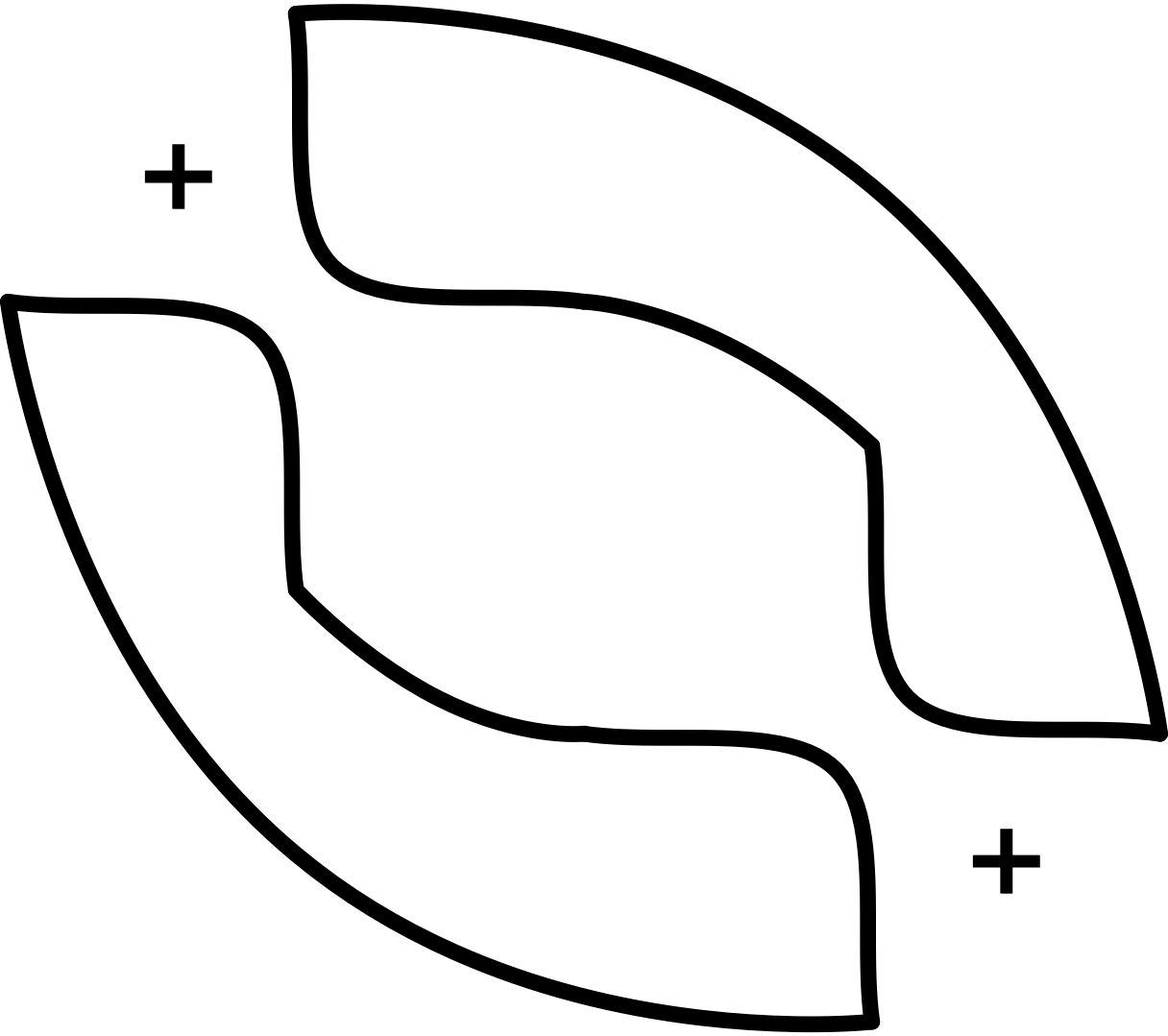} &
             \includegraphics[width=0.13\textwidth]{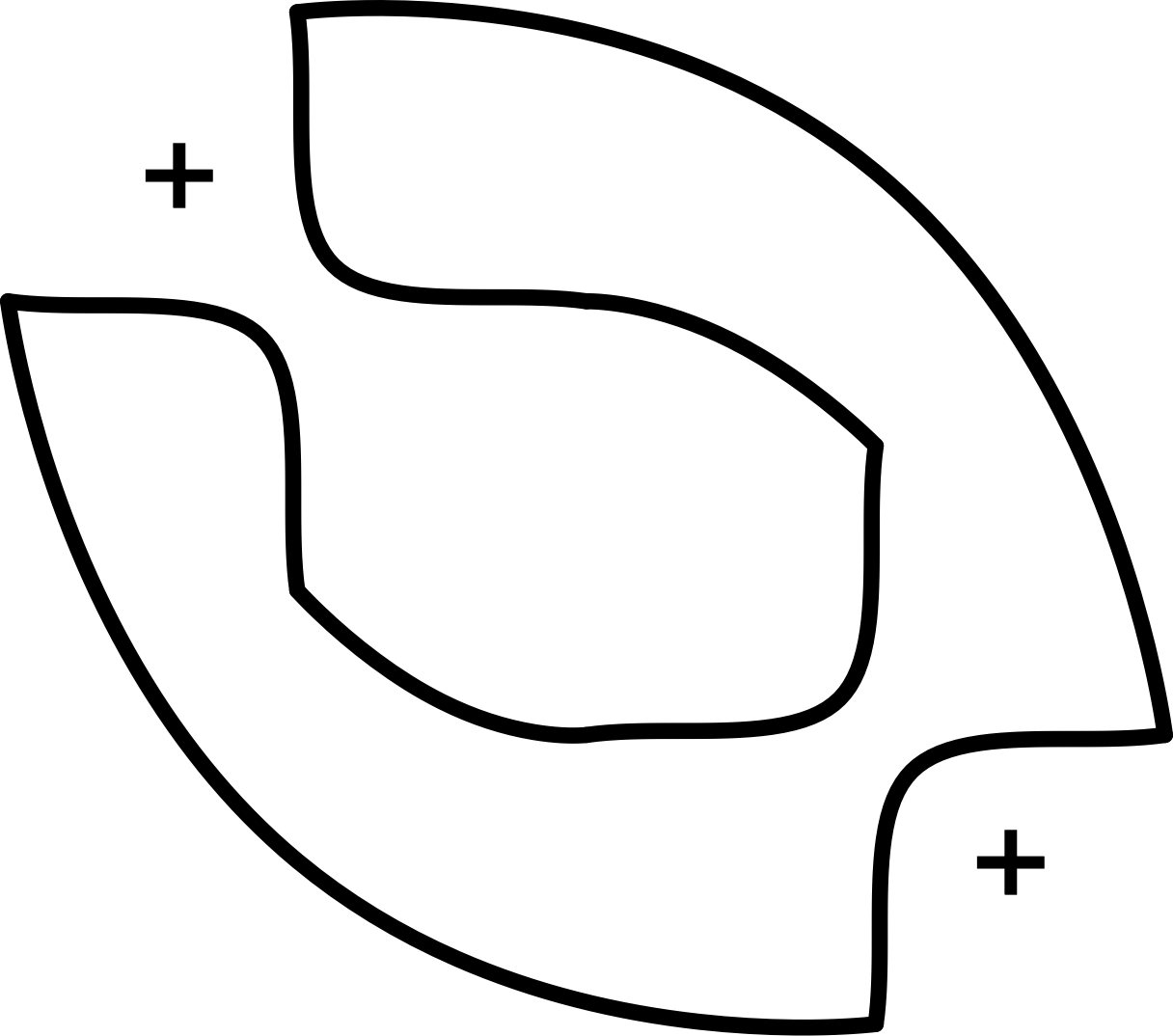} &
       \includegraphics[width=0.13\textwidth]{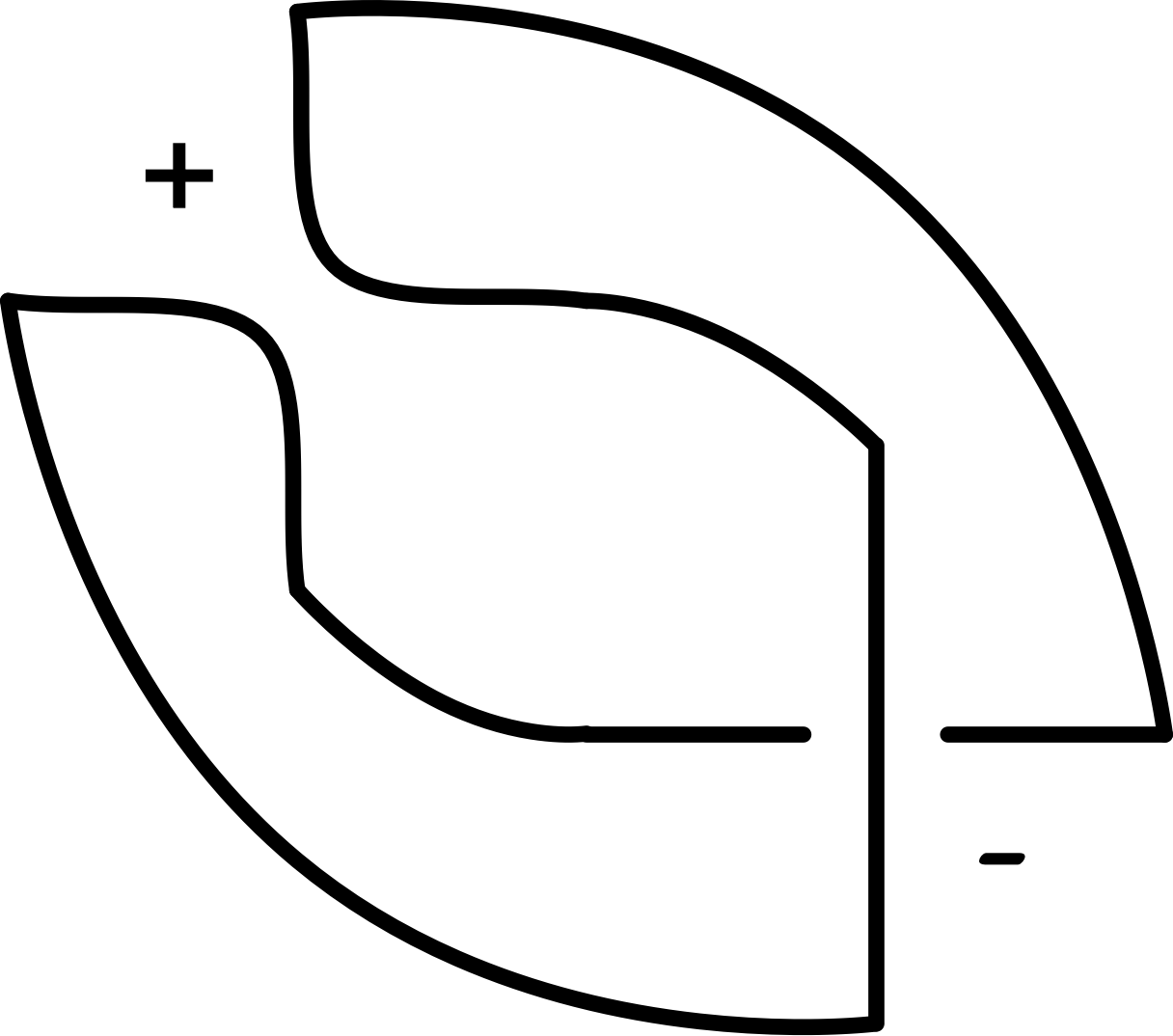} \\
       \\[-0.5em]
      & \includegraphics[width=0.13\textwidth]{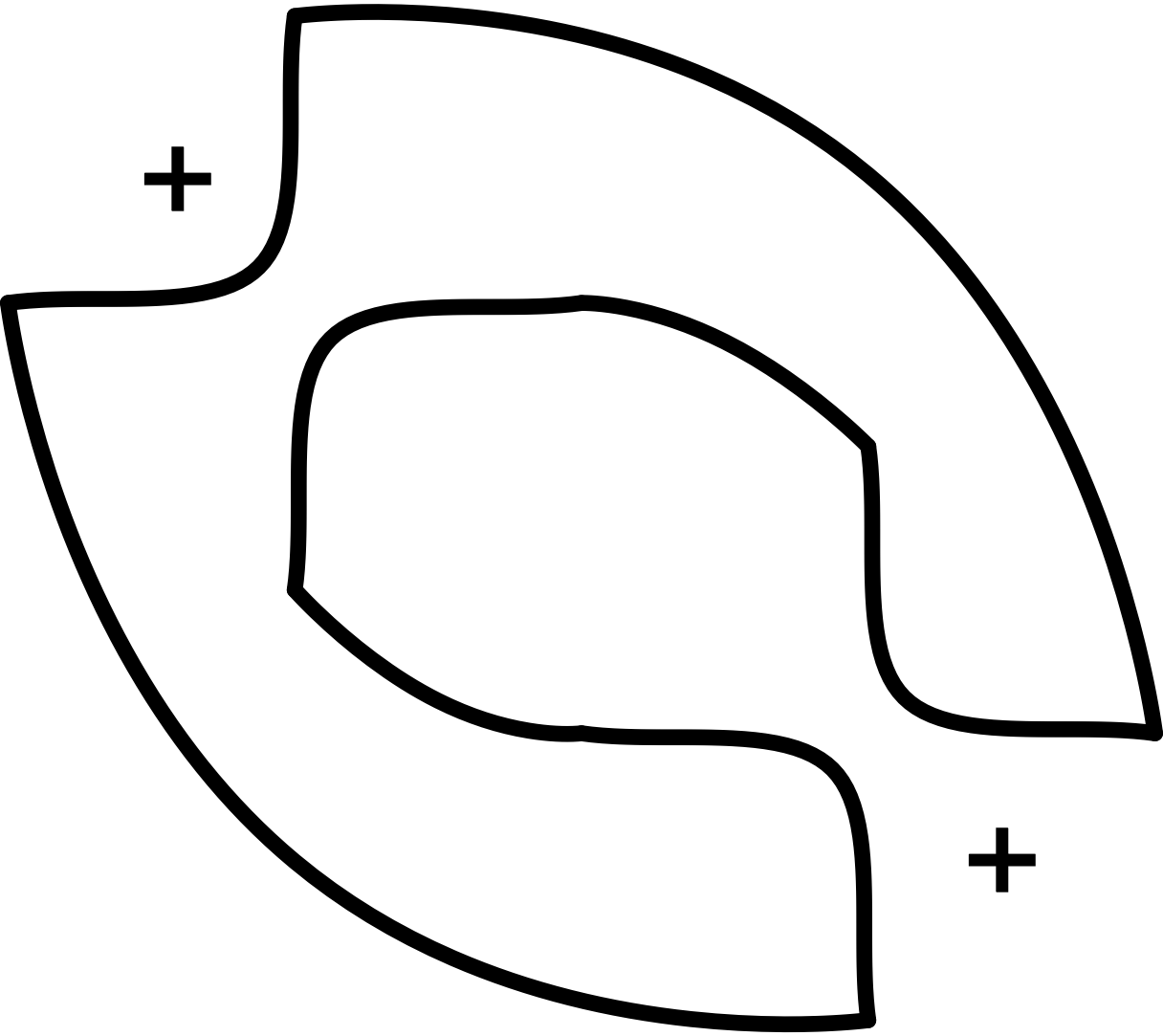} &
       \includegraphics[width=0.13\textwidth]{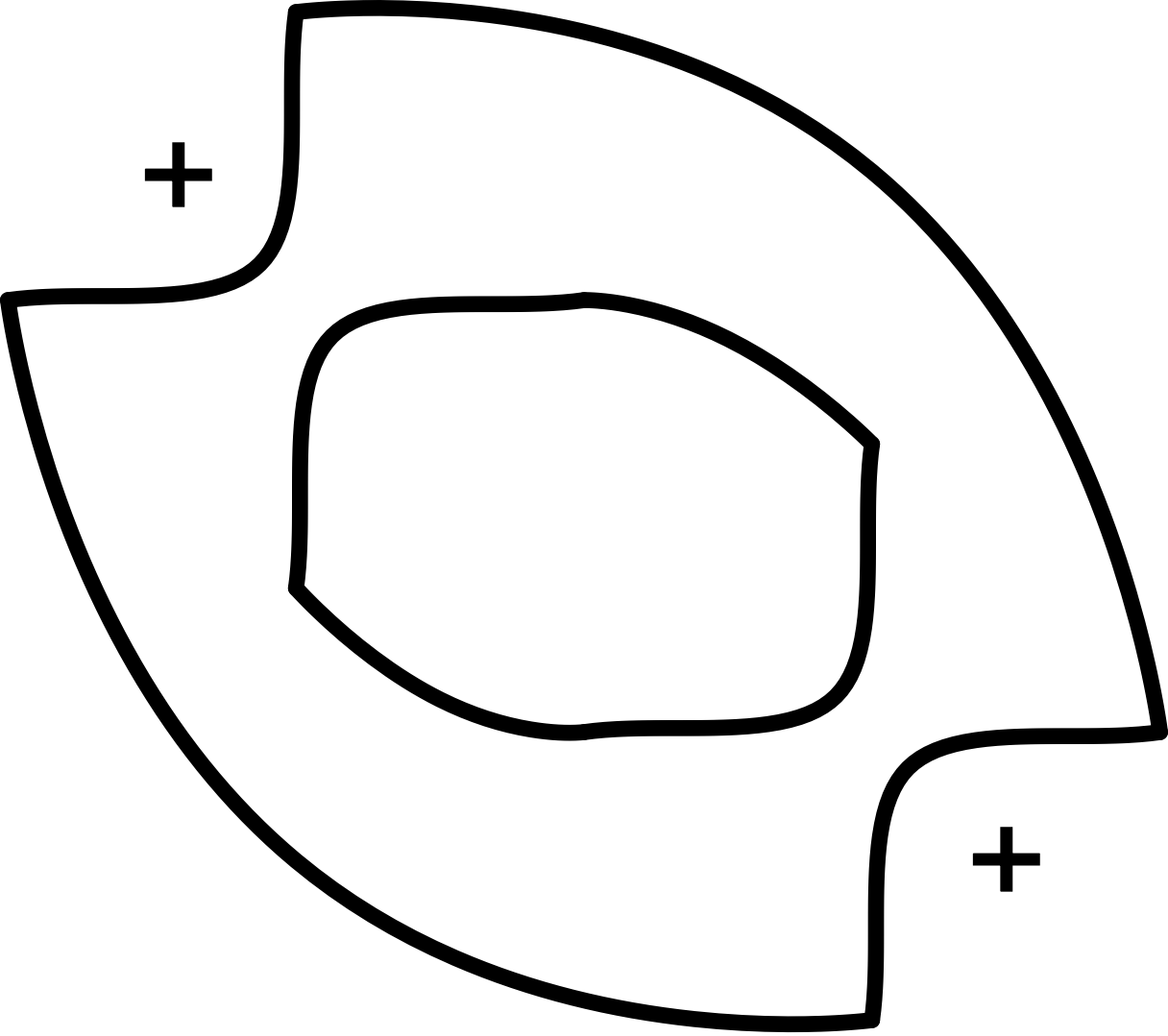} &
       \includegraphics[width=0.13\textwidth]{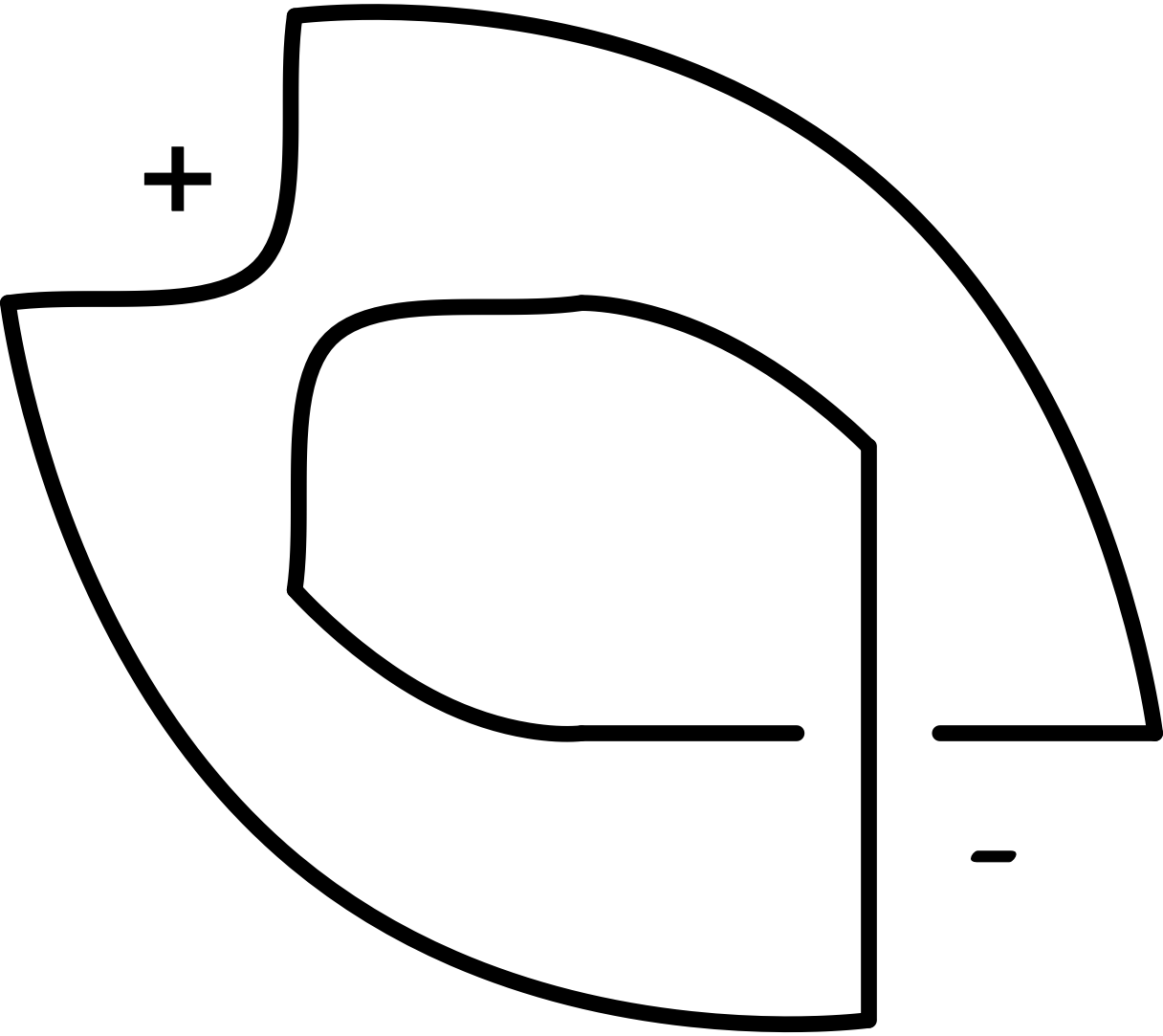} \\
       \\[-0.5em]
      & \includegraphics[width=0.13\textwidth]{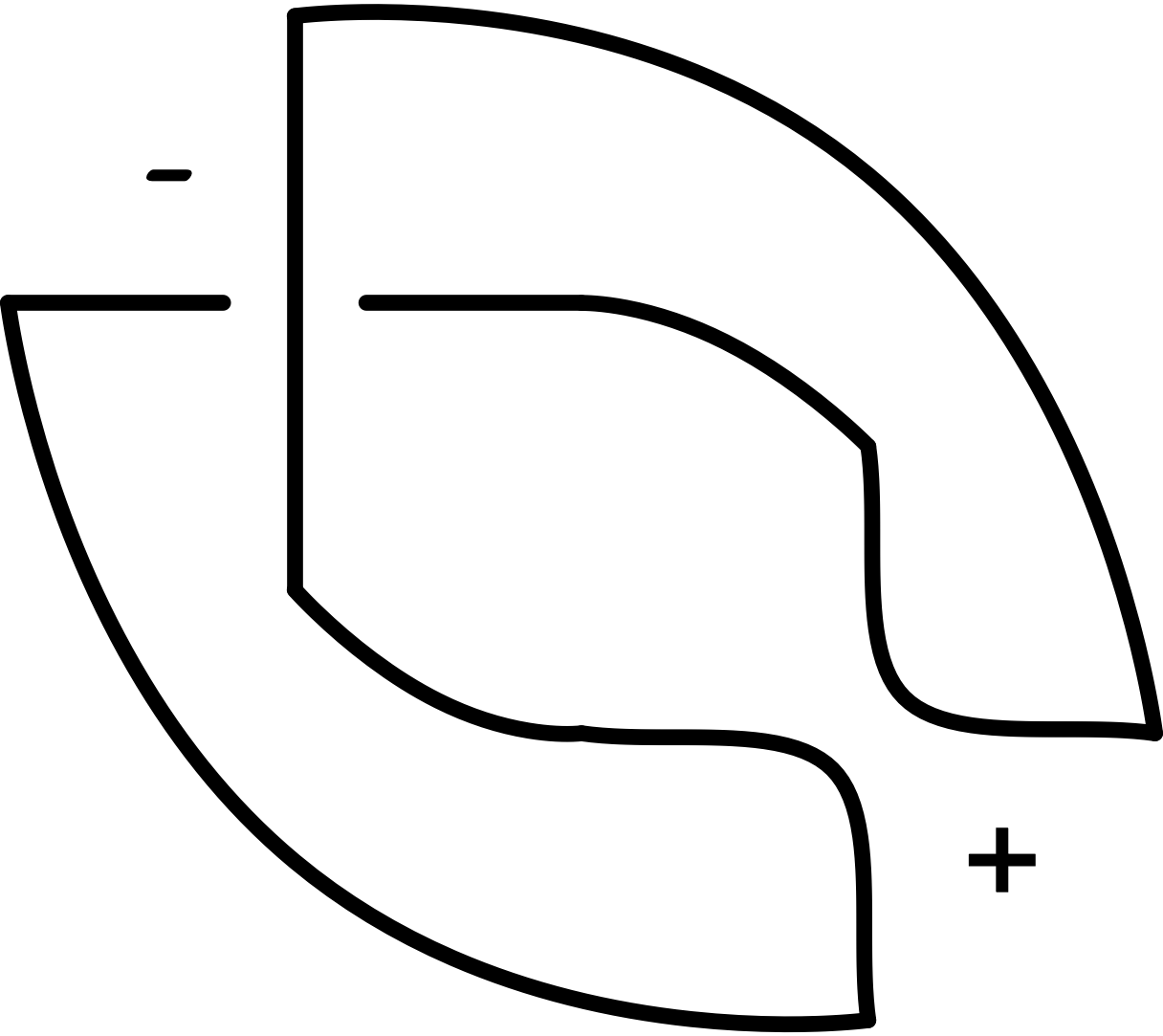} &
       \includegraphics[width=0.13\textwidth]{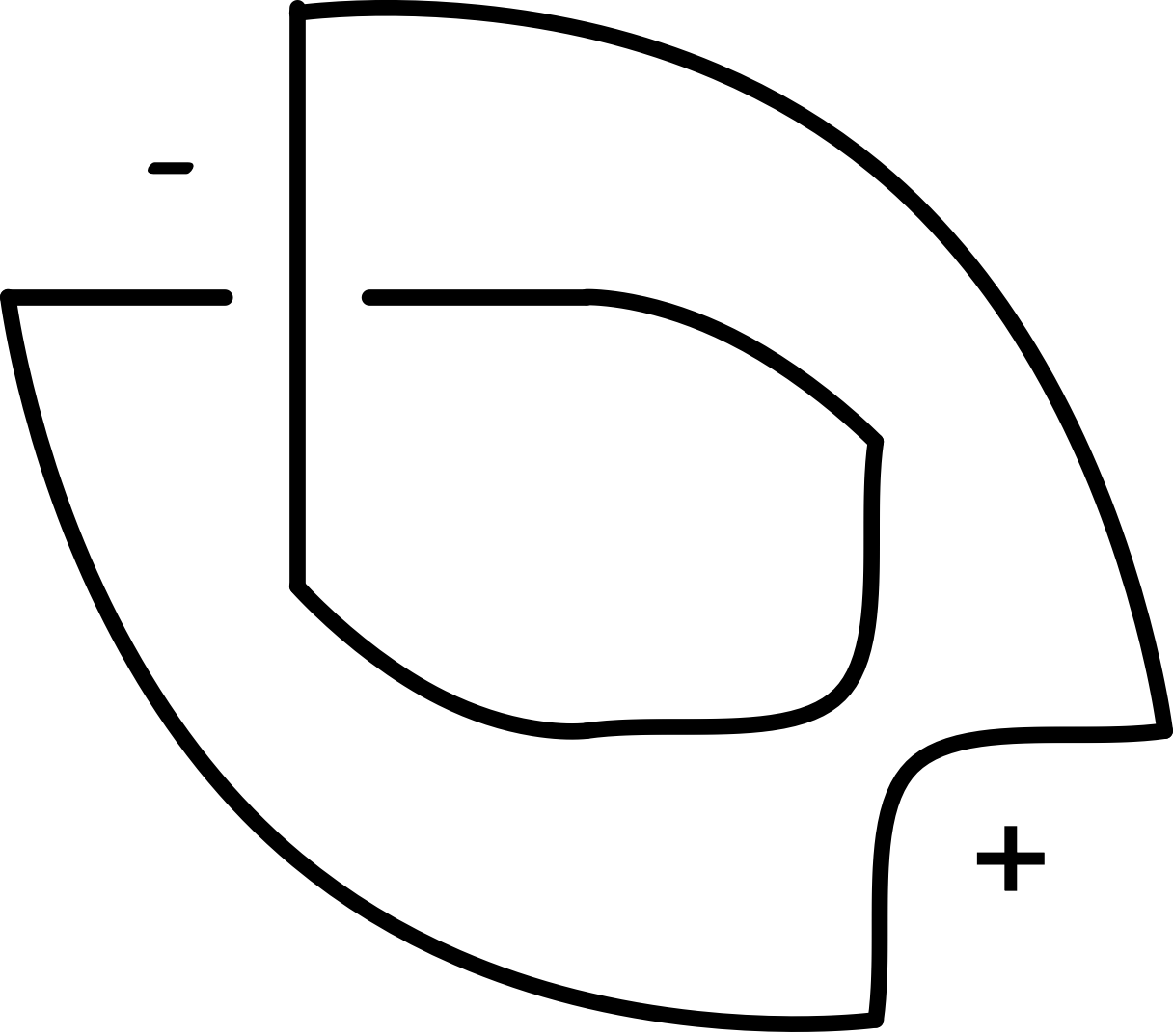} &
       \includegraphics[width=0.13\textwidth]{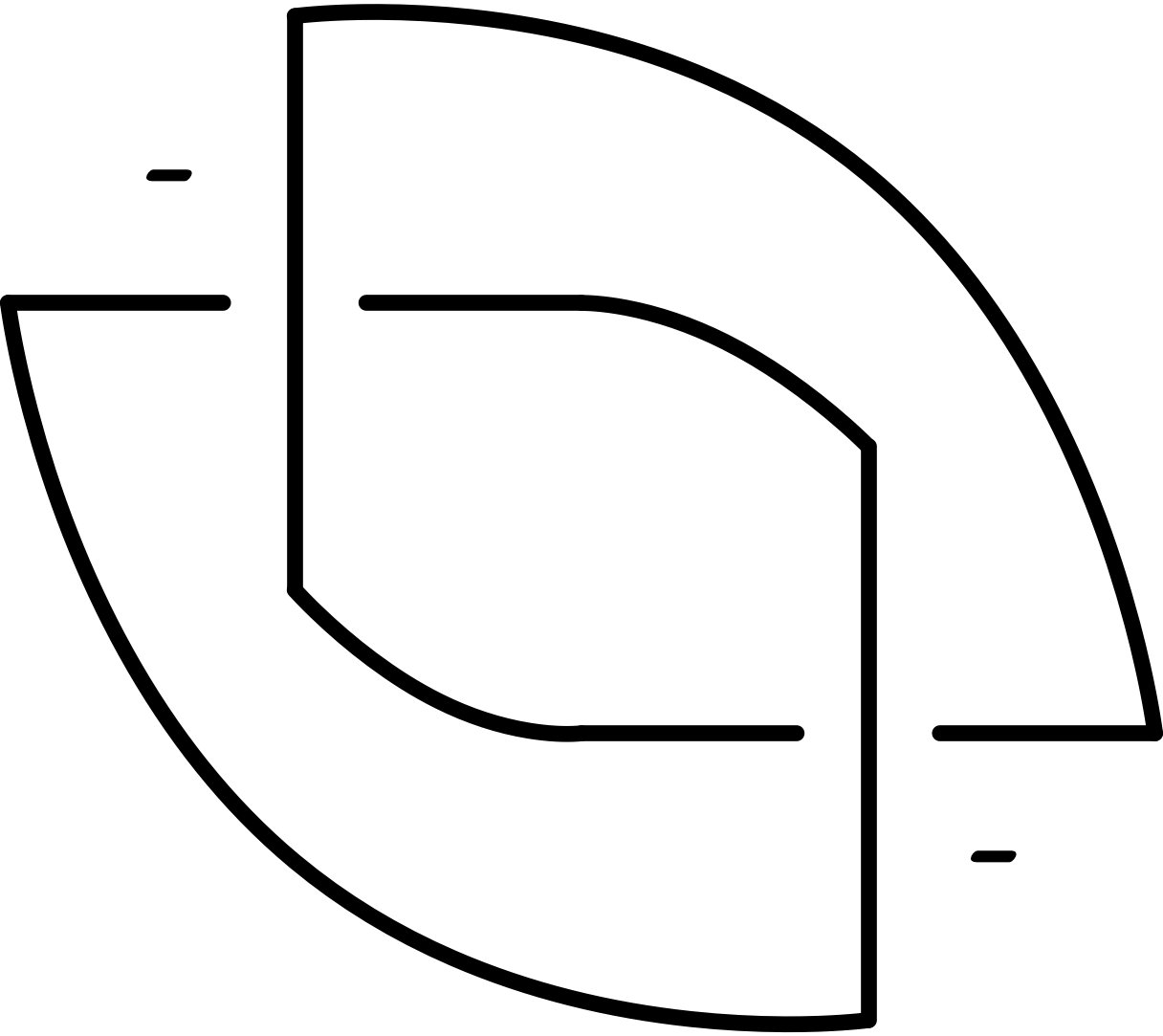} \\
       \\[-0.5em]
    \bottomrule
    \end{tabular}
  \end{center}
\end{table}

Conversely, for any {\sl dacp}, if we ignore the signs
at all vertices we get a valid even orientation (because
each sign applies to both pairs).
If a  {\sl dacp} comes from $\Phi(\tau)$
then we get back the  even orientation $\tau$.
 Therefore, the association from even orientations
 to {\sl dacp}'s is $1$-to-$3^{|V|}$, non-overlapping, and surjective.
Define $w$ to be a function assigning a weight to every signed pairing at every vertex and let the weight $\tilde{w}(\varphi)$ of an 
annotated circuit  partition $\varphi$, either undirected ({\sl acp})
 or directed ({\sl dacp}), be the product of weights at each vertex.
%In particular, when $w$ is defined such that
%$\left\{\begin{smallmatrix}
%a = w(\subinmatrix{1}_-) + w(\subinmatrix{2}_+) + w(\subinmatrix{3}_+) \\
%b = w(\subinmatrix{1}_+) + w(\subinmatrix{2}_-) + w(\subinmatrix{3}_+) \\
%c = w(\subinmatrix{1}_+) + w(\subinmatrix{2}_+) + w(\subinmatrix{3}_-) \\
%d = w(\subinmatrix{1}_-) + w(\subinmatrix{2}_-) + w(\subinmatrix{3}_-) \\
%\end{smallmatrix}\right.$.
For every vertex in the eight-vertex model with the parameter setting $(a, b, c, d)$,
%For every vertex with constraint matrix $\left[\begin{smallmatrix} d & & & a \\ & b & c & \\ & c & b & \\ a & & &d \end{smallmatrix}\right]$,
we define $w$ such that
\begin{equation}\label{eqn:weight}
\left\{\begin{smallmatrix}
%a = w(\sub{1}_-) + w(\sub{2}_+) + w(\sub{3}_+) \\
%b = w(\sub{1}_+) + w(\sub{2}_-) + w(\sub{3}_+) \\
%c = w(\sub{1}_+) + w(\sub{2}_+) + w(\sub{3}_-) \\
%d = w(\sub{1}_-) + w(\sub{2}_-) + w(\sub{3}_-) \\
a = w(\subinmatrix{1}_-) + w(\subinmatrix{2}_+) + w(\subinmatrix{3}_+) \\
b = w(\subinmatrix{1}_+) + w(\subinmatrix{2}_-) + w(\subinmatrix{3}_+) \\
c = w(\subinmatrix{1}_+) + w(\subinmatrix{2}_+) + w(\subinmatrix{3}_-) \\
d = w(\subinmatrix{1}_-) + w(\subinmatrix{2}_-) + w(\subinmatrix{3}_-) \\
\end{smallmatrix}\right..
\end{equation}
Note that for any $a,b,c,d$ this is a linear system of rank 4 in 
six variables, and there is a solution space of dimension 2 
(Lemma~\ref{lem:freedom_1} discusses this freedom).
Then the weight of an eight-vertex model configuration $\tau$ is equal to $\sum_{\varphi \in \Phi(\tau)}\tilde{w}(\varphi)$. This is obtained by 
writing a term in the summation in (\ref{Z-defn}), which is a  product of sums by (\ref{eqn:weight}),
%($*$),
 as a sum of products.
Note that a single {\sl acp}  has
the same weight when it becomes directed regardless
which directed state the {\sl dacp} is in.

We will illustrate the above in detail by the examples in 
 \tabref{tab:decomposition_a} and \tabref{tab:decomposition_c}.
We assume the same constraint $(a, b, c, d)$ is applied at $u$ and
$v$.
The orientation at one vertex determines the other in this graph $G$.
There are  a total  8 valid configurations, 4 of which are total reversals
of the other 4.   
 $Z(G) = 2[a^2 + b^2 + c^2 + d^2]$.
When we expand $Z(G)$ using (\ref{eqn:weight})
we get a total of 72 terms. These correspond to 72 {\sl dacp}'s.
There are 9 ways to assign
a pairing at $u$ and at $v$. 
If we consider the configuration in \tabref{tab:decomposition_a},
these 9 ways are listed under $\Phi(\tau)$, where the local orientation also
determines a sign $\pm$ at both $u$ and $v$.  These are 9 {\sl acp}'s
(without direction).
For each {\sl acp}  $\varphi$, the weight 
 $\tilde{w}(\varphi)$ is defined (without referring to
the {\sl dacp}, or the state of orientation on these circuits).
Three of the {\sl acp}'s (in the diagonal positions)
define two distinct circuits while the other six define one circuit each.
For each 2-tuple of pairings $(\rho_u, \rho_v)$ that
results in two circuits, the only valid annotations
assign $(+,+)$ or $(-,-)$ at $(u,v)$, giving a total of
6 {\sl acp}'s. And since each has two circuits, there are
a total of $24$ {\sl dacp}'s.
For the other six (off-diagonal) 2-tuples of pairings $(\rho_u, \rho_v)$ that 
results in a single  circuit, each has 4 valid annotations,
giving a total of
24 {\sl acp}'s. But these have only one circuit and thus
give 48 {\sl dacp}'s.
To appreciate the ``quantum superposition'' of the decomposition,
note that the same {\sl acp} that has $(\sub{2}_+, \sub{2}_+)$
at $(u, v)$ appears in both decompositions for
the distinct configurations in \tabref{tab:decomposition_a}
and \tabref{tab:decomposition_c}.

\begin{remark}
While a weight function $w$ satisfying (\ref{eqn:weight}) is not unique,
there are some regions of $(a,b,c,d)$ that can be specified 
directly in terms of  $w$
by any weight function $w$ satisfying (\ref{eqn:weight}),
and the specification is independent of the choice of
the weight function. 
E.g., the region $\overline{\mathcal{F}_>}$ is specified 
by
$\left\{\begin{smallmatrix}
w(\subinmatrix{1}_+) + w(\subinmatrix{2}_-) + w(\subinmatrix{3}_-) \ge 0\\
w(\subinmatrix{1}_-) + w(\subinmatrix{2}_+) + w(\subinmatrix{3}_-) \ge 0\\
w(\subinmatrix{1}_-) + w(\subinmatrix{2}_-) + w(\subinmatrix{3}_+) \ge 0\\
w(\subinmatrix{1}_+) + w(\subinmatrix{2}_+) + w(\subinmatrix{3}_+) \ge 0\\
\end{smallmatrix}\right.$.
Also $\mathcal{A}_\le$ is specified
by
$w(\sub{1}_-) \le w(\sub{1}_+)$,
$\mathcal{B}_\le$ by 
$w(\sub{2}_-) \le w(\sub{2}_+)$,
and 
$\mathcal{C}_\le$
by
$w(\sub{3}_-) \le w(\sub{3}_+)$.
In \lemref{lem:freedom_1}, we will show that
a nonnegative weight function $w$ satisfying (\ref{eqn:weight})
exists iff 
 $(a, b, c, d) \in \overline{\mathcal{F}_>}$.
\end{remark}

\begin{remark}
Although the
association from even orientations to
{\sl dacp}'s is  $1$-to-$3^{|V|}$, non-overlapping, and surjective,
the association from even orientations to
{\sl acp}'s is overlapping. If an {\sl acp} 
has $k$ circuits, it will be associated with $2^k$  even orientations.
It is this many-to-many association, with corresponding
weights, between  even orientations  and
{\sl acp}'s, that we call a \emph{quantum decomposition} of
eight-vertex model configurations, and each
is expressed as a ``superposition'' (weighted sum) of {\sl acp}'s.
\end{remark}

\captionsetup[subfigure]{labelformat=parens}
\begin{figure}[h!]
\centering
\begin{subfigure}[b]{0.32\linewidth}
\centering\includegraphics[width=0.78\linewidth]{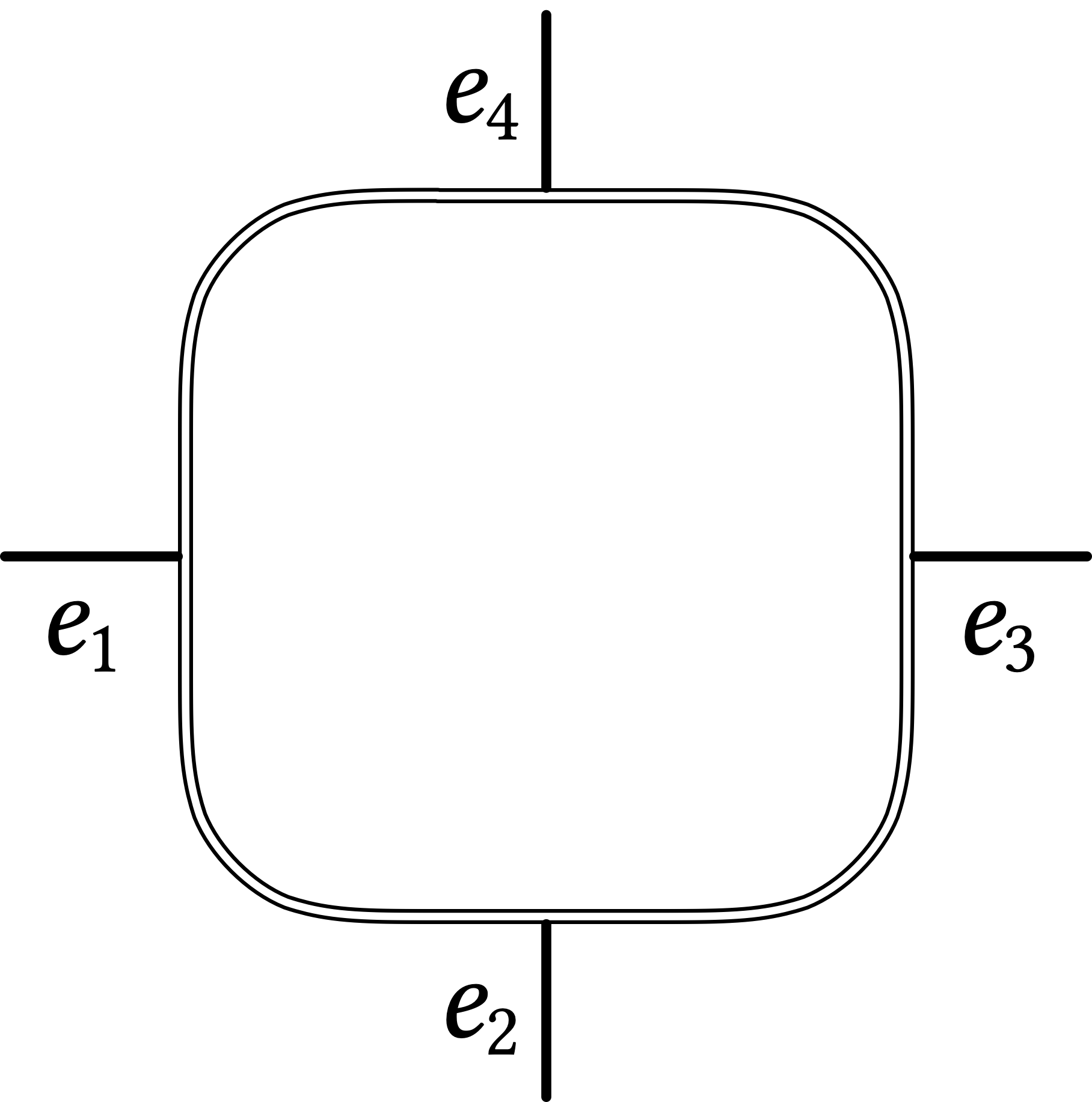}\caption{}\label{fig:gadget_plain}
\end{subfigure}
\begin{subfigure}[b]{0.32\linewidth}
\centering\includegraphics[width=0.82\linewidth]{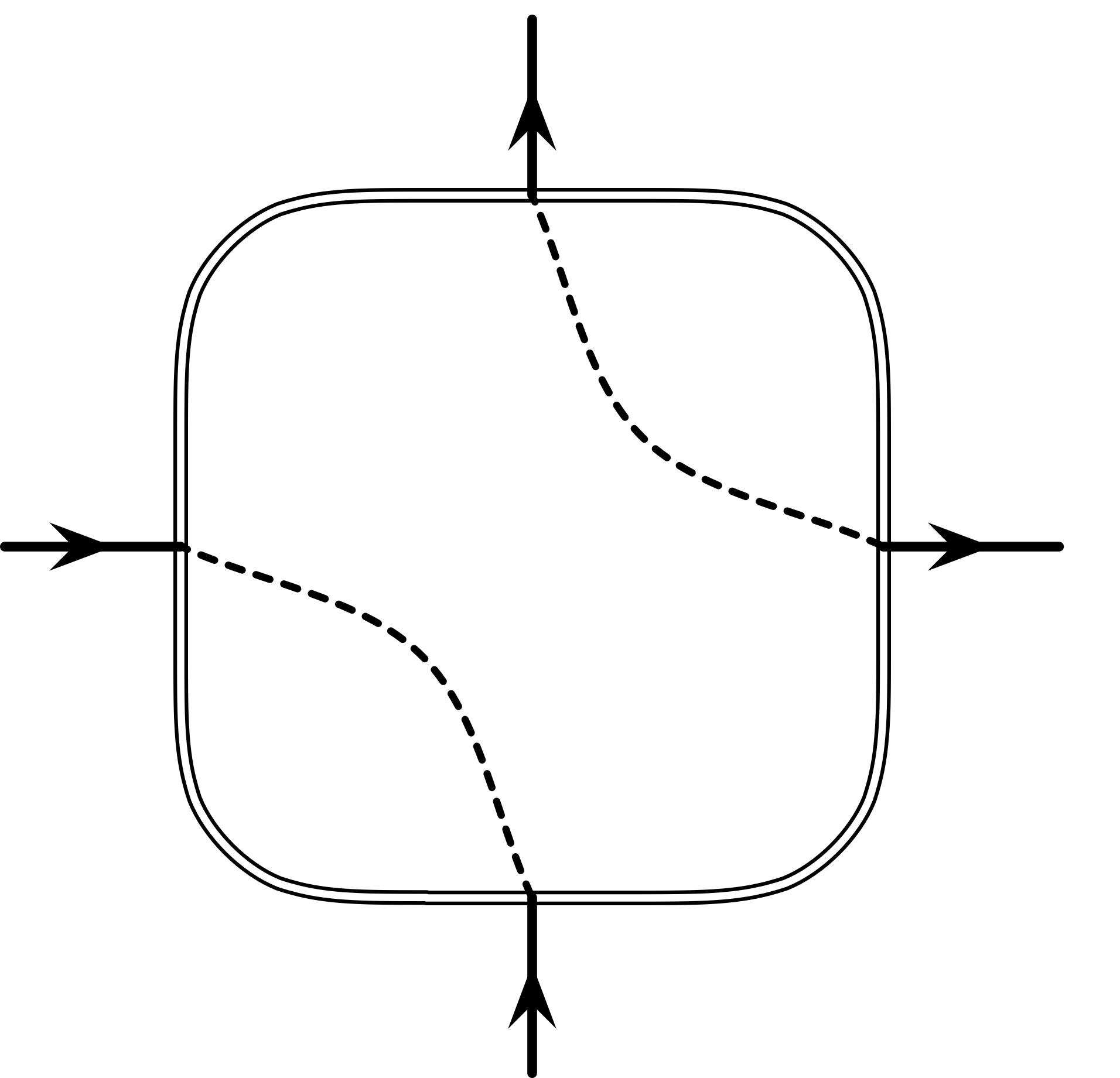}\caption{}\label{fig:gadget_a}
\end{subfigure}
\begin{subfigure}[b]{0.32\linewidth}
\centering\includegraphics[width=0.8\linewidth]{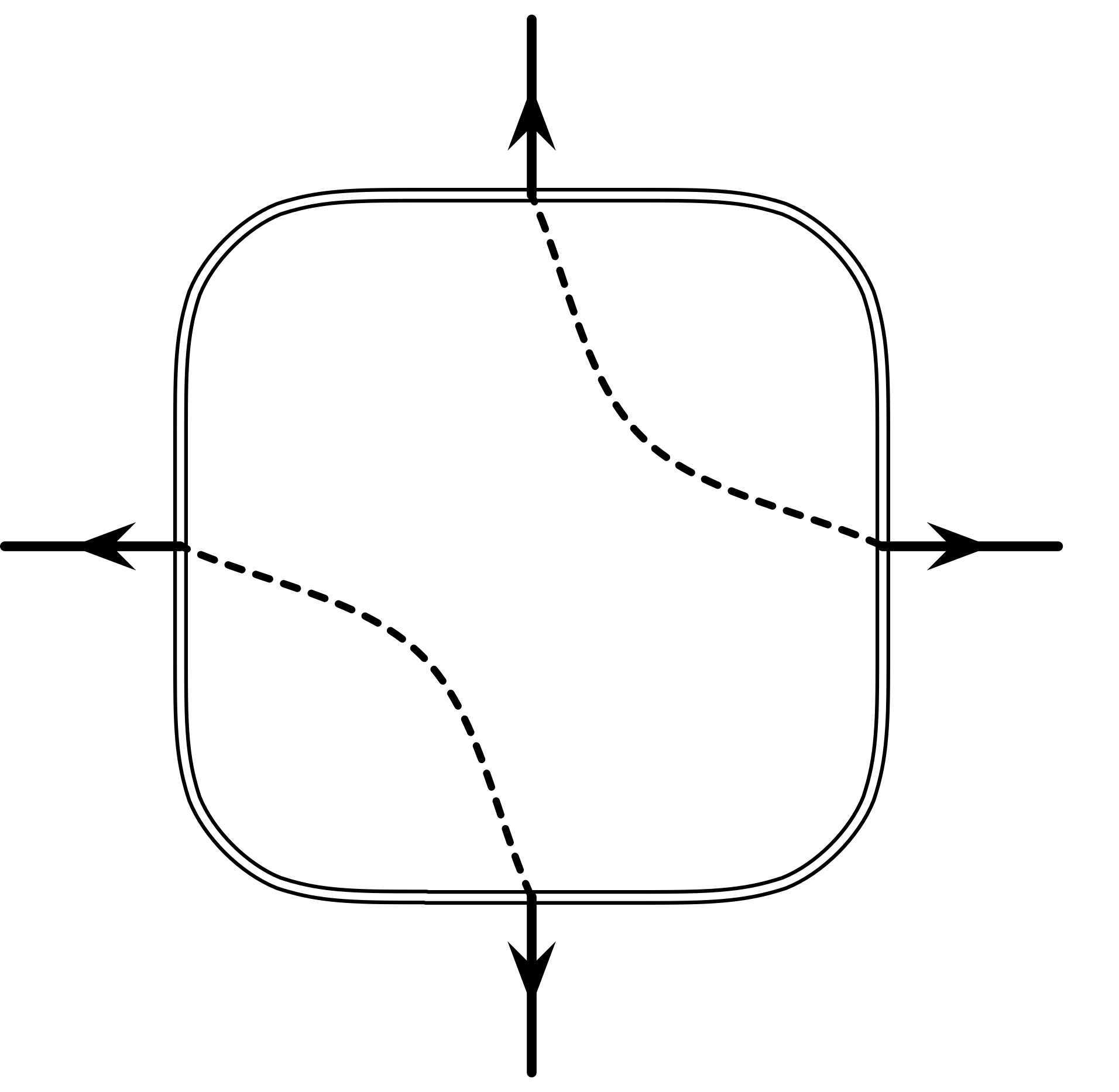}\caption{}\label{fig:gadget_d}
\end{subfigure}  
\caption{A 4-ary construction in the eight-vertex model.}\label{fig:gadget}
\end{figure}

A  \emph{4-ary construction} is 
a 4-regular graph $\Gamma$ having four ``dangling'' edges (\figref{fig:gadget_plain}),
and a constraint function on each node.
It defines a 4-ary constraint function with these four dangling
edges as input variables, when we sum 
the product of constraint function values on all vertices, over all
configurations on the internal edges of $\Gamma$.
%of the product of constraint function values on all vertices.
If we imagine the graph $\Gamma$ is shrunken to a single point
 except the 4 dangling edges, then
 a 4-ary construction can be viewed as a virtual vertex with  
parameters  $(a', b', c', d')$ in the 
eight-vertex model, for some $a', b', c', d' \ge 0$.
This is proved in the following lemma.
%We first prove the following lemma.
A \emph{planar 4-ary construction} is a 4-regular plane graph 
with four dangling edges on the outer face ordered
counterclockwise $e_1, e_2, e_3, e_4$.
%
%A \emph{planar 4-ary construction} is one obtainable from
%a 4-regular plane graph where we cut two edges to get
%the four dangling 
%edges:  $e_1$ and $e_2$ from one cut edge and $e_3$ and $e_4$ from another.
%
%This is equivalent to a 4-ary construction with a planar embedding
%such that  $e_1$ and $e_2$ are on the same face,
%and so are  $e_3$ and $e_4$ (it is possible they are all on
%the same face, then they are ordered clockwise on
%that face in the order $e_1, e_2, e_3, e_4$. 

\begin{lemma}
%%% JYC i think this lm allows diff constraint functions at diff nodes
%If the constraint function on every vertex 
If constraint functions in $\Gamma$ satisfy the even orientation rule and have arrow reversal symmetry, then the constraint function $f$ defined by $\Gamma$ also satisfies the even orientation rule and has arrow reversal symmetry.
\end{lemma}
%%% this proof should be omitted in the first 10 pages
\begin{proof}
Consider any even orientation on $\Gamma$,
and let $\Delta$ be the (sum of all in-degrees) $-$ (the sum of all out-degrees).
Each internal edge contrutes 0 to  $\Delta$.
By the even orientation rule, at every vertex
this difference is $0 \pmod 4$. Thus $\Delta \equiv 0 \pmod 4$.
Thus among the dangling edges, it also
satisfies the even orientation rule.

If we reverse all directions of an even orientation, which is
an involution,
each vertex contributes the same weight by arrow reversal symmetry.
Hence  $f$  also
has the arrow reversal symmetry.
%
%First, $f$ still obeys the even orientation rule, because it cannot take nonzero values on inputs with odd Hamming weight. If it were not the case, then there must exist a vertex on which the constraint function takes nonzero value on a input with odd Hamming weight, which contradicts the even orientation rule on this vertex.
%
%Second, $f$ still has arrow reversal symmetry.
%For any valid orientation of edges in the 4-ary construction contributing a term to $f(x)$,
%reversing the orientations on all edges has the same contribution to
%$f(\overline{x})$, because the constraint function 
%on each vertex of degree 4 satisfies the
%arrow reversal symmetry.
\end{proof}

A \emph{trail \& circuit partition} ({\sl tcp}) for a 4-ary construction $\Gamma$ is a partition of the edges in $\Gamma$ into edge-disjoint circuits and exactly two \textit{trails} (walks with no repeated edges) which
end in the four dangling edges.
An \emph{annotated trail \& circuit partition} ({\sl atcp})
for $\Gamma$ is a {\sl tcp} with a valid annotation,
which assigns an \emph{even} number of $-$ sign along each circuit.
Like circuits, each trail in an {\sl atcp} has exactly two directed states.
If an {\sl atcp} $\varphi$ has $k$ circuits (and $2$ trails), then
$\varphi$ defines $2^{k+2}$ \emph{directed annotated trail \& circuit partitions}
({\sl datcp}'s).
The weight $\tilde{w}(\varphi)$ of
 an annotated trail \& circuit partition 
$\varphi$, either an {\sl atcp} or {\sl datcp},
can be similarly defined.
Again set the weight function as in (\ref{eqn:weight}).
%($*$).
%$w$ such that
%$\left\{\begin{smallmatrix}
%a = w(\subinmatrix{1}_-) + w(\subinmatrix{2}_+) + w(\subinmatrix{3}_+) \\
%b = w(\subinmatrix{1}_+) + w(\subinmatrix{2}_-) + w(\subinmatrix{3}_+) \\
%c = w(\subinmatrix{1}_+) + w(\subinmatrix{2}_+) + w(\subinmatrix{3}_-) \\
%d = w(\subinmatrix{1}_-) + w(\subinmatrix{2}_-) + w(\subinmatrix{3}_-) \\
%\end{smallmatrix}\right.$.

Denote the constraint function of $\Gamma$ by $f$ and use the notations introduced in \secref{sec:prelim}.
Consider $f(0011)$.
Under the eight-vertex model, 
if a configuration $\tau$ of the 4-ary construction with constraint function $f$ 
has a nonzero contribution to $f(0011)$, it has $e_1, e_2$ coming in 
and $e_3, e_4$ going out. The contribution by $\tau$ is
a weighted sum over a set $\Phi_{0011}(\tau)$ of {\sl datcp}'s.
Each {\sl datcp} in $\Phi_{0011}(\tau)$ is 
captured in exactly one of the following three types,
according to how  $e_1, e_2, e_3, e_4$ are connected by
the two trails:
%it can be expressed as a weighted sum over a set $\Phi_{0011}(\tau)$ of {\sl datcp}. 
%According to how  $e_1, e_2, e_3, e_4$ are connected by
%the two trails, $\Phi_{0011}(\tau)$ is partitioned into three parts:
%There are three possible compositions of the two directed annotated trails
%and every {\sl datcp} $\varphi \in \Phi_{0011}(\tau)$ belongs to
%exactly one of these:
%, there are three possible compositions of the two directed annotated trails:
\vspace{-2mm}
\setlist[enumerate]{itemsep=-1mm}
\begin{enumerate}[(1)]
\item
%$\{e_1 \leadsto e_2, e_3 \leftrightsquigarrow e_4\}$
$\{\ \xrightarrow{e_1} \Square \xleftarrow{e_2},~ \xleftarrow{e_3} \Square \xrightarrow{e_4} \}$
and on both trails
the numbers of minus pairings are odd; or
\item
%$\{e_1 \leadsto e_4, e_2 \leadsto e_3\}$
$\{\ \xrightarrow{e_1} \Square \xrightarrow{e_4},~ \xrightarrow{e_2} \Square \xrightarrow{e_3} \}$
and on both trails 
the numbers of minus pairings are even (\figref{fig:gadget_a}); or
\item
%$\{e_1 \leadsto e_3, e_2 \leadsto e_4\}$
$\{\ \xrightarrow{e_1} \Square \xrightarrow{e_3},~ \xrightarrow{e_2} \Square \xrightarrow{e_4} \}$
and on both trails 
the numbers of minus pairings are even.
\end{enumerate}
\vspace{-2mm}
Let $\Phi_{0011,\sub{1}_-}$,
 $\Phi_{0011,\sub{2}_+}$ and $\Phi_{0011,\sub{3}_+}$ be the subsets of {\sl datcp}'s 
contributing to $f(0011)$
%(distributed in potentially many different eight-vertex configurations) 
defined in case (1), (2) and (3) respectively.
%captured by case (1), (2) and (3) respectively.
% above;
%denote by $\Phi_{0011,\sub{2}_+}$ the set of {\sl datcp} captured by case (2) above; denote by $\Phi_{0011,\sub{3}_+}$ the set of {\sl datcp} captured by case (3) above.
%In terms of {\sl datcp} of the 4-ary construction, 
The value $f(0011)$ is a weighted sum of contributions according to $\tilde{w}$
 from these three disjoint sets.
% $\Phi_{0011,\sub{1}_-}$, $\Phi_{0011,\sub{2}_+}$ and $\Phi_{0011,\sub{3}_+}$.
Defining the weight of a set $\Phi$ of {\sl datcp}'s by $W(\Phi) = \sum_{\varphi \in \Phi} \tilde{w}(\varphi)$ yields $f(0011) = W(\Phi_{0011,\sub{1}_-}) + W(\Phi_{0011,\sub{2}_+}) + W(\Phi_{0011,\sub{3}_+})$.
Similarly we can define $\Phi_{1100,\sub{1}_-}, \Phi_{1100,\sub{2}_+}$
and $\Phi_{1100,\sub{3}_+}$, and get
 $f(1100) = W(\Phi_{1100,\sub{1}_-}) +  W(\Phi_{1100,\sub{2}_+}) + W(\Phi_{1100,\sub{3}_+})$.
Note that there is a bijective weight-preserving map between $\Phi_{0011,\sub{1}_-}$ and $\Phi_{1100,\sub{1}_-}$ by reversing the direction of every circuit and trail of a {\sl datcp}. Thus, $W(\Phi_{0011,\sub{1}_-}) = W(\Phi_{1100,\sub{1}_-})$, $W(\Phi_{0011,\sub{2}_+}) = W(\Phi_{1100,\sub{2}_+})$,
and $W(\Phi_{0011,\sub{3}_+}) = W(\Phi_{1100,\sub{3}_+})$.
Consequently $f(0011) = f(1100)$.
Similarly we have $f(0110) = f(1001)$, $f(0101) = f(1010)$ and $f(0000) = f(1111)$.

For any pairing $\varrho$,
and for every 4-bit pattern $b_1b_2b_3b_4 \in \{0, 1\}^4$,
we can define $\Phi_{b_1b_2b_3b_4, \varrho_+}$ 
if (both) paired $b_i \not = b_j$,
and  $\Phi_{b_1b_2b_3b_4, \varrho_-}$ 
if (both) paired $b_i = b_j$.
%%% when +, the paired are 1-in-1-out. so bits of the pair are unequal.
%   when -, the paired are 2-in, or 2-out, so bits of the pair are equal.
Then a
  further important observation is that for each {\sl datcp} in $\Phi_{0011,\sub{1}_-}$, if we only reverse every edge in the  trail between $\ \xrightarrow{e_1} \Square \xleftarrow{e_2}\ $
% to $e_4 \leadsto e_1$ 
and keep the states of all circuits and the other trail unchanged, this {\sl datcp} has the same weight but now lies in $\Phi_{1111,\sub{1}_-}$.
%%% JYC: pl do a suitable picture
%(\figref{fig:gadget_d}).
This is because at every vertex $v$, reversing the orientation of any one branch of the given 
(annotated) pairing $\boldsymbol{\varrho}_v \in \{\sub{1}, \sub{2}, \sub{3}\} \times \{+,-\}$ does not change the value $w(\boldsymbol{\varrho}_v)$.
In this way, we set up a one-to-one weight-preserving map between $\Phi_{0011,\sub{1}_-}$ and $\Phi_{1111,\sub{1}_-}$, hence $W(\Phi_{0011,\sub{1}_-}) = W(\Phi_{1111,\sub{1}_-})$. 
%In this way the following equalities can be proved:
Combining the result in the last paragraph we have proved the
first item below, and we name its common value $W(\sub{1}_-)$.
%The other items  can be proved similarly.
The other items  are proved similarly.
\vspace{-2mm}
\setlist[itemize]{itemsep=-1mm}
\begin{itemize}
\item
$W(\sub{1}_-) = W(\Phi_{0011,\sub{1}_-}) = W(\Phi_{1100,\sub{1}_-}) = W(\Phi_{0000,\sub{1}_-}) = W(\Phi_{1111,\sub{1}_-})$;
\item
$W(\sub{2}_-) = W(\Phi_{0110,\sub{2}_-}) = W(\Phi_{1001,\sub{2}_-}) = W(\Phi_{0000,\sub{2}_-}) = W(\Phi_{1111,\sub{2}_-})$;
\item
$W(\sub{3}_-) = W(\Phi_{0101,\sub{3}_-}) = W(\Phi_{1010,\sub{3}_-}) = W(\Phi_{0000,\sub{3}_-}) = W(\Phi_{1111,\sub{3}_-})$;
\item
$W(\sub{1}_+) = W(\Phi_{0110,\sub{1}_+}) = W(\Phi_{1001,\sub{1}_+}) = W(\Phi_{0101,\sub{1}_+}) = W(\Phi_{1010,\sub{1}_+})$;
\item
$W(\sub{2}_+) = W(\Phi_{0011,\sub{2}_+}) = W(\Phi_{1100,\sub{2}_+}) = W(\Phi_{0101,\sub{2}_+}) = W(\Phi_{1010,\sub{2}_+})$;
\item
$W(\sub{3}_+) = W(\Phi_{0011,\sub{3}_+}) = W(\Phi_{1100,\sub{3}_+}) = W(\Phi_{0110,\sub{3}_+}) = W(\Phi_{1001,\sub{3}_+})$.
\end{itemize}
\vspace{-2mm}
Consequently, $f$ has
parameters
$\left\{\begin{smallmatrix}
a' = W(\subinmatrix{1}_-) + W(\subinmatrix{2}_+) + W(\subinmatrix{3}_+) \\
b' = W(\subinmatrix{1}_+) + W(\subinmatrix{2}_-) + W(\subinmatrix{3}_+) \\
c' = W(\subinmatrix{1}_+) + W(\subinmatrix{2}_+) + W(\subinmatrix{3}_-) \\
d' = W(\subinmatrix{1}_-) + W(\subinmatrix{2}_-) + W(\subinmatrix{3}_-) \\
\end{smallmatrix}\right.$.

\begin{proof}[Proof of \thmref{thm:property_3}]
For any weight function satisfying (\ref{eqn:weight}), %($*$),
one can easily verify that
$(a,b,c,d) \in \overline{\mathcal{F}_>}$ iff the following inequalities hold:
$\left\{\begin{smallmatrix}
w(\subinmatrix{1}_+) + w(\subinmatrix{2}_-) + w(\subinmatrix{3}_-) \ge 0\\
w(\subinmatrix{1}_-) + w(\subinmatrix{2}_+) + w(\subinmatrix{3}_-) \ge 0\\
w(\subinmatrix{1}_-) + w(\subinmatrix{2}_-) + w(\subinmatrix{3}_+) \ge 0\\
w(\subinmatrix{1}_+) + w(\subinmatrix{2}_+) + w(\subinmatrix{3}_+) \ge 0\\
\end{smallmatrix}\right.$.

By \lemref{lem:freedom_1}, we can assume $w$ is a nonnegative weight function.
By definition,
each of the six quantities
$W(\sub{1}_+), W(\sub{1}_-),  W(\sub{2}_+),
 W(\sub{2}_-), W(\sub{3}_+)$ and $W(\sub{3}_-)$
is a sum over a set of {\sl datcp}'s of products of values of $w$,
and thus they are all nonnegative. Hence,
the constraint function $f$ defined by $\Gamma$ satisfies
$\left\{\begin{smallmatrix}
W(\subinmatrix{1}_+) + W(\subinmatrix{2}_-) + W(\subinmatrix{3}_-) \ge 0\\
W(\subinmatrix{1}_-) + W(\subinmatrix{2}_+) + W(\subinmatrix{3}_-) \ge 0\\
W(\subinmatrix{1}_-) + W(\subinmatrix{2}_-) + W(\subinmatrix{3}_+) \ge 0\\
W(\subinmatrix{1}_+) + W(\subinmatrix{2}_+) + W(\subinmatrix{3}_+) \ge 0\\
\end{smallmatrix}\right.$.
This is equivalent to the assertion that the parameters $(a',b',c',d')$ of $f$
belong to the region $\overline{\mathcal{F}_>}$.
\end{proof}

%The following lemma is a key result for the closure properties.
%\begin{lemma}\label{lem:property_1}
%Let $\Gamma$ be a 4-ary construction in the eight-vertex model. If the constraint function on every vertex in $\Gamma$ satisfies
%$\left\{\begin{smallmatrix}
%w(\subinmatrix{1}_+) \ge w(\subinmatrix{1}_-)\\
%w(\subinmatrix{2}_+) \ge w(\subinmatrix{2}_-)\\
%w(\subinmatrix{3}_+) \ge w(\subinmatrix{3}_-)\\
%\end{smallmatrix}\right.$ where $w$ is nonnegative,
%then the constraint function $f$ defined by $\Gamma$ satisfies
%$\left\{\begin{smallmatrix}
%W(\subinmatrix{1}_+) \ge W(\subinmatrix{1}_-)\\
%W(\subinmatrix{2}_+) \ge W(\subinmatrix{2}_-)\\
%W(\subinmatrix{3}_+) \ge W(\subinmatrix{3}_-)\\
%\end{smallmatrix}\right.$.
%\end{lemma}
%\begin{proof}

\begin{proof}[Proof of \thmref{thm:property_1}]
By definition 
$(a, b, c, d) \in \mathcal{A}_\le \bigcap \mathcal{B}_\le \bigcap \mathcal{C}_\le$
means that 
$\left\{\begin{smallmatrix}
a + d \le b + c\\
b + d \le a + c\\
c + d \le a + b\\
\end{smallmatrix}\right.$.
By the weight function $w$ defined in (\ref{eqn:weight}) this is equivalent to %($*$) this is equivalent to
$\left\{\begin{smallmatrix}
w(\subinmatrix{1}_+) \ge w(\subinmatrix{1}_-)\\
w(\subinmatrix{2}_+) \ge w(\subinmatrix{2}_-)\\
w(\subinmatrix{3}_+) \ge w(\subinmatrix{3}_-)\\
\end{smallmatrix}\right.$.
%%% JYC for any w satisfies (*).
%%% it is an interesting property that even though w has 2-dim freedom, this
%%% set of a b c d has a characterization in ANY and ALL weight function $w$ in ($*$)
Since $\mathcal{A}_\le \bigcap \mathcal{B}_\le \bigcap \mathcal{C}_\le \subset \overline{\mathcal{F}_>}$, by \lemref{lem:freedom_1}
% and the proof of \thmref{thm:property_3} 
we can assume $w$ is nonnegative.
%Then \thmref{thm:property_1} is a direct consequence of \thmref{thm:property_1}.
To prove \thmref{thm:property_1} we only need to establish
$\left\{\begin{smallmatrix}
W(\subinmatrix{1}_+) \ge W(\subinmatrix{1}_-)\\
W(\subinmatrix{2}_+) \ge W(\subinmatrix{2}_-)\\
W(\subinmatrix{3}_+) \ge W(\subinmatrix{3}_-)\\
\end{smallmatrix}\right.$.
We prove $W(\sub{1}_+) \ge W(\sub{1}_-)$. Proof for the other two inequalities 
is symmetric.
% to this one.
%$W(\sub{1}_+) \ge W(\sub{1}_-)$ is saying that in a 4-ary construction 
%depicted  in \figref{fig:gadget_plain}, the weighted sum of all {\sl atcp} in which along each of the two annotated trails, $e_1$ to $e_2$ and $e_3$ to $e_4$, there is an even number of $-$ is no less than the weighted sum of all {\sl atcp} in which along each of the two annotated trails there is an odd number of $-$.
%
%%%JYC i don't see why this needs to be restated...
 
%%% moved to before lm 3.5 because it is more uniform defing there for atcp
%
%A \emph{trail \& circuit partition} (or {\sl tcp} for short) for a 4-ary construction $\Gamma$ is a partition of the edges in $\Gamma$ into edge-disjoint circuits and exactly two \textit{trails}.
%%% JYC since this does not even mention signs, so of course there is
% no requirement whatever number of -'s  being even...
%(without the restriction that there is an \emph{even} number of minus pairings along each circuit).
An {\sl atcp} is a {\sl tcp} together with a valid annotation.
Consider the set  $\Psi$ of {\sl tcp}'s  such that the two 
(unannotated) trails connect
 $e_1$ with $e_2$, and $e_3$ with $e_4$.
%%% JYC: no direction yet.
Denote by $\chi_{12}$ (respectively $\chi_{34}$) the trail in $\psi$
 connecting $e_1$ and $e_2$ (respectively $e_3$ and $e_4$).
Each {\sl tcp} $\psi \in \Psi$ may have many valid annotations.
%Denote by $\chi_{12}$ the trail connecting $e_1$ and $e_2$ in $\varphi$ and 
%by $\chi_{34}$ the trail connecting $e_3$ and $e_4$ in $\varphi$.

Since  $\Gamma$  is 4-regular,
any vertex inside $\Gamma$ appears exactly twice 
counting multiplicity in a {\sl tcp} $\psi$.
It appears either as a self-intersection point of
a trail or a circuit, or alternatively
 in
exactly two distinct trails/circuits.
%(including as a self-intersection point of some circuit or trail).
%Note that since each vertex in a trail/circuit has degree 2,
%either it is a self-intersection point, or appears in
%exactly two trails/circuits.
So when traversed,
in total one encounters an even number of $-$  among all circuits and 
the two trails
in any valid annotation of $\psi$, and 
since one encounters an even number of $-$ along each circuit,
the numbers of $-$ along $\chi_{12}$ and $\chi_{34}$ have the same parity.
%Since one encounters an even number of $-$ along each circuit,
%one encounters an even number of $-$ along $\chi_{12}$ iff
%one encounters an even number of $-$ along $\chi_{34}$.
We say a valid annotation of $\psi$ is \emph{positive} 
if there is an even number of $-$ along $\chi_{12}$ (and $\chi_{34}$),
and \emph{negative} otherwise.
% a valid annotation of $\psi$ is \emph{negative} if there is an odd number of $-$ along $\chi_{12}$ (and $\chi_{34}$).

To prove $W(\sub{1}_+) \ge W(\sub{1}_-)$, it suffices to prove that for each {\sl tcp} 
$\psi \in \Psi$, the total weight $W_+$ contributed by the set of positive annotations of $\psi$ is 
at least the total weight $W_-$ contributed by the set of negative annotations of $\psi$.
%\textcolor{red}{(Tianyu: I want to say this is somewhat surprising.)} 
We prove this nontrivial statement  by induction on the number $N$ of 
vertices shared by any two distinct circuits in $\psi$.

\medskip
\noindent
\textbf{Base case}: The base case is $N=0$. 
%is that in $\psi$, 
%no vertex in $\Gamma$ belongs to
% two distinct circuits/trails.
Let us first also assume that no trail or circuit is self-intersecting.
Then every vertex on any circuit $C$ of $\psi$
 is shared by $C$ and exactly one trail, $\chi_{12}$ or $\chi_{34}$.
%So for any circuit $C$, every vertex on $C$ is shared with either
% $\chi_{12}$ or $\chi_{34}$; 
Also, every vertex on $\chi_{12}$ or $\chi_{34}$ is shared with some circuit
or the other trail.

We will account for the product values of  $w(\boldsymbol{\varrho}_v)$
according to how $v$ is shared.
We first consider shared vertices of a circuit $C \in \psi$
with the trails.
Let $s, t \ge 0$ be the numbers of vertices $C$
 shares with $\chi_{12}$ and $\chi_{34}$, respectively. 
Let $x_i \ (1 \le i \le s)$ (if $s >0$) and $y_j \ (1 \le j \le t)$ (if $t >0$)
be these shared vertices respectively (for $s=0$ or $t=0$,  
the statements below are vacuously true).
For any $v$, if  
$\varrho$ is the pairing at $v$ according to $\psi$,
then let
  $w_{+}(v) = w(\varrho_+)$,
 and $w_{-}(v) = w(\varrho_-)$, both at $v$.
% $x_i$.
%Similarly we define $w_{+}(y_j)$ and $w_{-}(y_j)$ at $y_j$.
In any  valid annotation of $\psi$  (either positive or negative), 
one encounters an even number of $-$ on the vertices along $C$,
each of which is shared with exactly one of $\chi_{12}$ and $\chi_{34}$.
Hence the number of $-$ in $x_i \ (1 \le i \le s)$ has
the same parity as the number of $-$ in $y_j \ (1 \le j \le t)$. 
Other than having the same parity,
the annotation for $x_i \ (1 \le i \le s)$ 
is independent from the annotation for $y_j \ (1 \le j \le t)$
for a valid annotation, and from the annotations on other circuits.
Let $S_{+}(C)$ (respectively $S_{-}(C)$) be the sum of
products of $w(\boldsymbol{\varrho}_v)$ over 
  $v \in \{x_i \mid 1 \le i \le s\}$, summed over valid annotations 
 such that the number of $-$ in $x_i \ (1 \le i \le s)$ is even (respectively
odd). Similarly
let $T_{+}(C)$ (respectively $T_{-}(C)$) be the corresponding
sums for $y_j \ (1 \le j \le t)$.
% the weighted sum of ways to annotate $y_j \ (1 \le j \le t)$ such that the number of $-$ in $y_j \ (1 \le j \le t)$ is even (or odd, respectively).
We have
\[S_{+}(C) - S_{-}(C) = \prod_{i=1}^s \left(w_{+}(x_i) - w_{-}(x_i)\right) \ge 0, \;~~~ T_{+}(C) - T_{-}(C) = \prod_{j=1}^t \left(w_{+}(y_j) - w_{-}(y_j)\right) \ge 0.\]
Both differences are nonnegative by the hypothesis of
\thmref{thm:property_1}.
% \thmref{thm:property_1}.

%For convenience, set $S(C,+) = 1$ and $S(C,-) = 0$ if $s = 0$; similarly set $T(C,+) = 1$ and $T(C,-) = 0$ if $t = 0$.
%%% JYC no need: an empty sum =0 (if s=0, no odd - amng x_i is possible
%%% an empty product is 1.

%The weighted sum of ways to annotate vertices along $C$ 
%is $S(C,+) T(C,+) + S(C,-) T(C,-)$. 
%%% no need to say this, it seems.

The product $S_{+}(C) T_{+}(C)$ is
the  sum over all valid annotations of vertices on $C$ such that
% (1) 
the numbers of $-$ on vertices shared by $\chi_{12}$ and $C$
and by $\chi_{34}$ and $C$ are both even.
% and (2) the number of $-$ on vertices shared by $\chi_{34}$ and $C$ is even; 
Similarly $S_{-}(C) T_{-}(C)$
is the  sum over all valid annotations of vertices on $C$ such that 
the numbers of $-$ on vertices shared by $\chi_{12}$ and $C$ 
and by $\chi_{34}$ and $C$ are both odd.
We have $S_{+}(C) T_{+}(C) \ge S_{-}(C) T_{-}(C)$.

Next we also account for the vertices shared by $\chi_{12}$ and $\chi_{34}$
in $\psi$.
Let $p$ be this number and if $p>0$ let  $z_k \ (1 \le k \le p)$ be these
vertices.
Let $q$ be the number of circuits in $\psi$, denoted 
 by $C_l \ (1 \le l \le q)$.
Then we claim that
%A valid annotation of $\psi$ is positive if (1) the number of circuits that has an odd number of $-$ on vertices it shares with $\chi_{12}$ (and $\chi_{34}$) plus (2) the number of vertices in $z_k \ (1 \le k \le p)$ that has $-$ is even.
%Therefore, we have
\[W_+ - W_- =  \prod_{k=1}^p \left(w_{+}(z_k) - w_{-}(z_k)\right) 
\prod_{l=1}^q \left(S_{+}(C_l) T_{+}(C_l) - S_{-}(C_l) T_{-}(C_l)\right),\]
and in particular $W_+ - W_-  \ge 0$.
To prove this claim we only need to expand the product,
and separately collect terms that have a $+$ sign and a $-$ sign.
In a product term in the fully expanded sum,
%In a fully expanded term as a product,
let $p'$ be the number of
$- w_{-}(z_k)$, and $q'$ be the number of 
$- S_{-}(C_l) T_{-}(C_l)$.
Then  a product term has a $+$ sign (and thus included in $W_+$)
%to have a $+$ sign (and thus included in $W_+$), a product terms is collected
iff $p' + q' \equiv 0 \pmod 2$.

Now let us deal with the case when there are self-intersecting trails or circuits.
Suppose  $v$  is a self-intersecting vertex.  Let its
four incident edges be $\{e, f, g, h\}$.
Without loss of generality we assume the pairing $\varrho_v$ in $\psi$
 is $\{e, f\}$ and $\{g, h\}$ (\figref{fig:deletion_before}).
Define $\Gamma'$ to be the 4-ary construction 
obtained from $\Gamma$ by deleting
$v$, and merging $e$ with $f$, and $g$ with $h$ (\figref{fig:deletion_1}).
Define $W'_+$ and $W'_-$ similarly for $\Gamma'$ with {\sl tcp} being $\psi'
 = \psi \setminus \{\varrho_v\}$.

Since $v$ contributes either zero or two $-$ to the trail or circuit
 it belongs to, an annotation is valid for $\psi$ iff its restriction 
on  $\Gamma'$ is valid for $\psi'$. 
% (by dropping $\varrho_v$ and the sign at it).
Moreover, for every valid annotation of vertices in $\psi'$ 
contributing a factor to $W'_+$ (or $W'_-$), 
%let  $\varrho$  be the pairing  at $v$ according to $\psi$
% and if we impose an arbitrary sign there,
if we impose an arbitrary sign on $\varrho_v$,
we get
 a valid annotation 
for $\psi$ contributing a factor to $W_+$ (or $W_-$, respectively). 
If the sign of the annotation at $v$ is $+$ (or $-$ respectively) then
each product  term in $W'_+$ or $W'_-$
 gains the same extra factor $w_+(v)$.
If  the annotation at $v$ is $-$, then 
they gain the factor  $w_-(v)$.
Therefore, we have $W_+ - W_- = (w_+(v) + w_-(v)) (W'_+ - W'_-)$.
%This means that $W_+ - W_-$ is nonnegative if $W'_+ - W'_-$ is nonnegative.
Hence $W_+ \ge W_-$ if $W'_+ \ge W'_-$.
Thus we have
reduced from $\Gamma$ to $\Gamma'$ which has one 
fewer
self-intersections. Repeating this finitely many times
we end up with no self-intersections.
%the proof is reduced to the case without completed by induction.
% Repeating this procedure finitely many times, we can reduce to a 4-ary construction with no self-intersection.

\captionsetup[subfigure]{labelformat=parens}
\begin{figure}[h!]
\centering
\begin{subfigure}[b]{0.3\linewidth}
\centering\includegraphics[width=0.5\linewidth]{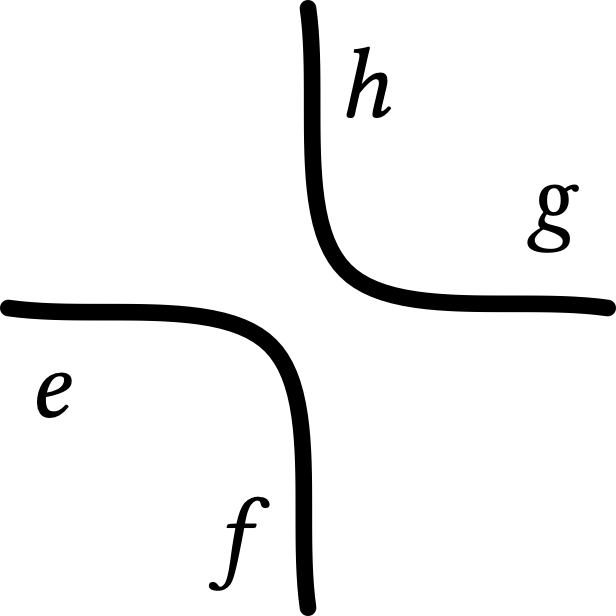}\caption{}\label{fig:deletion_before}
\end{subfigure}
\begin{subfigure}[b]{0.3\linewidth}
\centering\includegraphics[width=0.5\linewidth]{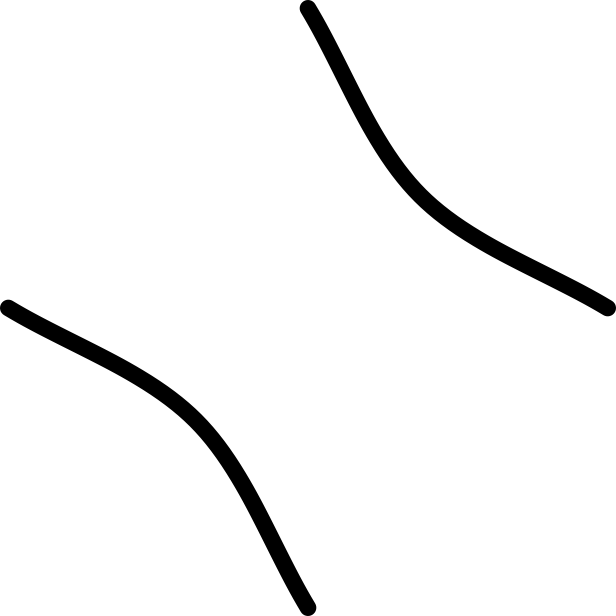}\caption{}\label{fig:deletion_1}
\end{subfigure}
\begin{subfigure}[b]{0.3\linewidth}
\centering\includegraphics[width=0.5\linewidth]{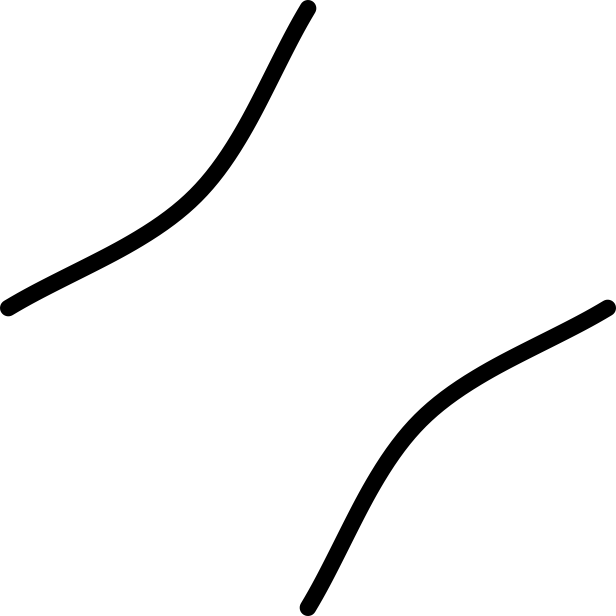}\caption{}\label{fig:deletion_2}
\end{subfigure}  
\caption{Possible ways of deleting a vertex. The vertex 
(not explicitly shown) at the center
of part (a) is removed in part (b) and (c).}\label{fig:deletion}
\end{figure}

\medskip
\noindent
\textbf{Induction step}:
Suppose $v$ is a shared vertex between two distinct circuits $C_1$ and $C_2$,
and
let $\{e, f, g, h\}$ be its incident edges in $\Gamma$.
We may assume the pairing $\varrho_v$ in $\psi$ is $\{e, f\}$ and $\{g, h\}$,
 and thus $e, f$ are in one circuit, say  $C_1$, while $g, h$
are in another circuit  $C_2$ (\figref{fig:deletion_before}).
Define $\Gamma'$ to be the 4-ary construction obtained
from $\Gamma$ by deleting $v$ and  merging $e$ with $f$, and $g$ with $h$ (\figref{fig:deletion_1}).
Define $\Gamma''$ to be the 4-ary construction obtained
from $\Gamma$ by deleting $v$ and merging $e$ with $h$, and $f$ with $g$ (\figref{fig:deletion_2}).
Note that in  $\Gamma'$, we have two circuits $C'_1$ and $C'_2$
 (each has one fewer vertex $v$ from $C_1$ and $C_2$),
 but in  $\Gamma''$ 
the two circuits are merged into one $C^*$.
Define $W'_+$ and $W'_-$ (respectively $W''_+$ and $W''_-$)
 similarly for $\Gamma'$ (respectively $\Gamma''$) 
 with {\sl tcp} being $\psi'
 = \psi \setminus \{\varrho_v\}$.

We can decompose $W_+ - W_-$ according to whether the sign on $\varrho_v$ is $+$ or $-$.
Recall that for any valid annotation of $\psi$, one encounters
an even number of $-$ along $C_1$ and $C_2$.
If the sign on $\varrho_v$ 
 is $+$, the number of $-$ along $C_1$ (and $C_2$)
at all vertices other than $v$ in any valid annotation is always even; 
if the sign on $\varrho_v$ is $-$, this number (for both $C_1$ and $C_2$) is
always odd.
$W_+ - W_-$ can be decomposed into two parts, corresponding to  terms
with  $\varrho_v$ being $+$ or  $-$ respectively.
All terms of the first (and second) part have
 a factor $w_+(v)$ (and $w_-(v)$ respectively).
And so we can write  
\begin{equation}\label{W+W--two-parts}
W_+ - W_- = w_+(v) [W_+ - W_-]_e + w_-(v) [W_+ - W_-]_o,
\end{equation}
where $[W_+ - W_-]_e$ and $[W_+ - W_-]_o$
collect terms in $W_+ - W_-$ in the first and second part respectively,
but without the factor at $v$.

However by considering valid annotations for $\Gamma'$ we also have
\begin{equation}\label{W+W-primes-even-part}
W'_+ - W'_- = [W_+ - W_-]_e,
\end{equation}
because a valid annotation  on both $C'_1$ and $C'_2$
is equivalent to
a valid annotation on both $C_1$ and $C_2$
with $v$ assigned $+$.
Similarly, by considering valid annotations for $\Gamma''$ we also have
\begin{equation}\label{W+W-primes-odd-part}
W''_+ - W''_- = [W_+ - W_-]_e + [W_+ - W_-]_o,
\end{equation}
because depending on whether $\varrho_v$ is assigned $+$ or $-$,
a valid annotation on both $C_1$ and $C_2$
gives either both an even or both an odd  number of $-$ on  $C_1 \setminus \{v\}$
and $C_2 \setminus \{v\}$, which is
equivalent to an even number of $-$ on the merged circuit $C^*$.

%Thus 
% [W_+ - W_-]_e = W'_+ - W'_-
% [W_+ - W_-]_o = (W''_+ - W''_- ) - (W'_+ - W'_-)
% W_+ - W_- = w_+(v) [ W'_+ - W'_- ] + w_-(v) [
%				(W''_+ - W''_- ) - (W'_+ - W'_-) ]
% = ( w_+(v) - w_-(v) ) [ W'_+ - W'_- ] 
%  +  w_-(v) (W''_+ - W''_- )
From (\ref{W+W--two-parts},\ref{W+W-primes-even-part},\ref{W+W-primes-odd-part})
we have
\[ W_+ - W_- =  ( w_+(v) - w_-(v) ) ( W'_+ - W'_- )
+  w_-(v) (W''_+ - W''_- ).\]
By induction, both $W'_+ \ge  W'_-$ and $W''_+  \ge W''_-$.
Since $w_+(v) \ge w_-(v)$ is given by hypothesis,
we get $W_+ \ge W_-$.
\end{proof}

%%% I got rid of this lemma 
% directly go for the proof of thm.

%\begin{lemma}\label{lem:property_2}
%Let $\Gamma$ be a 4-ary plane construction
 %in the eight-vertex model. If the constraint function on every vertex in $\Gamma$ satisfies
%$\left\{\begin{smallmatrix}
%w(\subinmatrix{1}_+) \ge w(\subinmatrix{1}_-)\\
%w(\subinmatrix{2}_+) \ge w(\subinmatrix{2}_-)\\
%w(\subinmatrix{3}_+) \le w(\subinmatrix{3}_-)\\
%\end{smallmatrix}\right.$
%where $w$ is nonnegative,
%then the constrain function $f$ defined by $\Gamma$ satisfies
%$\left\{\begin{smallmatrix}
%W(\subinmatrix{1}_+) \ge W(\subinmatrix{1}_-)\\
%W(\subinmatrix{2}_+) \ge W(\subinmatrix{2}_-)\\
%W(\subinmatrix{3}_+) \le W(\subinmatrix{3}_-)\\
%\end{smallmatrix}\right.$.
%\end{lemma}
%\begin{proof}

\begin{proof}[Proof of \thmref{thm:property_2}]
By Lemma~\ref{lem:freedom_1}, for $(a,b,c,d) \in \overline{\mathcal{F}_>}$
we can choose a nonnegative function $w$ to satisfy (\ref{eqn:weight}).%($*$).
It is easily verified that
for any weight function $w$ satisfying (\ref{eqn:weight}), %($*$),
$w(\sub{1}_+) \ge w(\sub{1}_-)$ iff $a+d \le b+c$
(in $\mathcal{A}_\le$),
$w(\sub{2}_+) \ge w(\sub{2}_-)$ iff $b+d \le a+c$
(in $\mathcal{B}_\le$),
$w(\sub{3}_+) \le w(\sub{3}_-)$ iff $c+d \ge a+b$
(in $\mathcal{C}_\ge$).
Since $(a,b,c,d) \in \mathcal{A}_\le \bigcap \mathcal{B}_\le \bigcap \mathcal{C}_\ge \bigcap \overline{\mathcal{F}_>}$
by hypothesis, 
we have a nonnegative function $w$ satisfying (\ref{eqn:weight}) %($*$)
and 
$\left\{\begin{smallmatrix}
w(\subinmatrix{1}_+) \ge w(\subinmatrix{1}_-)\\
w(\subinmatrix{2}_+) \ge w(\subinmatrix{2}_-)\\
w(\subinmatrix{3}_+) \le w(\subinmatrix{3}_-)\\
\end{smallmatrix}\right.$.

We say a {\sl tcp} $\psi$ of a  4-ary construction
has type-$\sub{1}$ if its two trails connect
dangling edges $e_1$ with $e_2$ and $e_3$ with $e_4$,
type-$\sub{2}$ if they connect
$e_1$ with $e_4$ and $e_2$ with $e_3$,
and type-$\sub{3}$  if they connect
$e_1$ with $e_3$ and $e_2$ with $e_4$.
Sometimes we also say a pairing $\varrho \in \{\sub{1}, \sub{2}, \sub{3}\}$ (without a sign) has type-$\varrho$.

We prove this theorem not only for 4-ary plane constructions, but for any 4-ary construction $\Gamma$ that satisfies the following property
\ref{cond:plane}.
%which is clearly satisfied by all  4-ary plane constructions.
\begin{equation} \label{cond:plane}
\text{\parbox{.9\textwidth}{For any {\sl tcp} $\psi$ of $\Gamma$ the number of vertices that have type-$\sub{3}$ pairings shared: (1) by any two distinct circuits is even; (2) by a trail and a circuit is even; (3) by two trails 
is even,
 if $\psi$  has type-$\sub{1}$
  or type-$\sub{2}$;
 and (4) by two trails is odd,
 if $\psi$  has  type-$\sub{3}$.}} \tag{$\mathcal{P}$}
\end{equation}
Observe that every 4-ary plane construction satisfies property \ref{cond:plane} by \emph{Jordan Curve Theorem}.

The structure of this proof is similar to that of the 
proof of \thmref{thm:property_1}, but the details are more delicate
because of the reversed inequality 
$w(\sub{3}_+) \le w(\sub{3}_-)$, which we need to use 
 property \ref{cond:plane} and a parity argument  to finesse.

Inheriting notations from the proof of \thmref{thm:property_1},
we prove that for any {\sl tcp} $\psi \in \Psi$, $W_+ \ge W_-$ if $\psi$ 
has type-$\sub{1}$ or type-$\sub{2}$; and $W_+ \le W_-$ if $\psi$ has type-$\sub{3}$.
We prove this statement still by induction on the number $N$ of 
vertices shared by any two distinct circuits in $\psi$.

\medskip
\noindent
\textbf{Base case}: The base case is $N=0$. 
Let us first assume that no trail or circuit is self-intersecting.
Consider the case $\psi$ has type-$\sub{1}$.
Then every vertex on any circuit $C$ of $\psi$
is shared by $C$ and exactly one trail, $\chi_{12}$ or $\chi_{34}$.
Also, every vertex on $\chi_{12}$ or $\chi_{34}$ is shared with some circuit
or the other trail.

For a circuit $C \in \psi$, by property \ref{cond:plane}
the number of vertices it shares with a trail that have
 a type-$\sub{3}$ pairing is even.
Denote the number of vertices it shares with $\chi_{12}$ that have a type-$\sub{1}$ or type-$\sub{2}$ pairing by $s$ and those that have a type-$\sub{3}$ pairing by $s'$,  and let the vertices be
$x_i \ (1 \le i \le s)$ and  $x'_i \ (1 \le i \le s')$ respectively; 
similarly denote the number of vertices it shares with $\chi_{34}$ that have a type-$\sub{1}$ or type-$\sub{2}$ pairing by $t$ and 
those that have a type-$\sub{3}$ pairing by $t'$, 
and let the vertices be
$y_j \ (1 \le j \le t)$ and $y'_j \ (1 \le j \le t')$
respectively (the following statement is still true if
there is any zero among $s,s', t, t'$).  Define the quantities
$S_+(C), S_-(C), T_+(C)$ and $T_-(C)$
as in the proof of \thmref{thm:property_1},
%and by a  similar argument,
then we have
\[S_+(C) - S_-(C) = \prod_{i=1}^{s'} \left(w_+(x'_i) - w_-(x'_i)\right) \prod_{i=1}^s \left(w_+(x_i) - w_-(x_i)\right),\]
% \ge 0, \]
\[T_+(C) - T_-(C) = \prod_{j=1}^{t'} \left(w_+(y'_j) - w_-(y'_j)\right) \prod_{j=1}^t \left(w_+(y_j) - w_-(y_j)\right).\]
% \ge 0.\]
We have $S_+(C) \ge  S_-(C)$ because each $w_+(x_i) - w_-(x_i) \ge 0$
and each $w_+(x'_i) - w_-(x'_i) \le 0$ but  $s'$ is even.
By the same argument, $T_+(C) \ge T_-(C)$. 
%Both differences are nonnegative because
% due to the fact that $s'$ and $t'$ are both even.

Now we account for the shared vertices between the two trails.
According to  property \ref{cond:plane},
the number of vertices shared by $\chi_{12}$ and $\chi_{34}$ that have a type-$\sub{3}$ pairing must also be even.
Denote the number of vertices in $\psi$ shared by $\chi_{12}$ and $\chi_{34}$ that have a type-$\sub{1}$ or type-$\sub{2}$ pairing by $p$ and those that have a type-$\sub{3}$ pairing by $p'$,
 and  denote these vertices by $z_k \ (1 \le k \le p)$ and 
 $z'_k \ (1 \le k \le p')$ (again the following is still true if
 $p$ or $p'$ is 0).
Then, by the same proof,
\[W_+ - W_- =  \prod_{k=1}^{p'} \left(w_+(z'_k) - w_-(z'_k)\right) \prod_{k=1}^p \left(w_+(z_k) - w_-(z_k)\right) \prod_{l=1}^q \left(S_+(C_l) T_+(C_l) - S_-(C_l) T_-(C_l)\right).\]
Since $S_+(C_l) T_+(C_l) \ge S_-(C_l) T_-(C_l)$  ($1 \le l \le q$),
$w_+(z_k) - w_-(z_k) \ge 0$ ($ 1 \le k \le p$),
and $w_+(z'_k) - w_-(z'_k) \le 0$ ($ 1 \le k \le p'$ and $p'$ is even),
we get $W_+ \ge W_-$. 

The same proof applies for $\psi$ of  type-$\sub{2}$.
For type-$\sub{3}$ we have the corresponding $p'$ odd,
and thus $W_+ \le W_-$.

The way to deal with self-intersections is exactly the same as in the proof of \thmref{thm:property_1}.
We will not repeat here.

\medskip
\noindent
\textbf{Induction step}:
When the pairings at intersections between distinct circuits are all of type-$\sub{1}$ or type-$\sub{2}$ only, our proof is the same as the induction step of the proof of \thmref{thm:property_1}. We only note that the constructions of 
 $\Gamma'$ and  $\Gamma''$ in that proof preserve  property \ref{cond:plane}.
%%%
% in Gamma'', self- intersection (shared) (of type 1 , 2 only since we assumed
% no type 3 between C1 C2.) may have been created.
% all 4 (actually 3) kinds in property P, are unchanged.
% the 1st is actually 0, as we assume there are no type 3 between 2 circuits.
%
%since no type-$\sub{3}$ shared vertices are created.
%
%
 %When removing vertices and merging the incident edges, we can make sure that  property \ref{cond:plane} is always preserved.
%This process can introduce self-intersections of type-$\sub{1}$ or type-$\sub{2}$ on circuits because sometimes two circuits are merged into one, which can be dealt with as in the proof of \thmref{thm:property_1}.
%However, it will never introduce self-intersections of type-$\sub{3}$ so that we can use the inductive hypothesis..
%
%
When there are type-$\sub{3}$ intersections between distinct circuits, we show how to reduce to the previous case by getting rid of all type-$\sub{3}$ intersections while preserving  property \ref{cond:plane}.
%Imagine $C_1$ starts from an edge ``outside'' $C_2$. Once $C_1$ goes into $C_2$ at a vertex $u$, it has to go out of $C_2$ at another vertex $v$ before it closes at the starting edge or goes into $C_2$ again.
%Therefore, such $u$ and $v$ appear in pairs.

\captionsetup[subfigure]{labelformat=parens}
\begin{figure}[h!]
\centering
\begin{subfigure}[b]{0.3\linewidth}
\centering\includegraphics[width=0.7\linewidth]{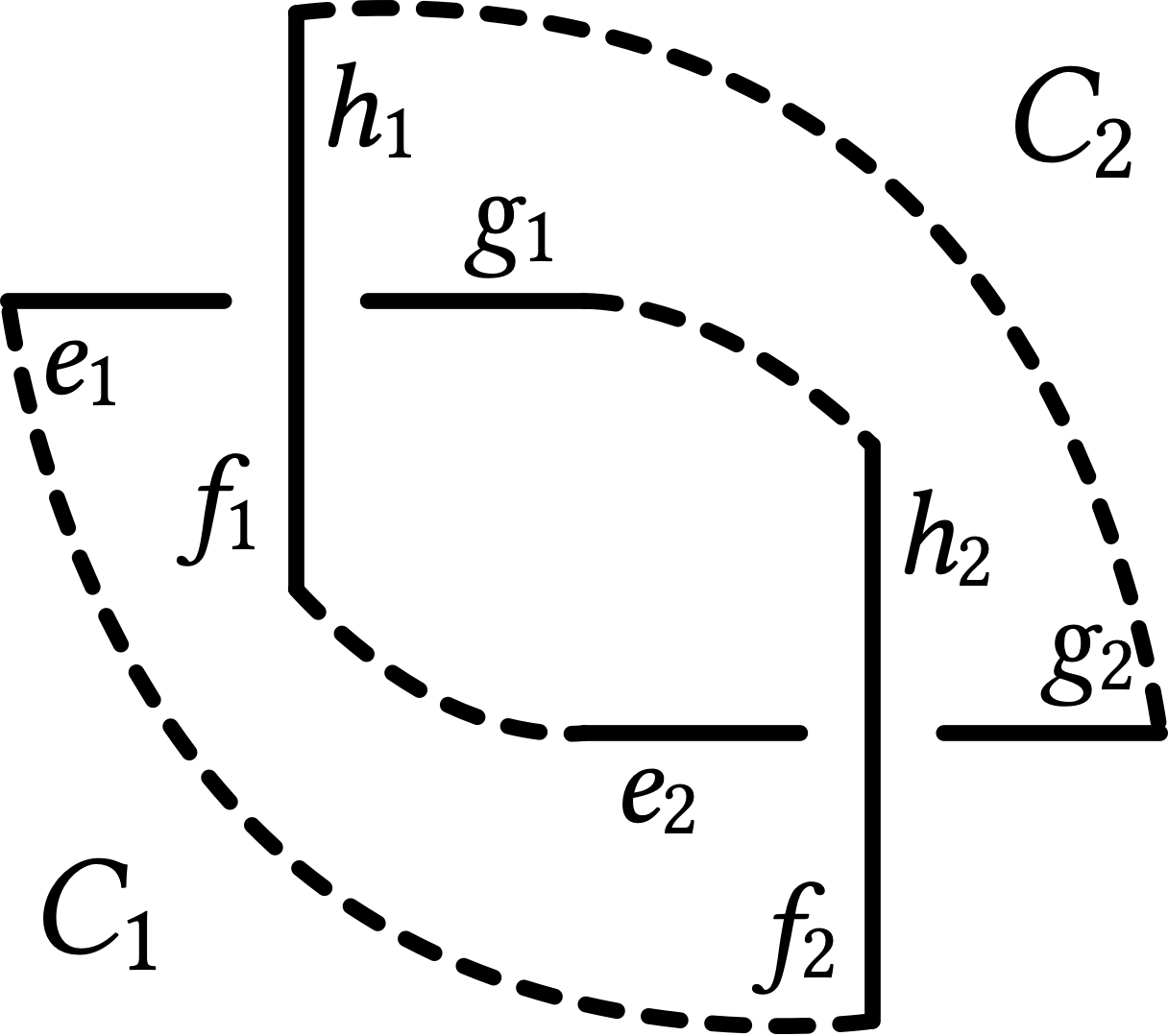}\caption{$\Gamma$}\label{fig:deletion_planar_before}
\end{subfigure}
\begin{subfigure}[b]{0.3\linewidth}
\centering\includegraphics[width=0.7\linewidth]{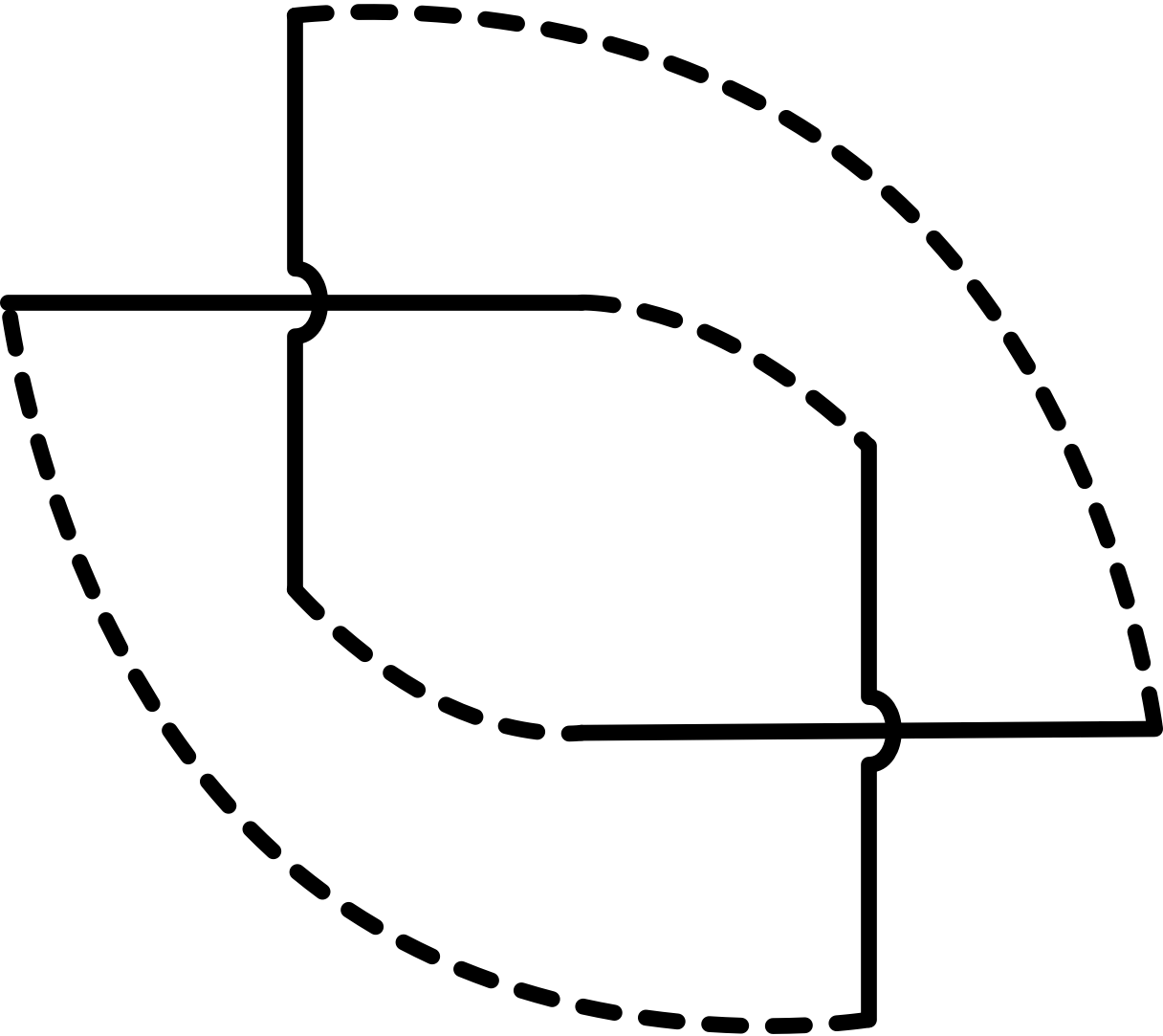}\caption{$\Gamma'$}\label{fig:deletion_planar_11}
\end{subfigure}
\begin{subfigure}[b]{0.3\linewidth}
\centering\includegraphics[width=0.7\linewidth]{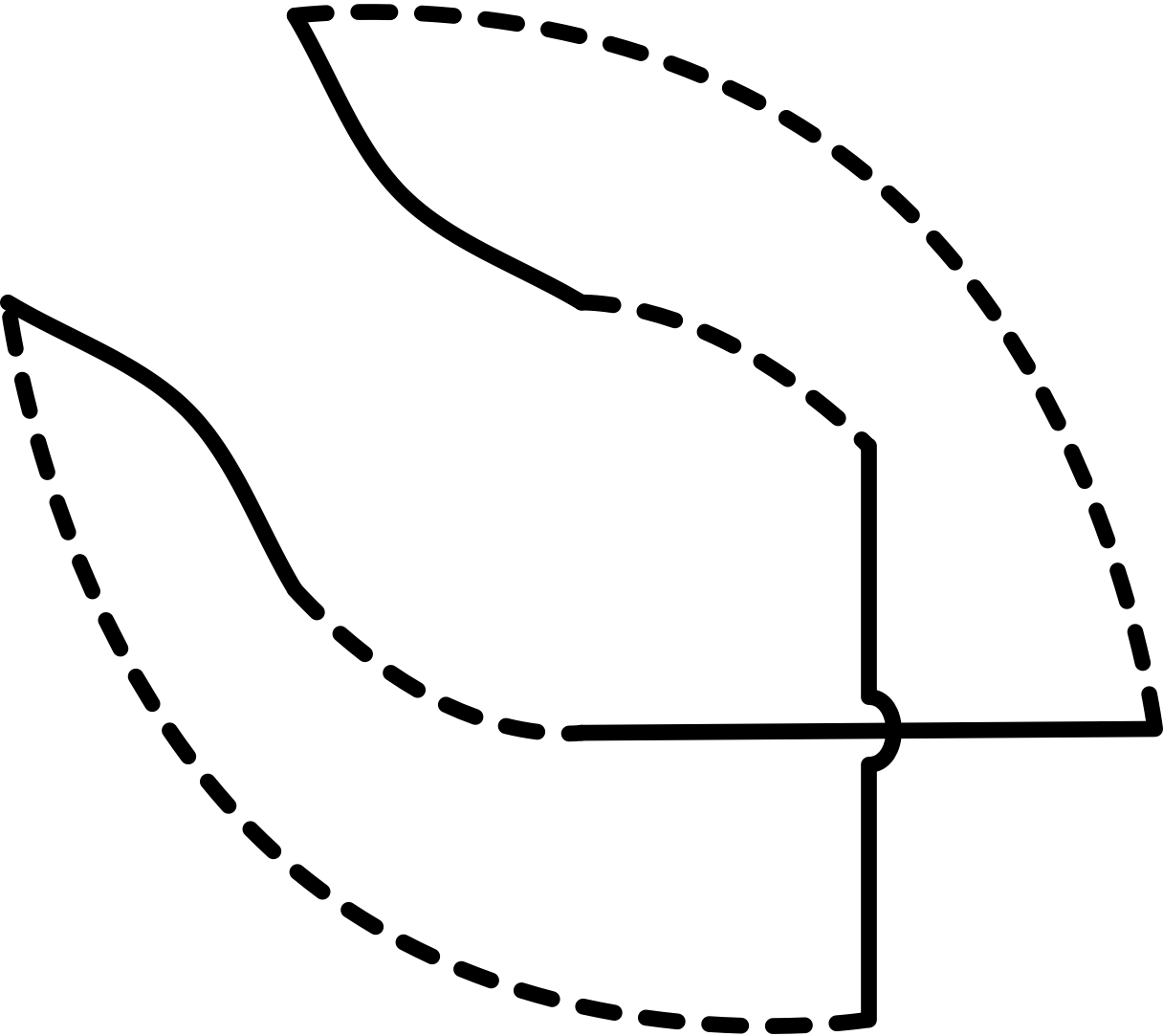}\caption{$\Gamma''$}\label{fig:deletion_planar_21}
\end{subfigure}  
\caption{}\label{fig:deletion}
\end{figure}

For any two circuits $C_1$ and $C_2$ in $\psi$, the number of intersections of type-$\sub{3}$ between these two circuits must be even (according to  property \ref{cond:plane}).
Suppose this number is not zero, let $u$ and $v$ be two vertices with intersections of type-$\sub{3}$ between $C_1$ and $C_2$.
For $b \in \{1, 2\}$, let 
$\{e_b, f_b, g_b, h_b\}$ be the edges incident to $u$ and $v$ in $\Gamma$
respectively.
We may assume the pairings at $u$ and $v$ are $\{e_b, g_b\}$ and $\{f_b, h_b\}$ 
 and thus $e_b, g_b$ are in one circuit, while $f_b, h_b$
are in another circuit.
Futhermore, we may name the edges
so that $e_1, g_1$ and $f_2, h_2$ are in the same circuit, say $C_1$,
and $e_2, g_2$ and $f_1, h_1$ are in another circuit (in this case $C_2$) (\figref{fig:deletion_planar_before}).
Define $\Gamma'$ to be the 4-ary construction obtained
from $\Gamma$ by deleting $u, v$ and merging $e_b$ with $g_b$, and $f_b$ with $h_b$ for $b \in \{1, 2\}$ (\figref{fig:deletion_planar_11}).
Define $\Gamma''$ to be the 4-ary construction obtained
from $\Gamma$ by deleting $u, v$ and merging $e_1$ with $f_1$, $g_1$ with $h_1$, $e_2$ with $g_2$, and $f_2$ with $h_2$ (\figref{fig:deletion_planar_21}).
%%% JYC we may want to draw a picture here
Note that in  $\Gamma'$, we have two circuits $C'_1$ and $C'_2$
 (each has two fewer vertices $u$ and $v$ from $C_1$ and $C_2$),
 but in  $\Gamma''$ 
the two circuits are merged into one $C^*$.
Define $W'_+$ and $W'_-$ (respectively $W''_+$ and $W''_-$)
 similarly for $\Gamma'$ (respectively $\Gamma''$) 
 with {\sl tcp} being $\psi'
 = \psi \setminus \{\varrho_{u}, \varrho_{v}\}$.

We can decompose $W_+ - W_-$ according to whether the signs on $\varrho_{u}$ 
and $\varrho_{v}$ are $+$ or $-$.
Recall that for any valid annotation of $\psi$, one encounters
an even number of $-$ along $C_1$ and $C_2$.
If the signs on $\varrho_{u}$ and $\varrho_{v}$
are both $+$ or both $-$, the number of $-$ along $C_1$ (and $C_2$)
at all vertices other than $u$ and $v$ in any valid annotation is always even; 
if the signs on $\varrho_{u}$ and $\varrho_{v}$ are different (one $+$ and one $-$), this number (for both $C_1$ and $C_2$) is
always odd.
$W_+ - W_-$ can be decomposed into four parts, corresponding to terms
with the signs on $\varrho_{u}$ and $\varrho_{v}$ being $\pm$.
% being
% the same or different respectively.
%All terms of the first and second part have
 %a factor $w_+(u)w_+(v)$ or  $w_-(u)w_-(v)$
 %(and $w_+(u)w_-(v)$ or $ w_-(u)w_+(v)$ respectively).
%And so we can write  
So we can write
\begin{eqnarray}\label{eqn:W-inplanarproof}
W_+ - W_- &=& \,\/~~w_+(u)w_+(v) [W_+ - W_-]_{++} + w_-(u)w_-(v) [W_+ - W_-]_{--} \\
&&+ w_+(u)w_-(v) [W_+ - W_-]_{+-} + w_-(u)w_+(v) [W_+ - W_-]_{-+},
\end{eqnarray}
where $[W_+ - W_-]_{\pm\pm}$ 
collect terms in $W_+ - W_-$ in the respective parts,
but without the factors at $u$ and $v$.

Let $X$ (respectively $X'$) be the set of vertices of $C_1$ (excluding $u,v$)
between $e_1$ and $f_2$ (respectively between $g_1$ and $h_2$).
Let $Y$ (respectively $Y'$) be the set of vertices of $C_2$ (excluding $u,v$)
between $h_1$ and $g_2$ (respectively between $f_1$ and $e_2$).
If we write $\sigma(x) =1$ if the annotation on $x$ is $-$,
and $\sigma(x) =0$ otherwise, then the requirement
for an annotation on $C'_1$ and $C'_2$ to be valid
is $\sum_{x \in X \cup X'} \sigma(x) \equiv
\sum_{x \in Y \cup Y'} \sigma(x) \equiv 0 \pmod 2$.
This is equivalent to requiring
an extension 
that assigns the same sign to both $u$ and $v$ (either $(++)$ or $(--)$)
 to be a valid  annotation on $C_1$ and $C_2$.
The latter  is just $\sum_{x \in X \cup X' \cup \{u, v\}} \sigma(x) \equiv
\sum_{x \in Y \cup Y'  \cup \{u, v\}} \sigma(x) \equiv 0 \pmod 2$,
conditioned on  $\sigma(u) = \sigma(v)$.
Hence
\begin{equation}\label{eqn:planar-Wprime}
W'_+ - W'_- = [W_+ - W_-]_{++} = [W_+ - W_-]_{--}
\end{equation}

The requirement for an  annotation  to be valid on $C^*$
is $\sum_{x \in X \cup X' \cup Y \cup Y'} \sigma(x) \equiv 0 \pmod 2$,
which is equivalent to
either $\sum_{x \in X \cup X'} \sigma(x)  \equiv 
\sum_{x \in Y \cup Y'} \sigma(x) \equiv 0 \pmod 2$,
or
 $\sum_{x \in X \cup X'} \sigma(x)  \equiv 
\sum_{x \in Y \cup Y'} \sigma(x) \equiv 1 \pmod 2$.
This is equivalent to combining two types of
extensions to a valid  annotation on $C_1$ and $C_2$,
where type (1) assigns the same sign to both 
$u$ and $v$ (either $(++)$ or $(--)$),
or type (2) assigns different signs to 
$u$ and $v$ (either $(+-)$ or $(-+)$).
Hence, in addition to (\ref{eqn:planar-Wprime}) we have
\begin{equation}\label{eqn:W+-issameas-+}
[W_+ - W_-]_{+-} = [W_+ - W_-]_{-+},
\end{equation}
and also
\begin{equation}\label{eqn:planar-Wprimeprime}
W''_+ - W''_- = [W_+ - W_-]_{++} + [W_+ - W_-]_{+-}.
\end{equation}

%However by considering valid annotations for $\Gamma'$ we also have
%\[W'_+ - W'_- = [W_+ - W_-]_{++} = [W_+ - W_-]_{--},\]
%because a valid annotation  on both $C'_1$ and $C'_2$
%is equivalent to
%a valid annotation on both $C_1$ and $C_2$
%with $u$ and $v$ having the same sign.
%Similarly, by considering valid annotations for $\Gamma''$ we have

It follows from (\ref{eqn:W-inplanarproof}, \ref{eqn:W+-issameas-+},
 \ref{eqn:planar-Wprime},
\ref{eqn:planar-Wprimeprime}) that
\begin{eqnarray*}
&& W_+ - W_-\\
 &= & (w_+(u) w_+(v) + w_-(u) w_-(v)) [W_+ - W_-]_{++}
+ (w_+(u) w_-(v) + w_-(u) w_+(v) ) [W_+ - W_-]_{+-}\\
 &= &  
 [(w_+(u) - w_-(u)) (w_+(v) - w_-(v))] ( W'_+ - W'_- )
+ [w_+(u)w_-(v) + w_-(u)w_+(v)] (W''_+ - W''_- ).
\end{eqnarray*}

%W''_+ - W''_- = [W_+ - W_-]_{++} + [W_+ - W_-]_{--}  
%%+ [W_+ - W_-]_{+-} + [W_+ - W_-]_{-+},\]
%% [W_+ - W_-]_e + [W_+ - W_-]_o,\]
%because there is a 1-4 correspondence
%between all valid annotations  on the merged circuit $C^*$
%and all valid annotations on $C_1$ and $C_2$, extending 
%the signs at $u$ and  $v$ in all four ways.
%%have exactly four extensions are in 1-1 correspregardless which of the four cases happens, $u$ has $+$ and $v$ has $-$ or $u$ has $-$ and $v$ has $+$,
%%a valid annotation on both $C_1$ and $C_2$
%%gives either both an even or both an odd  number of $-$ on  $C_1 \setminus \{u, v\}$
%%and $C_2 \setminus \{u, v\}$, which is
%%equivalent to an even number of $-$ on the merged circuit $C^*$.
%Thus 
%\[ W_+ - W_- =  (w_+(u) - w_-(u)) (w_+(v) - w_-(v)) ( W'_+ - W'_- )
%+ (w_+(u)w_-(v) + w_-(u)w_+(v)) (W''_+ - W''_- ).\]

Note that if $\Gamma$ satisfies  property \ref{cond:plane}, $\Gamma'$ and $\Gamma''$ also satisfy  property \ref{cond:plane}, but with fewer 
intersections of type-$\sub{3}$ between distinct circuits.
By induction, both $W'_+ - W'_- \ge 0$ and $W''_+ - W''_- \ge 0$.
Since $w_+(u) \le w_-(u)$ and $w_+(v) \le w_-(v)$ are given by hypothesis,
the product $(w_+(u) - w_-(u)) (w_+(v) - w_-(v)) \ge 0$.
Also $w_+(u)w_-(v) + w_-(u)w_+(v) \ge 0$ as $w$ is nonnegative.
Therefore, $W_+ \ge  W_-$.

%Note that if $\Gamma$ satisfies  property \ref{cond:plane}, $\Gamma'$ and $\Gamma''$ also satisfies  property \ref{cond:plane}.
%In this way, we reduce from $\Gamma$ to $\Gamma'$ which has two 
%fewer intersections of type-$\sub{3}$ between $C_1$ and $C_2$ and to $\Gamma''$ which has $C_1$ and $C_2$ become a single circuit. Repeating this process, we can reduce to a 4-ary construction with no intersections of type-$\sub{3}$ between distinct circuits.

\medskip

We have finished the proof for $W(\sub{1}_+) \ge W(\sub{1}_-)$.
The proof for $W(\sub{2}_+) \ge W(\sub{2}_-)$ is the same.
The proof for $W(\sub{3}_+) \le W(\sub{3}_-)$ can be adapted. 
%%% I already talked about the base case in Base Case proof.
%The only difference is in the base case where the number of vertices shared by $\chi_{13}$ and $\chi_{24}$ that have a type-$\sub{3}$ pairing must be \emph{odd} 
%(according to  property \ref{cond:plane}), and  the two trails
We only need to note that  the two trails in $\Gamma'$ and $\Gamma''$ 
are unchanged  from $\Gamma$, thus both are still of type-$\sub{3}$.
Thus inductively we have $W'_+ - W'_- \le 0$ and $W''_+ - W''_- \le 0$,
 and thus $W_+  - W_- \le 0$ as  a nonnegative combination of these two
quantities.
\end{proof}

\corref{cor:property_4} follows immediately from
\thmref{thm:property_1} and \thmref{thm:property_2}
since $\mathcal{C}_= = \mathcal{C}_\le \bigcap \mathcal{C}_\ge$
and $\mathcal{A}_\le \bigcap \mathcal{B}_\le \bigcap \mathcal{C}_\le
\subset \overline{\mathcal{F}_>}$, and therefore
the intersection of the two regions
$\mathcal{A}_\le \bigcap \mathcal{B}_\le \bigcap \mathcal{C}_\le$
and $\mathcal{A}_\le \bigcap \mathcal{B}_\le \bigcap \mathcal{C}_\ge \bigcap \overline{\mathcal{F}_>}$
is precisely 
$\mathcal{A}_\le \bigcap \mathcal{B}_\le \bigcap \mathcal{C}_=$.

\begin{lemma}\label{lem:freedom_0}
Suppose $x, x', y, y', z, z' \in \mathbb{R}$ satisfy
the eight inequalities: $X + Y + Z \ge 0$ where
$X \in \{x, x'\}, Y \in \{y, y'\}, Z \in \{z, z'\}$.
Then there exist \emph{nonnegative} 
$\tilde{x}, \tilde{x}', \tilde{y}, \tilde{y}', \tilde{z},
 \tilde{z}'$ such that all eight sums $X + Y + Z$ are unchanged 
when $x, x', y, y', z, z'$ are substituted by the respective values
$\tilde{x}, \tilde{x}', \tilde{y}, \tilde{y}', \tilde{z},
 \tilde{z}'$.
\end{lemma}
\begin{proof}
The condition is obviously symmetric so that there is a symmetry
group $S_2 \times S_2 \times S_2$ acting on $\{x, x', y, y', z, z'\}$.
Thus, we may assume without loss of generality that
$x \le x', y \le y', z \le z'$. 
Let $\alpha, \beta$ be
two distinct symbols among $x, y, z$. 
For any $c \in \mathbb{R}$,
if we add $c$ to $\alpha$ and $\alpha'$,
and subtract $c$ from $\beta$ and $\beta'$,
%(and leave the third symbol unchanged)
the  eight sums $X + Y + Z$ are unchanged,
because in each $X + Y + Z$  exactly one of $\alpha$ and $\alpha'$ appears once
and also exactly one of $\beta$ and $\beta'$ appears once.

Note that $x+y+z \ge 0$. In two steps we can replace
$\{x, x', y, y', z, z'\}$ by
\begin{eqnarray*}
x \rightarrow x+ y+z, && x' \rightarrow x'+ y+z\\
y \rightarrow y-y =0, && y' \rightarrow y'-y\\
z \rightarrow z-z =0, && z' \rightarrow z'-z
\end{eqnarray*}
This completes the proof.
\end{proof}

\begin{lemma}\label{lem:freedom_1}
The parameter setting
$(a, b, c, d)$ belongs to $\overline{\mathcal{F}_>}$ iff there exists
a nonnegative weight function $w$ satisfying (\ref{eqn:weight}). % ($*$).
\end{lemma}
\begin{proof}
The assignment of the weight function $w$ satisfying (\ref{eqn:weight}) %($*$)
can be viewed as a linear system on  six variables
$(x, x', y, y', z, z') = 
(w(\sub{1}_+), w(\sub{1}_-),  w(\sub{2}_+),
 w(\sub{2}_-), w(\sub{3}_+), w(\sub{3}_-))$.
This linear system has rank 4 and therefore there is a nonempty
solution space of dimension 2.

Pick any solution to (\ref{eqn:weight}). %($*$).
Recall that  membership
$(a, b, c, d) \in \overline{\mathcal{F}_>}$ is characterized by
$\left\{\begin{smallmatrix}
w(\subinmatrix{1}_+) + w(\subinmatrix{2}_-) + w(\subinmatrix{3}_-) \ge 0\\
w(\subinmatrix{1}_-) + w(\subinmatrix{2}_+) + w(\subinmatrix{3}_-) \ge 0\\
w(\subinmatrix{1}_-) + w(\subinmatrix{2}_-) + w(\subinmatrix{3}_+) \ge 0\\
w(\subinmatrix{1}_+) + w(\subinmatrix{2}_+) + w(\subinmatrix{3}_+) \ge 0\\
\end{smallmatrix}\right.$.
We also have $a,b,c,d \ge 0$.
Hence we can apply  \lemref{lem:freedom_0}
and get a nonnegative valued $w$
satisfying (\ref{eqn:weight}). %($*$).

The reverse direction is obvious, once we realize that
membership in $\overline{\mathcal{F}_>}$ is characterized
by the four inequalities above, for any solution to (\ref{eqn:weight}).
\end{proof}

\begin{notation}
Fix for each vertex $v$ in a 4-regular graph  $G$ a 
 weight function $w$ on signed pairings (satisfying  (\ref{eqn:weight}) 
at $v$).
% w may depend on the actual signature at v.
%Given a 4-regular graph $G$ and a weight function $w$ on signed pairings on each of the vertex in $G$.
Let $Z_v(\boldsymbol{\varrho})$ be the weighted sum of the set of all $dacp$ having
the signed pairing $\boldsymbol{\varrho}$ at $v$.
\end{notation}

\begin{corollary} \label{cor:proportion}
If at each vertex in a 4-regular graph $G$ we have
 a nonnegative weight function $w$ such that $ w\left(\sub{1}_+\right) \ge w\left(\sub{1}_-\right)$, $ w\left(\sub{2}_+\right) \ge w\left(\sub{2}_-\right)$, and $ w\left(\sub{3}_+\right) \ge w\left(\sub{3}_-\right)$, then $Z_v\left(\sub{1}_+\right) \ge Z_v\left(\sub{1}_-\right)$, $Z_v\left(\sub{2}_+\right) \ge Z_v\left(\sub{2}_-\right)$, and $Z_v\left(\sub{3}_+\right) \ge Z_v\left(\sub{3}_-\right)$ at  each vertex $v$ in $G$.
\end{corollary}
\begin{proof}
Let  $(a, b, c, d)$ be the parameters of the constraint
function at a vertex $v$.
We first collect terms in the partition function  $Z(G)$  according to which
of the 8 local configurations in \figref{fig:orientations} 
$v$ is in. If
we remove the vertex $v$ from $G$, the rest of $G$ forms  
 a 4-ary construction whose dangling edges are those incident to $v$. Using
notations for 4-ary constructions, we can write $Z(G)$ as
{\footnotesize
\begin{align*}
& 2a \left[W\left(\Phi_{0011,\sub{1}_-}\right) + W\left(\Phi_{0011,\sub{2}_+}\right) + W\left(\Phi_{0011,\sub{3}_+}\right)\right]
+ 2b \left[W\left(\Phi_{0110,\sub{1}_+}\right) + W\left(\Phi_{0110,\sub{2}_-}\right) + W\left(\Phi_{0110,\sub{3}_+}\right)\right] + \\
&  2c \left[W\left(\Phi_{0101,\sub{1}_+}\right) + W\left(\Phi_{0101,\sub{2}_+}\right) + W\left(\Phi_{0101,\sub{3}_-}\right)\right]
+ 2d \left[W\left(\Phi_{0000,\sub{1}_-}\right) + W\left(\Phi_{0000,\sub{2}_-}\right) + W\left(\Phi_{0000,\sub{3}_-}\right)\right] \\
= & 2\left(w(\sub{1}_-) + w(\sub{2}_+) + w(\sub{3}_+)\right)\left[W(\sub{1}_-) + W(\sub{2}_+) + W(\sub{3}_+)\right] 
 + 2\left(w(\sub{1}_+) + w(\sub{2}_-) + w(\sub{3}_+)\right)\left[W(\sub{1}_+) + W(\sub{2}_-) + W(\sub{3}_+)\right] + \\
&  2\left(w(\sub{1}_+) + w(\sub{2}_+) + w(\sub{3}_-)\right)\left[W(\sub{1}_+) + W(\sub{2}_+) + W(\sub{3}_-)\right] 
 + 2\left(w(\sub{1}_-) + w(\sub{2}_-) + w(\sub{3}_-)\right)\left[W(\sub{1}_-) + W(\sub{2}_-) + W(\sub{3}_-)\right].
\end{align*}}
Now we collect terms according to the 6 signed
pairings $\boldsymbol{\varrho}$ at $v$.
 These are precisely $Z_v\left(\sub{1}_+\right)$, $Z_v\left(\sub{1}_-\right)$, $Z_v\left(\sub{2}_+\right)$, $Z_v\left(\sub{2}_-\right)$, $Z_v\left(\sub{3}_+\right)$, and $Z_v\left(\sub{3}_-\right)$
respectively, and $Z(G)$ is the sum of these 6 terms
{\footnotesize
\begin{align*}
 & 2w(\sub{1}_+) \left[2 W(\sub{1}_+) + W(\sub{2}_+) + W(\sub{2}_-) + W(\sub{3}_+) + W(\sub{3}_-)\right] 
+ 2w(\sub{1}_-) \left[2 W(\sub{1}_-) + W(\sub{2}_+) + W(\sub{2}_-) + W(\sub{3}_+) + W(\sub{3}_-)\right] +\\
&  2w(\sub{2}_+) \left[2 W(\sub{2}_+) + W(\sub{1}_+) + W(\sub{1}_-) + W(\sub{3}_+) + W(\sub{3}_-)\right] 
 + 2w(\sub{2}_-) \left[2 W(\sub{2}_-) + W(\sub{1}_+) + W(\sub{1}_-) + W(\sub{3}_+) + W(\sub{3}_-)\right] +\\
&  2w(\sub{3}_+) \left[2 W(\sub{3}_+) + W(\sub{1}_+) + W(\sub{1}_-) + W(\sub{2}_+) + W(\sub{2}_-)\right] 
 + 2w(\sub{3}_-) \left[2 W(\sub{3}_-) + W(\sub{1}_+) + W(\sub{1}_-) + W(\sub{2}_+) + W(\sub{2}_-)\right].
\end{align*}}
%where the six terms are $Z_v\left(\sub{1}_+\right)$, $Z_v\left(\sub{1}_-\right)$, $Z_v\left(\sub{2}_+\right)$, $Z_v\left(\sub{2}_-\right)$, $Z_v\left(\sub{3}_+\right)$, and $Z_v\left(\sub{3}_-\right)$, respectively. 
Notice the common multipliers  when comparing
the three pairs $Z_v\left(\sub{1}_+\right)$ vs.
$Z_v\left(\sub{1}_-\right)$,
$Z_v\left(\sub{2}_+\right)$ vs. $Z_v\left(\sub{2}_-\right)$, and 
$Z_v\left(\sub{3}_+\right)$ vs. $Z_v\left(\sub{3}_-\right)$.
The corollary follows because $ w\left(\sub{1}_+\right) \ge w\left(\sub{1}_-\right)$, $ w\left(\sub{2}_+\right) \ge w\left(\sub{2}_-\right)$, and $ w\left(\sub{3}_+\right) \ge w\left(\sub{3}_-\right)$ by the assumption, and $ W\left(\sub{1}_+\right) \ge W\left(\sub{1}_-\right)$, $ W\left(\sub{2}_+\right) \ge W\left(\sub{2}_-\right)$, and $ W\left(\sub{3}_+\right) \ge W\left(\sub{3}_-\right)$ by \thmref{thm:property_1}.
\end{proof}

\begin{corollary} \label{cor:proportion_planar}
If at each vertex in a 4-regular \emph{plane} graph $G$ we have
 a nonnegative weight function $w$ such that   $ w\left(\sub{1}_+\right) \ge w\left(\sub{1}_-\right)$, $ w\left(\sub{2}_+\right) \ge w\left(\sub{2}_-\right)$, 
and  $ w\left(\sub{3}_+\right) \le w\left(\sub{3}_-\right)$, 
 then $Z_v\left(\sub{1}_+\right) \ge Z_v\left(\sub{1}_-\right)$, $Z_v\left(\sub{2}_+\right) \ge Z_v\left(\sub{2}_-\right)$, and $Z_v\left(\sub{3}_+\right) \le Z_v\left(\sub{3}_-\right)$ at  each vertex $v$ in $G$.
%
%Given a nonnegative weight function $w$ on each vertex in a 4-regular plane graph $G$ such that $ w\left(\sub{1}_+\right) \ge w\left(\sub{1}_-\right)$, $ w\left(\sub{2}_+\right) \ge w\left(\sub{2}_-\right)$, and $ w\left(\sub{3}_+\right) \le w\left(\sub{3}_-\right)$, we have $Z_v\left(\sub{1}_+\right) \ge Z_v\left(\sub{1}_-\right)$, $Z_v\left(\sub{2}_+\right) \ge Z_v\left(\sub{2}_-\right)$, and $Z_v\left(\sub{3}_+\right) \le Z_v\left(\sub{3}_-\right)$ for any vertex $v$ in $G$.
\end{corollary}
\begin{proof}
The proof is similar to that  of \corref{cor:proportion}, with the only difference 
that we are given $ w\left(\sub{3}_+\right) \le w\left(\sub{3}_-\right)$,
and we have $ W\left(\sub{3}_+\right) \le W\left(\sub{3}_-\right)$ by \thmref{thm:property_2}.
\end{proof}

%% file: fpras.tex
\documentclass[paper]{subfiles}

\begin{theorem}\label{thm:fpras}
There is an FPRAS for $Z(a, b, c, d)$ if $(a, b, c, d) \in \mathcal{F}_{\le^2} \bigcap \mathcal{A}_\le \bigcap \mathcal{B}_\le \bigcap \mathcal{C}_\le$.
\end{theorem}

\begin{theorem}\label{thm:fpras_planar}
There is an FPRAS for $Z(a, b, c, d)$ on planar graphs if $(a, b, c, d) \in \mathcal{F}_{\le^2} \bigcap \mathcal{A}_\le \bigcap \mathcal{B}_\le \bigcap \mathcal{C}_\ge$.
\end{theorem}

\begin{remark}
Our FPRAS result is actually stronger.
The FPRAS in \thmref{thm:fpras} for general graphs (including planar graphs) works even if different constraint functions from
% a finite subset of
%%% I commented this out
%%% JYC: i don't think  people appreciate this formal model of computation
% issue, and will unnecessarily misunderstood as the result being weaker.
% and then get confused...
 $\mathcal{F}_{\le^2} \bigcap \mathcal{A}_\le \bigcap \mathcal{B}_\le \bigcap \mathcal{C}_\le$ are assigned at different vertices. 
Similarlry the FPRAS in \thmref{thm:fpras_planar} for (only) planar graphs works even if different constraint functions from
% a finite subset of 
$\mathcal{F}_{\le^2} \bigcap \mathcal{A}_\le \bigcap \mathcal{B}_\le \bigcap \mathcal{C}_\ge$ are assigned at different vertices. (For logical reasons
concerning models of computation, the functions should 
take values in algebraic numbers, and if these functions are not
chosen from a fixed finite set then the description of each
constraint function used must be included in the input.
In this section for simplicity, we assume all constraint functions are from a fixed finite subset.)
%count toward the input size. 
%With some care the requirement ``a finite subset of'' can also be dropped off.
\end{remark}

We design our FPRAS using the common approach of approximately counting via almost uniformly sampling~\cite{JERRUM1986169, doi:10.1137/0218077, Dyer:1991:RPA:102782.102783, Sinclair92improvedbounds, Jerrum-book} by showing that a Markov chain designed for the six-vertex model can be adapted for the eight-vertex model.
The Markov chain we adapt is the \emph{directed-loop algorithm} which was invented by Rahman and Stillinger~\cite{doi:10.1063/1.1678874} and
is widely used for the six-vertex model (e.g., \cite{YANAGAWA1979329, Barkema98montecarlo, PhysRevE.70.016118}). The state space of our Markov chain $\mathcal{MC}$ for the eight-vertex model consists of \emph{even orientations} and \emph{near-even orientations}, which is an extension of the space of valid configurations; the transitions of this algorithm are composed of creating, shifting, and merging of the two defects on edges.
Some examples of the states in the directed-loop algorithm are shown in \figref{fig:moves} where the state in \figref{fig:moves_0} is an even orientation and the state in \figref{fig:moves_2_1} and the state in \figref{fig:moves_2_2} are near-even orientations with exactly two defects. Some typical moves in the directed-loop algorithm are as follows: the transition from the state in \figref{fig:moves_0} to the state in \figref{fig:moves_2_1} creates two defects; the transition from the state in \figref{fig:moves_2_1} to the state in \figref{fig:moves_0} merges two defects; the transitions between \figref{fig:moves_2_1} and \figref{fig:moves_2_2} shift one of the defects.
(Formal description of this Markov chain will be given shortly.)

\captionsetup[subfigure]{labelformat=parens}
\begin{figure}[h!]
\centering
\begin{subfigure}[b]{0.3\linewidth}
\centering\includegraphics[width=0.6\linewidth]{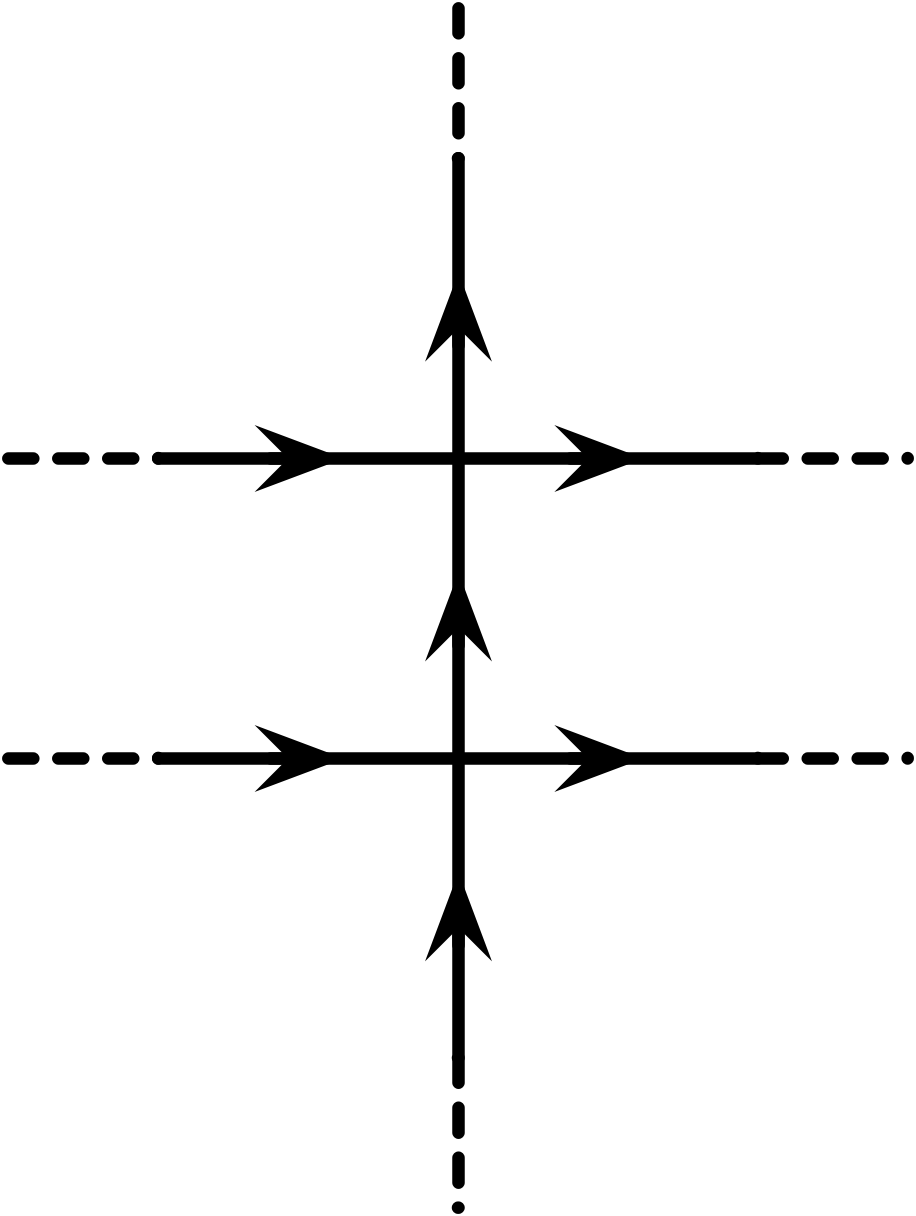}\caption{}\label{fig:moves_0}
\end{subfigure}
\begin{subfigure}[b]{0.3\linewidth}
\centering\includegraphics[width=0.6\linewidth]{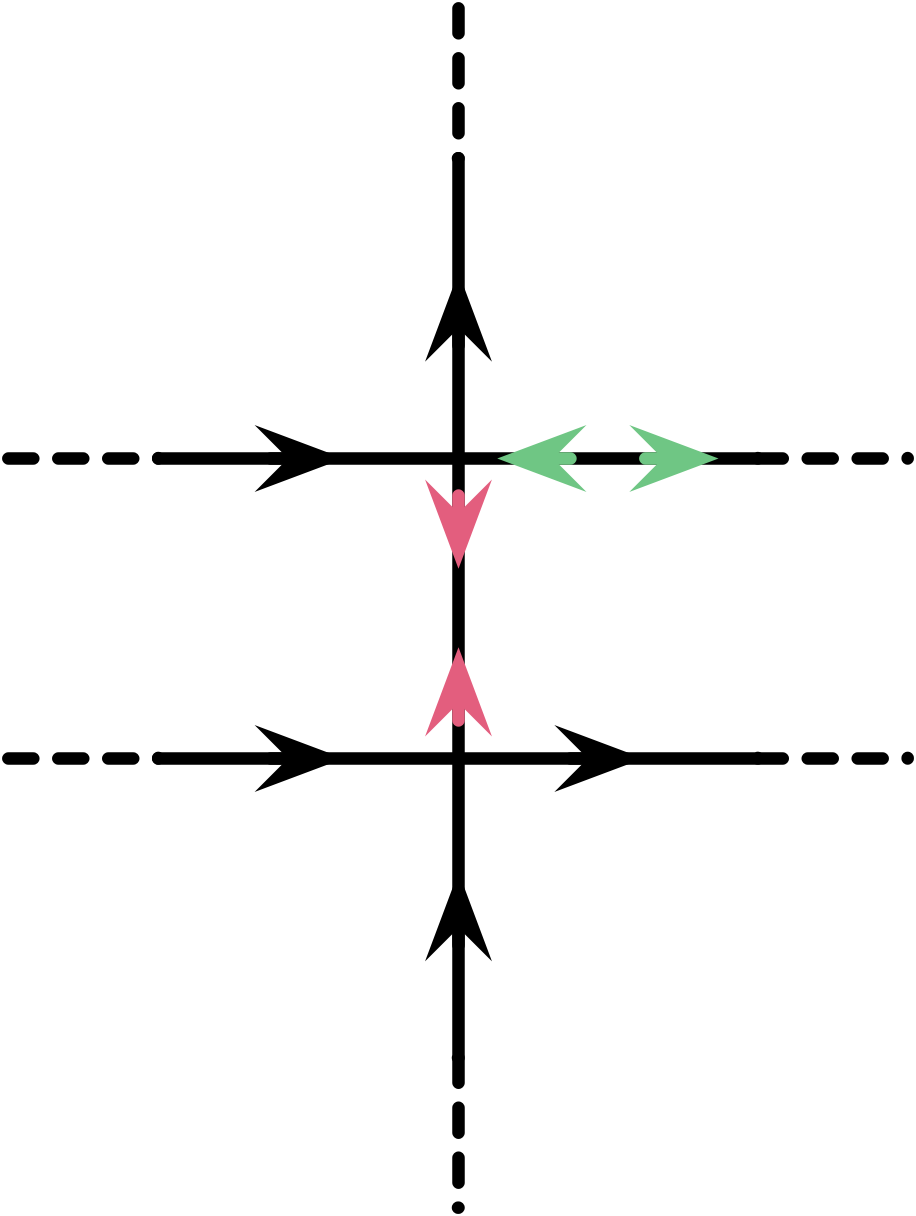}\caption{}\label{fig:moves_2_1}
\end{subfigure}
\begin{subfigure}[b]{0.3\linewidth}
\centering\includegraphics[width=0.6\linewidth]{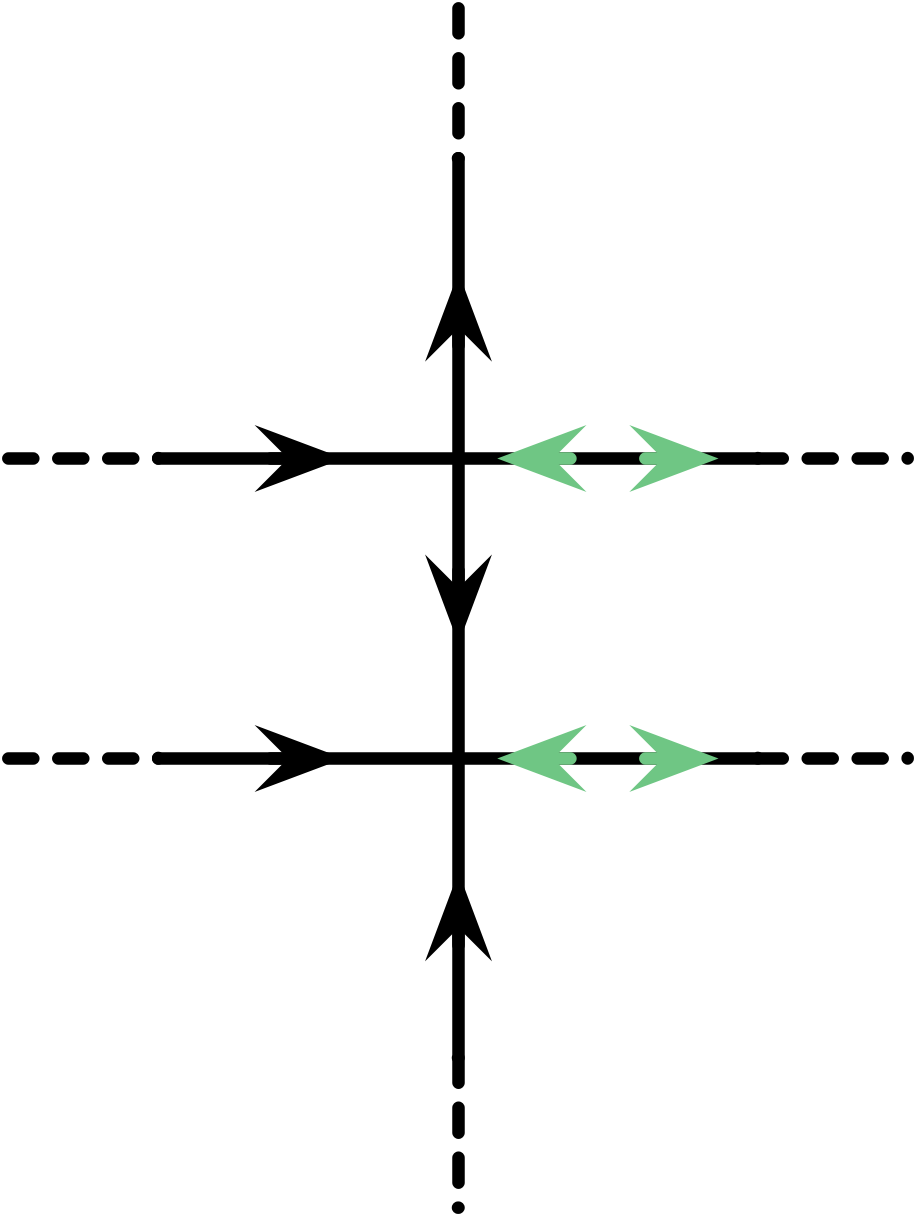}\caption{}\label{fig:moves_2_2}
\end{subfigure}  
\caption{Examples of the states in the directed-loop algorithm.}\label{fig:moves}
\end{figure}

\begin{notation}
For a 4-regular graph, denote the set of even orientations by $\Omega_0$ and the set of near-even orientations by $\Omega_2$. The state space of $\mathcal{MC}$ is $\Omega = \Omega_0 \cup \Omega_2$. Let $\mathcal{Z}(S)$ be the weighted sum of states in the set $S$.
\end{notation}

\begin{proof}[Proof of \thmref{thm:fpras}]
We will show (later) that $\mathcal{MC}$ is \emph{irreducible} and \emph{aperiodic}, and it satisfies the \emph{detailed balance condition} under the Gibbs distribution.
By the theory of Markov chains, we have an almost uniform sampler of $\Omega_0 \cup \Omega_2$.
This sampler is efficient if $\mathcal{MC}$ is rapidly mixing.
In this proof we show that for a 4-regular graph, if all constraint functions
used in an instance belong to $\mathcal{F}_{\le^2} \bigcap \mathcal{A}_\le \bigcap \mathcal{B}_\le \bigcap \mathcal{C}_\le$, then
% we have
\vspace{-2mm}
\setlist[enumerate]{itemsep=-1mm}
\begin{enumerate}[(1)]
\item
the $\mathcal{MC}$ is rapidly mixing via a \emph{conductance argument}~\cite{doi:10.1137/0218077, Dyer:1991:RPA:102782.102783, Sinclair92improvedbounds, Jerrum-book};
\item
even orientations take a non-negligible proportion in the state space;
\item
there exists a 
%way to do 
self-reduction (to reduce the computation of the partition function of a graph to that of a ``smaller'' graph)~\cite{JERRUM1986169}.
\end{enumerate}
\vspace{-2mm}
We remark that all three parts (1)(2)(3) depend on the idea of quantum decomposition and the closure properties shown in \secref{sec:properties}.

According to \lemref{lem:congestion}, when $(a, b, c, d) \in \mathcal{F}_{\le^2}$, the conductance of this $\mathcal{MC}$ is polynomially bounded \emph{if} $\frac{\mathcal{Z}(\Omega_2)}{\mathcal{Z}(\Omega_0)}$ is polynomially bounded.
According to \corref{cor:ratio},
when $(a, b, c, d) \in 
\mathcal{A}_\le \bigcap \mathcal{B}_\le \bigcap \mathcal{C}_\le$, $\frac{\mathcal{Z}(\Omega_2)}{\mathcal{Z}(\Omega_0)}$ is polynomially bounded, which proves part (2) above.
Combining \lemref{lem:congestion} and \corref{cor:ratio}, we can also conclude part (1).
As a consequence of (1) and (2), we are able to efficiently sample valid eight-vertex configurations according to the Gibbs measure on $\Omega_0$ (almost uniformly), and in the following algorithm we only work with states in $\Omega_0$, the set of even orientations.

Before we state the algorithm, we need to extend the type of vertices a graph can have in the eight-vertex model.
Previously, a graph can only have degree 4 vertices, on each of which a constraint function satisfies the even orientation rule and arrow reversal symmetry.
Now, a graph can also have degree 2 vertices, on each of which the constraint function satisfies the ``1-in-1-out'' rule and both valid local configurations have weight $1$.
Both \lemref{lem:congestion} and \corref{cor:ratio} still hold with this extension, because such a degree 2 vertex and its two incident edges just work together as a single edge.

We design the following algorithm to approximately compute the partition function $Z(G)$ via sampling with the directed-loop algorithm $\mathcal{MC}$.
As we have argued in \secref{sec:properties}, the partition function of the eight-vertex models can be viewed as the weighted sum over a set of {\sl dacp}'s. Since every constraint function belongs to $\overline{\mathcal{F}_>}$, by \lemref{lem:freedom_1} for  each vertex 
%with a parameter setting $(a, b, c, d)
%\in \overline{\mathcal{F}_>}$
%%%  nonnegative weight function $w$ may be diff for diff vertices
we can choose a nonnegative weight function $w$ on signed pairings at $v$.
 For a vertex $v \in V$, the ratios among different signed pairings $\{\sub{1}, \sub{2}, \sub{3}\} \times \{+, -\}$ in weighted {\sl dacp}'s can be uniquely determined by the ratios among different orientations (represented by $a$, $b$, $c$, and $d$) at $v$.
For example, if we express $Z(G)$ as $2aA + 2bB + 2cC + 2dD$ according to the
local orientation configuration at $v$,
as in the proof of \corref{cor:proportion}, we see that indeed
$2w(\sub{1}_-) ( A + D)$ is the weight for finding the
signed pairing $\sub{1}_-$ at $v$.
%%% JYC: pl check here!!!
%%% Tianyu: correct
As long as the partition function is not zero (this can be easily tested in polynomial time), there is a signed pairing $\boldsymbol{\varrho}$ showing up at $v$ with probability at least $\frac{1}{6}$ among all six signed pairings.
Moreover, according to \corref{cor:proportion}, one of the pairings
in  $\{\sub{1}_+, \sub{2}_+, \sub{3}_+\}$ shows up at $v$ with probability at least $\frac{1}{6}$.
Therefore, running $\mathcal{MC}$ on $G$, we can approximate, with a sufficient $1/{\rm poly}(n)$ precision, the probability of having $\boldsymbol{\varrho} \in \{\sub{1}_+, \sub{2}_+, \sub{3}_+\}$ at $v$, denoted by $\Pr_v(\boldsymbol{\varrho})$.
Denote by $G_{v, \boldsymbol{\varrho}}$ the graph with $v$ being split into $v_1$ and $v_2$ 
and the edges reconnected according to $\boldsymbol{\varrho}$. Recall that the degree 2 vertices $v_1$ and $v_2$ must satisfy the ``1-in-1-out'' rule in any valid configuration.
Write the partition function of $G_{v, \boldsymbol{\varrho}}$ as $Z(G_{v, \boldsymbol{\varrho}})$, we have $\Pr_v (\boldsymbol{\varrho}) = w(\boldsymbol{\varrho})Z(G_{v, \boldsymbol{\varrho}}) / Z(G)$ which means $Z(G) = w(\boldsymbol{\varrho})Z(G_{v, \boldsymbol{\varrho}}) / \Pr_v (\boldsymbol{\varrho})$. To approximate $Z(G)$ it suffices to approximate $Z(G_{v, \boldsymbol{\varrho}})$, which can be done by running $\mathcal{MC}$ on $G_{v, \boldsymbol{\varrho}}$ and recursing.
Repeating this process for $|V|$ steps we decompose the graph $G$ into the base case, a set of disjoint cycles.
The partition function of this cycle graph is just $2^C$ where $C$ is the number of cycles. By this self-reduction, the partition function $Z(G)$ can be approximated.
\end{proof}

\begin{proof}[Proof of \thmref{thm:fpras_planar}]
The proof is similar to that of \thmref{thm:fpras}, with the help of \corref{cor:ratio_planar} and \corref{cor:proportion_planar}, two corollaries of the closure property \thmref{thm:property_2} which 
%only 
%%% JYC i commented out this "only. because we did prove it 
% holds for property (P) graphs
holds on planar graphs. 

Given a plane graph $G$ with a constraint function on every vertex from $\mathcal{F}_{\le^2} \bigcap \mathcal{A}_\le \bigcap \mathcal{B}_\le \bigcap \mathcal{C}_\ge$, we can still efficiently sample even orientations according to the Gibbs measure. However, in order to do self-reduction, we have to prove something more. 

To make our algorithm work, we need to extend the type of vertices in the eight-vertex model again.
Previously in the proof of \thmref{thm:fpras}, a graph can have degree 4 vertices, on each of which the constraint function satisfies the even orientation rule and arrow reversal symmetry, and degree 2 vertices, on each of which the constraint function satisfies the ``1-in-1-out'' rule and both valid local configurations have weight $1$.
Now, a graph can also have degree 2 vertices, on each of which the constraint function satisfies the ``2-in/2-out'' rule and both valid local configurations have weight $1$.
One can check that \lemref{lem:congestion} still holds even with this extension.

The self-reduction still processes one vertex $v$ at a time.
As long as the partition function is not zero, there is a
signed pairing $\boldsymbol{\varrho}$ showing up at $v$ with probability at least $\frac{1}{6}$ among all six signed pairings.
Moreover, according to \corref{cor:proportion_planar}, one of the pairings $\{\sub{1}_+, \sub{2}_+, \sub{3}_-\}$ shows up at $v$ with probability at least $\frac{1}{6}$. If $\boldsymbol{\varrho}$ is $\sub{1}_+$ or $\sub{2}_+$, let
 $G_{v, \boldsymbol{\varrho}}$ be the graph with $v$ being split into $v_1$ and $v_2$ 
and the edges reconnected according to $\boldsymbol{\varrho}$. The degree 2 vertices $v_1$ and $v_2$ must satisfy the ``1-in-1-out'' rule in any valid configuration, just as in the proof of \thmref{thm:fpras}.

If $\boldsymbol{\varrho}$ is $\sub{3}_-$, let $G_{v, \sub{3}_-}$ be the graph with $v$ being split into $v_1$ and $v_2$  and the edges reconnected according to $\sub{3}_-$. This time, the degree 2 vertices $v_1$ and $v_2$ must satisfy the ``2-in/2-out'' rule in any valid configuration.
Observe that
\thmref{thm:property_2} holds for $G_{v, \sub{3}_-}$ if and only if it holds for $G'_{v, \sub{3}_-}$, which is 
obtained from $G$ by replacing  $v$ by a virtual vertex $v'$ with parameter setting $(a, b, c, d) = (0, 0, 1, 1)$ (this is equivalent to 
choosing $w(\sub{3}_-) = 1$ and $w$ being 0 on the other five signed pairings,
 for a nonnegative $w$ at $v'$).
Since $(0, 0, 1, 1) \in \mathcal{A}_\le \bigcap \mathcal{B}_\le \bigcap \mathcal{C}_\ge \bigcap \overline{\mathcal{F}_>}$, \thmref{thm:property_2} and consequently \corref{cor:ratio_planar} still hold for $G'_{v, \sub{3}_-}$ thus also for $G_{v, \sub{3}_-}$.
(Note that this  $G'_{v, \sub{3}_-}$ is not involved algorithmically
in subsequent steps; its only purpose is to show that
\thmref{thm:property_2} holds for $G_{v, \sub{3}_-}$, on which
the algorithm continues.)

The subsequent steps in the self-reduction step for $v$
 are the same as in the proof of \thmref{thm:fpras}.
The base case is a decomposition of $G$ into a set of disjoint cycles with an even number of degree 2 vertices that satisfy the ``2-in/2-out'' rule.
This is proved by using the  \emph{Jordan Curve Theorem}:
The graph is initially planar. Any step replacing  $v$ with $v_1$ and $v_2$
for
 $\sub{1}_+$ or $\sub{2}_+$ in
 $G_{v, \boldsymbol{\varrho}}$ does not create any non-planar crossings
nor vertices satisfying the ``2-in/2-out'' rule.
Only the third type of steps  replacing  $v$ with $v_1$ and $v_2$
for
$\sub{3}_-$ in
 $G_{v, \sub{3}_-}$ create a non-planar crossing
and also a vertex satisfying the ``2-in/2-out'' rule at each
crossing locally at each branch of the crossing.
Thus at the end we are left with a set of disjoint cycles 
where along each cycle degree 2 vertices 
 satisfying the ``2-in/2-out'' rule are in 1-1
correspondence with non-planar crossings. By 
the \emph{Jordan Curve Theorem} this number is even, for every cycle. 
The partition function of this cycle graph is just $2^C$ where $C$ is the number of cycles. Again, the partition function $Z(G)$ can be approximated.
\end{proof}

\begin{corollary}\label{cor:ratio}
Given a 4-regular graph $G = (V, E)$, if the constraint function on every vertex is from $\mathcal{A}_\le \bigcap \mathcal{B}_\le \bigcap \mathcal{C}_\le$, then $\frac{\mathcal{Z}(\Omega_2)}{\mathcal{Z}(\Omega_0)} \le \binom{|E|}{2}$.
\end{corollary}
\begin{proof}
\captionsetup[subfigure]{labelformat=parens}
\begin{figure}[h!]
\centering
\begin{subfigure}[b]{0.45\linewidth}
\centering\includegraphics[width=0.9\linewidth]{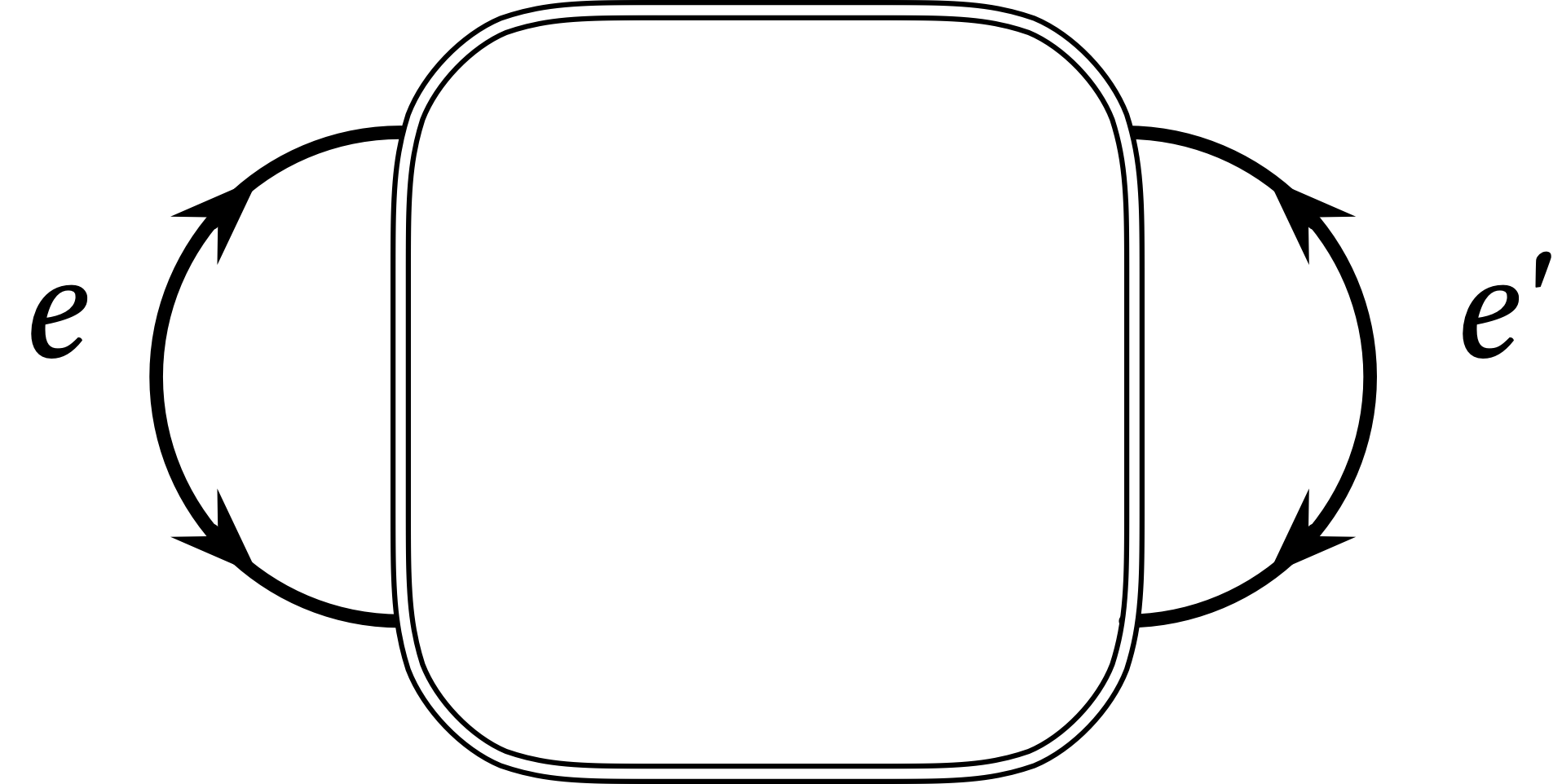}\caption{A near-even orientation with defects at $e$ and $e'$.}\label{fig:congestion_state}
\end{subfigure}
\hspace{0.05\linewidth}
\begin{subfigure}[b]{0.45\linewidth}
\centering\includegraphics[width=0.71\linewidth]{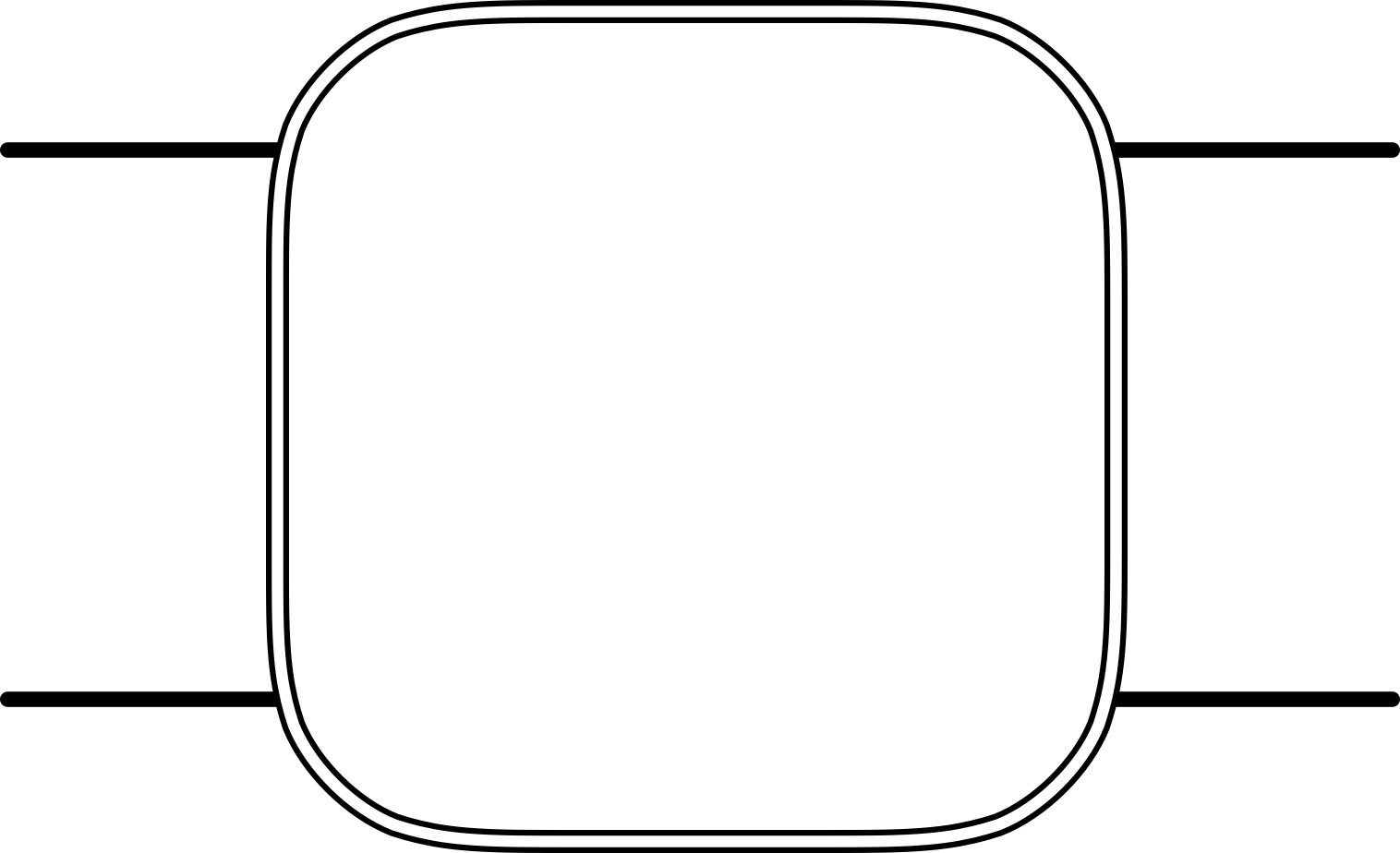}\caption{A 4-ary construction by cutting open $e$ and $e'$.}\label{fig:congestion_gadget}
\end{subfigure}
\caption{}\label{fig:congestion}
\end{figure}

For each near-even orientation, there are exactly two defective edges.
Let $\Omega_2^{\{e, e'\}} \subseteq \Omega_2$ be 
the set of near-even orientations in which $e, e'$ are these two defective edges.
We have $\frac{\mathcal{Z}(\Omega_2)}{\mathcal{Z}(\Omega_0)} = 
\sum_{\{e, e'\} \ \in \ {E \choose 2}}\frac{\mathcal{Z}(\Omega_2^{\{e, e'\}})}{\mathcal{Z}(\Omega_0)}$.

For any $\tau \in \Omega_2$, each of $e$ and $e'$ 
%JYC: I can not say it "can". I will say it "may"
% and I don't think they are "independently"  just possibilities.
 %can independently 
may have both half-edges coming in or going out, with 4 possibilities.
An example is in \figref{fig:congestion_state} where both $e$ and $e'$ have their half-edges going out.
If we ``cut open'' $e$ and $e'$ as shown in \figref{fig:congestion_gadget},
we get a 4-ary construction $\Gamma$ using degree 4 vertices with constraint functions in $\mathcal{A}_\le \bigcap \mathcal{B}_\le \bigcap \mathcal{C}_\le$.
Denote the constraint function of $\Gamma$ by $(a', b', c', d')$, with the input order being counter-clockwise starting from the upper-left edge.
For this 4-ary construction $\Gamma$ we observe that: the set of near-even orientations in $\Omega_2^{\{e, e'\}}$ contributes a total weight $(a' + a' + d' + d')$, i.e. $\mathcal{Z}(\Omega_2^{\{e, e'\}}) = 2(a' + d')$; the set of even orientations in $\Omega_0$ has a total weight $\mathcal{Z}(\Omega_0) = 2(b' + c')$.
By \thmref{thm:property_1} we know that for the 4-ary construction $\Gamma$, $a' + d' \le b' + c'$. Therefore, $\frac{\mathcal{Z}(\Omega_2^{\{e, e'\}})}{\mathcal{Z}(\Omega_0)} \le 1$. In total, $\frac{\mathcal{Z}(\Omega_2)}{\mathcal{Z}(\Omega_0)} \le \binom{|E|}{2}$.
\end{proof}

\begin{corollary}\label{cor:ratio_planar}
Given a 4-regular plane graph $G = (V, E)$, if the constraint function on every vertex is from $\mathcal{A}_\le \bigcap \mathcal{B}_\le \bigcap \mathcal{C}_\ge \bigcap \overline{\mathcal{F}_>}$, then $\frac{\mathcal{Z}(\Omega_2)}{\mathcal{Z}(\Omega_0)} \le \binom{|E|}{2}$.
\end{corollary}
\begin{proof}
For any 4-regular plane graph $G = (V, E)$, 
if we cut the
two defective edges of $\tau \in \Omega_2$,
we obtain a planar $\Gamma$ with 4 dangling edges
using constraint
functions from $\mathcal{A}_\le \bigcap \mathcal{B}_\le \bigcap \mathcal{C}_\ge \bigcap \overline{\mathcal{F}_>}$.
We name  $e_1$ and $e_2$ the two dangling edges cut from one edge in $G$,
and  $e_3$ and $e_4$ cut from the other.
Both $e_1$ and $e_2$ now reside in a single face of $\Gamma$,
and so do $e_3$ and $e_4$. We can modify the proof of
 \thmref{thm:property_2} to establish that for $\Gamma$,
we still have $a' + d' \le b' + c'$.
%
%By \thmref{thm:property_2}, the parameters
%$(a', b', c', d')$ of the  constraint function defined by
% $\Gamma$ also belongs to $\mathcal{A}_\le \bigcap \mathcal{B}_\le \bigcap \mathcal{C}_\ge \bigcap \overline{\mathcal{F}_>}$. In particular $a' + d' \le b' + c'$.
\end{proof}

%In the above proof for general graphs, the roles of $a', b', c'$ are symmetric. Thus we can bound $a' + d' \le b' + c'$, $b' + d' \le a' + c'$, and $c' + d' \le a' + b'$ as well.
%However, for plane graphs, the entry $c'$ in the constraint function of $g$ can only correspond to configurations in $\mathcal{Z}(\Omega_0)$. Therefore, we only need to bound $a' + d' \le b' + c'$ and $b' + d' \le a' + c'$. According to \thmref{thm:property_2}, this is true as long as the constraint function on every vertex is from $\mathcal{A}_\le \bigcap \mathcal{B}_\le \bigcap \mathcal{C}_\ge \bigcap \overline{\mathcal{F}_>}$.
%\end{proof}

%%%%%%%%%%%%%%%%%%%%%%%%%%%%%%%

Although $\mathcal{MC}$ runs on the even orientations and near-even orientations of a 4-regular graph $G$, it is formally defined and analyzed using the edge-vertex incidence graph $G'$ of $G$ introduced in \secref{sec:prelim}.

Let $G' = (V,U,E)$ be the edge-vertex incidence graph of
$G$, an instance of $Z(a,b,c,d)$.
Each vertex in $V$ is assigned  $(\neq_2)$;
each vertex  $u \in U$ is assigned a constraint function 
$f_u \in \mathcal{F}_{\le^2}$.  An \textit{assignment} $\sigma$ assigns a
value in $\{0, 1\}$ to each edge $e \in E$.
The  state space of $\mathcal{MC}$ is $\Omega = \Omega_0 \cup \Omega_2$, which 
consists of  ``perfect''  or ``near-perfect'' assignments
to $E$, defined as follows: all assignments satisfy the ``two-0-two-1/four-0/four-1'' rule at every vertex $u \in U$ of degree 4;
  all assignments satisfy the ``one-0-one-1''
at every $v \in V$ with \emph{possibly exactly two} exceptions.
Assignments in $\Omega_0$ have no exceptions, and are ``perfect'' (corresponding to the even orientations in $G$).
Assignments in $\Omega_2$ have exactly two exceptions, and are ``near-perfect'' (corresponding to the near-even orientations in $G$).
Thus any $\sigma \in \Omega_0$ sastifies all $(\neq_2)$ on $V$,
and  any $\sigma \in \Omega_2$ sastifies all $(\neq_2)$ on $V -\{v', v''\}$
for some two vertices $v', v'' \in V$ where it satisfies $(=_2)$
(which outputs 1 on inputs 00, 11 and outputs 0 on 01, 10).
For any assignment $\sigma \in \Omega$ and any subset $S \subseteq \Omega$,
define the \textit{weight} function $\mathcal{W}$ by 
 $\mathcal{W}(\sigma) = \prod_{u \in U} f_u(\sigma |_{E(u)})$ and $\mathcal{Z}(S) = \sum_{\sigma \in S} \mathcal{W}(\sigma)$.
Then the \textit{Gibbs measure} for $\Omega$ is defined
by $\pi(\sigma) = \frac{\mathcal{W}(\sigma)}{\mathcal{Z}(\Omega)}$,
assuming $\mathcal{Z}(\Omega)>0$.

Transitions in $\mathcal{MC}$ are comprised of three types of moves.
Suppose $\sigma \in \Omega_0$.
An  $\Omega_0$-to-$\Omega_2$ move from $\sigma$
 takes a degree 4 vertex $u \in U$ and
two incident edges $e' = (v', u), e'' = (v'', u) \in V \times U$,
%satisfying $\{\sigma(e'), \sigma(e'')\}  = \{0, 1\}$,
and changes it to $\sigma_2 \in \Omega_2$
which flips both $\sigma(e')$ and $\sigma(e'')$.
The effect is that %we still have $\{\sigma_2(e'), \sigma_2(e'')\}  = \{0, 1\}$, but
at $v'$ and $v''$, $\sigma_2$ satisfies $(=_2)$ instead of $(\neq_2)$.
An $\Omega_2$-to-$\Omega_0$ move is the opposite.
An $\Omega_2$-to-$\Omega_2$ move is, intuitively,  to 
\emph{shift} one $(=_2)$ from one vertex  $v' \in V$ 
to another $v^* \in V$, where for some $u \in U$, 
 $v'$ and $v^*$ are both incident to $u$
 and the ``two-0-two-1/four-0/four-1'' rule at $u$ is preserved.
%and  the ``two-0 two-1'' rule at $u$ is preserved. 
Formally, let  $\sigma \in \Omega_2$ be a near-perfect assignment
with $v', v'' \in V$ being the two exceptional vertices (i.e.,
$\sigma$ satisfies $(=_2)$ at $v'$ and $v''$). 
Let $v^* \in V - \{v', v''\}$ be such that
for some $u \in U$, both $e' = (v', u), e^* = (v^*, u) \in E$.
%and $\{\sigma(e'), \sigma(e^*)\}  = \{0, 1\}$.
Then an $\Omega_2$-to-$\Omega_2$ move changes $\sigma$ to $\sigma^*$
by flipping  both $\sigma(e')$ and $\sigma(e^*)$.
The effect is that %we still have $\{\sigma^*(e'), \sigma^*(e^*)\}  = \{0, 1\}$, but
$\sigma^*$ satisfies $(\neq_2)$ at $v'$ 
and  $(=_2)$  at $v^*$. Note that $\sigma^*$ continues to satisfy
 $(=_2)$ at $v''$.

The above describes a symmetric binary relation \emph{neighbor} ($\sim$)
 on $\Omega$.  No two states in $\Omega_0$ are neighbors.
Set $n = |U|$.
The number of neighbors of a $\Omega_0$-state is at most $6n$ (by first picking a vertex and then picking a pair of edges incident to this vertex) and the number of neighbors of a $\Omega_2$-state is at most a constant.
The transition probabilities $P(\cdot, \cdot)$ of $\mathcal{MC}$ are \textit{Metropolis} moves between neighboring states:
\[
P(\sigma_1, \sigma_2) = 
\left\{
\begin{array}{ll}
\frac{1}{12n}\min\left(1, \frac{\pi(\sigma_2)}{\pi(\sigma_1)}\right) & \text{if }\sigma_2 \sim \sigma_1; \\
1 - \frac{1}{12n}\sum_{\sigma' \sim \sigma_1}\min\left(1, \frac{\pi(\sigma')}{\pi(\sigma_1)}\right) & \text{if } \sigma_1 = \sigma_2; \\
0 & \text{otherwise.}
\end{array}
\right.
\]
$\mathcal{MC}$ is aperiodic due to the ``lazy'' movement; one can verify that $\mathcal{MC}$ is irreducible by creating, shifting, and merging two $(=_2)$'s; as the transitions are Metropolis moves, detailed balance conditions are satisfied with regard to $\pi$. 
By results from~\cite{doi:10.1137/0218077, Sinclair92improvedbounds},
such a Markov chain is rapidly mixing if there is a \textit{flow} whose congestion can be bounded by a polynomial in $n$.
\begin{lemma}\label{lem:congestion}
Assume $\mathcal{Z}(\Omega_0) > 0$. Given $f_u \in \mathcal{F}_{\le^2}$ for every vertex $u \in U$, there is a flow on $\Omega$ with congestion at most $O\left(n^2 \left(\frac{\mathcal{Z}(\Omega)}{\mathcal{Z}(\Omega_0)}\right)^2\right)$, using paths of length $O(n)$.
\end{lemma}
\begin{proof}
The idea is to design a flow $\mathfrak{F}: \mathcal{P} \rightarrow \mathbb{R}^+$ from $\Omega_2$ to $\Omega_0$ which satisfies
\[\sum_{p \in \mathcal{P}_{\sigma_2\sigma_0}} \mathfrak{F}(p) = \pi(\sigma_2)\pi(\sigma_0)\text{, \; for all } \sigma_2 \in \Omega_2, \sigma_0 \in \Omega_0,\]
where $\mathcal{P}_{\sigma_2\sigma_0}$ is defined to be a set of simple directed paths from $\sigma_2$ to $\sigma_0$ in $\mathcal{MC}$
and $\mathcal{P} = \bigcup_{\sigma_2 \in \Omega_2, \sigma_0 \in \Omega_0}\mathcal{P}_{\sigma_2\sigma_0}$. 
Once the congestion of $\mathfrak{F}$ from $\Omega_2$ to $\Omega_0$ is polynomially bounded, so is the flow from $\Omega_0$ to $\Omega_2$ by symmetric construction. Moreover, there is a flow from $\Omega_2$ to $\Omega_2$ (or from $\Omega_0$ to $\Omega_0$) whose congestion can also be polynomially bounded by randomly picking an intermediate state in $\Omega_0$ (or $\Omega_2$, respectively). Thus we have a flow on $\Omega$ with polynomially bounded congestion. This technique has been used in \cite{Jerrum:2004:PAA:1008731.1008738, DBLP:journals/corr/abs-1301-2880}.
In the following we show that the congestion of $\mathfrak{F}$
from  $\Omega_2$ to $\Omega_0$ is
 bounded by $O(n^2)\frac{\mathcal{Z}(\Omega_2)}{\mathcal{Z}(\Omega_0)}$.
Then the bound in the lemma for a flow on $\Omega$  follows. 

To describe the flow $\mathfrak{F}$, we first specify the
sets of paths that are going to take the flow.
In line with the definition of $\Omega_0$ and $\Omega_2$, 
we define $\Omega_4$ to be the set of assignments where there are exactly four violations of $(\neq_2)$ in $V$. Let $\Omega' = \Omega_0 \cup \Omega_2 \cup \Omega_4$.
For $\sigma, \sigma' \in \Omega'$, let  $\sigma \oplus \sigma'$
denote the  \textit{symmetric difference} (or bitwise  XOR),
where we view $\sigma$ and $\sigma'$ as two bit strings in $\{0, 1\}^{|E|}$.
This is a 0-1 assignment to the edge set of the edge-vertex incidence graph
$G' = (V, U, E)$ of $G$. We also treat $\sigma \oplus \sigma'$ as an edge subset
of $E$ (corresponding to bit 
positions having bit 1, where  $\sigma$ and $\sigma'$
assign opposite values), and this defines an edge-induced
%%% JYC called edge-induced subgraph: 
% is a subset of the edges of a graph G together with any vertices 
% that are their endpoints.
subgraph of $G'$, which we will just call it $\sigma \oplus \sigma'$.
%%%
%JYC added a word.
%%%
 Since at every $u \in U$ of degree 4, the 
``two-0-two-1/four-0/four-1'' rule is satisfied by both $\sigma$ and $\sigma'$,
this edge-induced subgraph has even degree (0, 2, or 4) at every $u \in U$.

%%% Tianyu: the following paragraph is new and only for the eight-vertex model
Let us introduce the set of {\sl atcp}'s (annotated trail \& circuit partitions) for the symmetric difference $\sigma \oplus \sigma'$.
It is similar to the notions of {\sl acp} for 4-regular graphs and {\sl atcp} for 4-ary constructions defined in \secref{sec:properties}.
%%% Tianyu: added the following sentence.
Let us assume $\sigma \in  \Omega_0$ and $\sigma' \in  \Omega_2$, and the set of {\sl atcp}'s for $\sigma \oplus \sigma'$ in general cases when $\sigma, \sigma' \in \Omega'$ can be similarly defined.
If $\sigma \in  \Omega_0$ and $\sigma' \in  \Omega_2$,
on the edge where $\sigma'$ is defective (but $\sigma \in  \Omega_0$ is
not), $\sigma \oplus \sigma'$ has a degree 1 vertex.
First we assign a pairing (that groups four incident edges
into two unordered pairs) at every vertex of degree 4
in $\sigma \oplus \sigma'$. This partitions the edges of $\sigma \oplus \sigma'$
into a set of  edge-disjoint circuits and 
exactly one
 \emph{trail} which ends in the two vertices in $V$ of degree 1.
Then we affix a $\pm$ at every vertex $u \in U$ of degree 2 or degree 4 
in $\sigma \oplus \sigma'$
as follows: If $u \in U$ has degree 4 in $\sigma \oplus \sigma'$
then $\sigma$ and $\sigma'$ represent total reversal orientations of each other 
at $u$, and thus the pairing at $u$ has the same sign
according to \tabref{tab:correspondence} for $\sigma$ and $\sigma'$. We affix
this sign at $u$.
If $u \in U$ has degree 2 in $\sigma \oplus \sigma'$,
then $\sigma$ and $\sigma'$
disagree on exactly two edges. On these two edges,
if one assigns 01 the other assigns 10
(and vice versa), and if one assigns 00 the other assigns 11 (and vice versa).
We  affix $+$ at $u$ in the first case, and $-$ in the second case.
One can check that for any {\sl atcp} $\varphi$ of $\sigma \oplus \sigma'$, one encounters an even number of $-$ along any circuit of $\varphi$.
%%%%
%An {\sl atcp} of $\sigma \oplus \sigma'$ is a partition of the edges in $\sigma \oplus \sigma'$ into edge-disjoint circuits and exactly one \emph{trail} which ends in the two vertices in $V$ of degree 1, together with a sign ($+$ or $-$) on each vertex $u \in U$ of degree 2 or degree 4 with the following restriction.
%Given $\sigma$ and $\sigma'$,
%every local configuration of $\sigma$ (or symmetrically $\sigma'$) at a vertex $u \in U$ of degree 4
%defines exactly three signed pairings at $u$
%according to \tabref{tab:correspondence}.
%Besides, every local configuration of $\sigma$ (or symmetrically $\sigma'$) at a vertex $u \in U$ of degree 2
%defines exactly one signed pairing at $u$: there are exactly two edges incident to $u$ in $\sigma$ and $\sigma'$, so the pairing is obvious; its sign is $+$ if the two edges have consistent orientations in \emph{both} $\sigma$ and $\sigma'$, and is $-$ if the two edges have contrary orientations in \emph{both} $\sigma$ and $\sigma'$.
%One can check that for any {\sl atcp} $\varphi$ of $\sigma \oplus \sigma'$, one encounters an even number of $-$ along any circuit in $\varphi$.

Denote by $U_4 \subseteq U$ the degree-4 vertices in $\sigma \oplus \sigma'$.
Then there are exactly $3^{|U_4|}$ {\sl atcp}'s for $\sigma \oplus \sigma'$.
Note that an {\sl atcp} of $\sigma \oplus \sigma'$ is uniquely determined by a family of signed pairings on $U_4$.
This is a 1-1 correspondence and we will identify the two sets.
For any signed pairing in $\{\sub{1}, \sub{2}, \sub{3}\} \times \{+, -\}$ on a vertex
$u$ with constraint matrix $M(f_u) = \left[\begin{smallmatrix} d & & & a \\ & b & c & \\ & c & b & \\ a & & & d \end{smallmatrix}\right]$, define
the weight function $\mathfrak{w}$ for signed pairings as follows,
$\left\{\begin{smallmatrix}
a^2 = \mathfrak{w}(\subinmatrix{1}_-) + \mathfrak{w}(\subinmatrix{2}_+) + \mathfrak{w}(\subinmatrix{3}_+) \\
b^2 = \mathfrak{w}(\subinmatrix{1}_+) + \mathfrak{w}(\subinmatrix{2}_-) + \mathfrak{w}(\subinmatrix{3}_+) \\
c^2 = \mathfrak{w}(\subinmatrix{1}_+) + \mathfrak{w}(\subinmatrix{2}_+) + \mathfrak{w}(\subinmatrix{3}_-) \\
d^2 = \mathfrak{w}(\subinmatrix{1}_-) + \mathfrak{w}(\subinmatrix{2}_-) + \mathfrak{w}(\subinmatrix{3}_-) \\
\end{smallmatrix}\right.$.
Note that $\mathfrak{w}$ has a \textit{nonnegative} solution if and only if $f_u \in \mathcal{F}_{\le^2}$ by a proof similar to that of \lemref{lem:freedom_1}.
Let $\Phi_{\sigma \oplus \sigma'}$ be the set of {\sl atcp}'s for $\sigma \oplus \sigma'$.
For $\varphi \in \Phi_{\sigma \oplus \sigma'}$, define
\[\mathfrak{W}(\sigma, \sigma', \varphi) := \left(\prod_{u \in U\setminus U_4}f_u\left(\sigma |_{E(u)}\right)f_u\left(\sigma' |_{E(u)}\right)\right)\left(\prod_{u \in U_4}\mathfrak{w}(\varphi(u))\right),\]
where $\varphi(u)$ is the signed pairing given by $\varphi$ at $u$.
Then for all distinct $\sigma, \sigma' \in \Omega'$, we have
\begin{align*}
\sum_{\varphi \in \Phi_{\sigma \oplus \sigma'}}\mathfrak{W}(\sigma, \sigma', \varphi)
& = \sum_{\varphi \in \Phi_{\sigma \oplus \sigma'}} \left(\prod_{u \in U\setminus U_4}f_u\left(\sigma |_{E(u)}\right)f_u\left(\sigma' |_{E(u)}\right)\right)\left(\prod_{u \in U_4}\mathfrak{w}(\varphi(u))\right) \\
& = \left(\prod_{u \in U\setminus U_4}f_u\left(\sigma |_{E(u)}\right)f_u\left(\sigma' |_{E(u)}\right)\right) \left(\sum_{\varphi \in \Phi_{\sigma \oplus \sigma'}} \prod_{u \in U_4}\mathfrak{w}(\varphi(u)) \right) \\
& = \left(\prod_{u \in U\setminus U_4}f_u\left(\sigma |_{E(u)}\right)f_u\left(\sigma' |_{E(u)}\right)\right) \left(\prod_{u \in U_4}f_u\left(\sigma |_{E(u)}\right)f_u\left(\sigma' |_{E(u)}\right)\right) \\
& = \prod_{u \in U} f_u\left(\sigma |_{E(u)}\right)f_u\left(\sigma' |_{E(u)}\right) \\
& = \mathcal{W}(\sigma)\mathcal{W}(\sigma').
\end{align*}
The equality from line 2 to line 3 is due to
the following: when the degree (in the
induced subgraph  $\sigma \oplus \sigma'$)  of a vertex $u \in U$ is 4, $\sigma$ and $\sigma'$ must take the same value at $u$,
since one represents a total reversal of all arrows of another; 
thus $f_u\left(\sigma |_{E(u)}\right)f_u\left(\sigma' |_{E(u)}\right)$ is in $\{a^2, b^2, c^2, d^2\}$. Then
\[\prod_{u \in U_4}
f_u\left(\sigma |_{E(u)}\right)f_u\left(\sigma' |_{E(u)}\right)
=  \sum_{\varphi \in \Phi_{\sigma \oplus \sigma'}} \prod_{u \in U_4}\mathfrak{w}(\varphi(u))\]
 is obtained by using the sum expressions for $a^2$, $b^2$, $c^2$, and $d^2$ in terms of 
$\mathfrak{w}(\sub{1}_+)$, $\mathfrak{w}(\sub{1}_-)$,
$\mathfrak{w}(\sub{2}_+)$, $\mathfrak{w}(\sub{2}_-)$,
$\mathfrak{w}(\sub{3}_+)$, and $\mathfrak{w}(\sub{3}_-)$,
 and then expressing the product-of-sums as a sum-of-products.

Now we are ready to specify the ``paths'' which take nonzero flow from $\sigma_2 \in \Omega_2$ to $\sigma_0 \in \Omega_0$.
In order to transit from $\sigma_2$ to $\sigma_0$, paths in $\mathcal{P}_{\sigma_2\sigma_0}$ go through states in $\Omega$ that gradually decrease the number of conflicting assignments along trails and circuits in $\sigma_2 \oplus \sigma_0$.
We first specify a total order on $E$, the set of edges of $G'$.
This induces a total order on circuits by lexicographic order. 
In the induced subgraph $\sigma_2 \oplus \sigma_0$, exactly
 two vertices in $V$ have degree 1 (called \textit{endpoints}) and
all other vertices have degree 2 or degree 4.
The set of paths in $\mathcal{P}_{\sigma_2\sigma_0}$ are  designed to be in 1-to-1 correspondence with elements in $\Phi_{\sigma_2 \oplus \sigma_0}$.
Given any family of signed pairings $\varphi \in \Phi_{\sigma_2 \oplus \sigma_0}$,
we have a unique decomposition of the induced 
subgraph $\sigma_2 \oplus \sigma_0$ as an edge disjoint union
of one trail 
$[e_1] (v_1, e'_1, u_1, e_2, v_2, e'_2, u_2, \ldots, e_k, v_k)  [e'_k]$
(where $e_1$ and $e'_k$ are not part of the trail),
 and zero or more edge disjoint circuits, which are
ordered lexicographically.
Here $v_i \in V$ and $u_i \in U$,
and the two exceptional vertices are $v_1$ and $v_k$ where
$\sigma_2$ satisfies $(=_2)$.
%The unique path $p_\varphi$ first ``pushes''  the $(=_2)$
%from $v_1$, to $v_2$, then to $v_3, \ldots, v_{k-1}$, and then ``merge'' 
%at $v_{k}$, arriving at a configuration in $\Omega_0$.
The unique path $p_\varphi$ first reverses all arrows along the trail, starting from the smaller of $e'_1$ and $e_k$.
If we assume, without loss of generality, $e'_1$ is the smaller one,
then $p_\varphi$ ``pushes''  the $(=_2)$
from $v_1$, to $v_2$, then to $v_3, \ldots, v_{k-1}$, and then ``merge'' 
at $v_{k}$, arriving at a configuration in $\Omega_0$.
Next  $p_\varphi$  reverses all arrows on each
circuit in lexicographic order, and within each circuit $C$
it starts at the least edge $e$ (according to the edge order)
% lexicographically the least edge $e$ 
%and reverses all arrows
%on $C$
%in the direction that $e$ suggests.
and reverses all arrows on $C$ in a cyclic order starting in the direction 
indicated by $\sigma_2$ on $e$.
%in the direction defined by the starting cyclic orientation of $\sigma_2$. 
(Technically it flips a pair of incident edges to vertices in $U$
in each step.)
Such paths  $p_\varphi$  are well-defined and
are valid paths in $\mathcal{MC}$ since along any path every state is in $\Omega = \Omega_0 \cup \Omega_2$ and every move is a valid transition defined in $\mathcal{MC}$.
With regard to the flow distribution, the  flow value
put on $p_\varphi$ is $\frac{\mathfrak{W}(\sigma_2, \sigma_0, \varphi)}{\left(\mathcal{Z}(\Omega)\right)^2}$, making
the following hold for all $\sigma_2 \in \Omega_2, \sigma_0 \in \Omega_0$:
\begin{align*}
\sum_{p_\varphi \in \mathcal{P}_{\sigma_2\sigma_0}} \mathfrak{F}(p_\varphi) & = \sum_{\varphi \in \Phi_{\sigma_2 \oplus \sigma_0}}\frac{\mathfrak{W}(\sigma_2, \sigma_0, \varphi)}{\left(\mathcal{Z}(\Omega)\right)^2} \\
& = \frac{\mathcal{W}(\sigma_2)\mathcal{W}(\sigma_0)}{\left(\mathcal{Z}(\Omega)\right)^2} \\
& = \pi(\sigma_2)\pi(\sigma_0).
\end{align*}

Note that in each path, no edge is flipped more than once, so the length is $O(n)$.
For any transition $(\sigma', \sigma'')$ where $\sigma' \neq \sigma''$, we have $P(\sigma', \sigma'') = \frac{1}{12n}\min\left(1, \frac{\pi(\sigma'')}{\pi(\sigma')}\right) = \Omega\left(\frac{1}{n}\right)$, as $\frac{\pi(\sigma'')}{\pi(\sigma')}$ is a constant. (This is a constant because we have restricted
the constraint function $f_u$ to be from a fixed finite set $\mathcal{F}$.)
%%%JYC: here is a slight problem: if we really take F_{<=^2} as a set,
%%% assigning possibly different sig at each u \in U,
%%% it is possible to have exponentially small values and ratios.
%%% one way is to say we consider Hol(\neq_2 | F) for ony finite subsets
%%% of  F_{<=^2}  
 Let $H_{\sigma'} = \{\sigma_2 \oplus \sigma_0 \ |\ \sigma_2 \in \Omega_2, \sigma_0 \in \Omega_0, \exists \varphi \in \Phi_{\sigma_2 \oplus \sigma_0} \text{ s.t. } \sigma' \in p_\varphi\}$.
The congestion of $\mathfrak{F}$ is
\begin{align*}
& \max_{\text{transition }(\sigma', \sigma'')}\frac{1}{\pi(\sigma')P(\sigma', \sigma'')}
\sum_{\substack{\sigma_2 \in \Omega_2 \\ \sigma_0 \in \Omega_0}}
\sum_{\substack{p_\varphi \in \mathcal{P}_{\sigma_2 \sigma_0} \\ p_\varphi \ni (\sigma',\sigma'')}}\frac{\mathfrak{W}(\sigma_2, \sigma_0, \varphi)}{\left(\mathcal{Z}(\Omega)\right)^2} \\
\le & \max_{\sigma' \in \Omega}\frac{O(n)}{\mathcal{W}(\sigma') \mathcal{Z}(\Omega)}
\sum_{\substack{\sigma_2 \in \Omega_2 \\ \sigma_0 \in \Omega_0}}
\sum_{\substack{\varphi \in \Phi_{\sigma_2 \oplus \sigma_0} \\ p_\varphi \ni \sigma'}}
\mathfrak{W}(\sigma_2, \sigma_0, \varphi) \\
\le & \max_{\sigma' \in \Omega}\frac{O(n)}{\mathcal{W}(\sigma') \mathcal{Z}(\Omega)}
\sum_{\sigma_2 \in \Omega_2}\sum_{\eta \in H_{\sigma'}}
\sum_{\varphi \in \Phi_\eta}
\mathfrak{W}(\sigma_2, \sigma_2 \oplus \eta, \varphi)\\
= & \max_{\sigma' \in \Omega}\frac{O(n)}{\mathcal{W}(\sigma') \mathcal{Z}(\Omega)}
\sum_{\eta \in H_{\sigma'}}
\sum_{\varphi \in \Phi_\eta}
\sum_{\sigma_2 \in \widetilde{\Omega}_2}
\mathfrak{W}(\sigma_2, \sigma_2 \oplus \eta, \varphi).
\end{align*}
On the last line above we exchange the order of summations
%where $\widetilde{\Omega}_2$ is the set of $\Omega_2$-states that are compatible with the symmetric difference $\eta$ and its {\sl atcp} $\varphi$.
where $\widetilde{\Omega}_2$ is the subset of  $\Omega_2$-states 
of the form $\sigma_2 = \eta \oplus \sigma_0$, for some
$\sigma_0 \in \Omega_0$ such that 
$p_{\varphi}$ (which passes through $\sigma'$) 
goes from $\sigma_2$ to $\sigma_0$.
These are $\Omega_2$-states
 ``compatible'' with the symmetric difference $\eta$ and its {\sl atcp} $\varphi$.
The number of states in $\widetilde{\Omega}_2$ is bounded by the length of the longest path $O(n)$ because $\sigma'$ is an intermediate state on a path.
Fix any $\sigma' \in \Omega$.
For any $\sigma_2 \in \Omega_2$, and $\eta \in H_{\sigma'}$ 
consisting of exactly one connected component with two endpoints of degree 1
and all other vertices having even degree
(and zero or more connected components of even degree vertices), observe that $\sigma' \oplus \eta \in \Omega'$. Indeed, if $\sigma' \in \Omega_0$ then $\sigma' \oplus \eta \in \Omega_2$; if $\sigma' \in \Omega_2$ then 
depending on whether $\sigma'$ 
\vspace{-2mm}
\setlist[enumerate]{itemsep=-1mm}
\begin{enumerate}[(1)]
\item is $\sigma_2$, or
\item appears in the process of reversing arrows on the trail with two endpoints, or
\item appears after reversing arrows on the trail with endpoints,
\end{enumerate}
\vspace{-2mm}
$\sigma' \oplus \eta$ lies in $\Omega_0$, $\Omega_2$, or $\Omega_4$, respectively.
For the edges not in $\eta$, $\sigma'$ agrees with $\sigma_2$ and $\sigma_2 \oplus \eta$ as the path $p_\varphi$ never ``touches'' them, and so does $\sigma' \oplus \eta$.
Recall that
\[\mathfrak{W}(\sigma_2, \sigma_2 \oplus \eta, \varphi) = \left(\prod_{u \in U\setminus U_4}f_u\left(\sigma_2 |_{E(u)}\right)f_u\left( (\sigma_2 \oplus \eta) |_{E(u)}\right)\right)\left(\prod_{u \in U_4}\mathfrak{w}(\varphi(u))\right).\] 
%For every degree-0 vertex $u \in U$ (this notion of degree is 
%in terms of the induced subgraph $\eta$, thus a degree-0 vertex $u \in U$ 
%is not in the induced subgraph $\eta$), $f_u$ takes the same value in all $\sigma_2$, $\sigma_2 \oplus \eta$, $\sigma'$, and $\sigma' \oplus \eta$.
For every vertex $u \in U$ that is not in $\eta$, $f_u$ takes the same value in all $\sigma_2$, $\sigma_2 \oplus \eta$, $\sigma'$, and $\sigma' \oplus \eta$.
For every vertex $u \in U$ that is degree-2 in $\eta$, assuming $M(f_u) = \left[\begin{smallmatrix} d & & & a \\ & b & c & \\ & c & b & \\ a & & & d\end{smallmatrix}\right]$,
$f_u\left(\sigma_2 |_{E(u)}\right)$ and $f_u\left( (\sigma_2 \oplus \eta) |_{E(u)}\right)$ take two different elements in $\{a, b, c, d\}$.  Meanwhile, $f_u\left(\sigma' |_{E(u)}\right)$ and $f_u\left(\sigma' \oplus \eta |_{E(u)}\right)$ also take these two elements (possibly in the opposite order).
%%%modified
For example, at the vertex $u$ shown in \figref{fig:xor_degree2},
$f_u\left(\sigma_2 |_{E(u)}\right) = a$ and $f_u\left(\sigma_2 \oplus \eta |_{E(u)}\right) = c$.
The two solid edges are in $\eta$ and assignments on the two dotted edges are 
shared by $\sigma_2$ and $\sigma_2 \oplus \eta$, as well as
$\sigma'$ and $\sigma' \oplus \eta$.
On the path $p_{\varphi}$ from $\sigma_2$ to $\sigma_2 \oplus \eta$ decided by $\varphi$: if $\sigma'$ appears before reversing the two solid edges, then $\sigma'$ agrees with $\sigma_2$ on them ($f_u\left(\sigma' |_{E(u)}\right) = a$) and $\sigma' \oplus \eta$ agrees with $\sigma_2 \oplus \eta$ on them ($f_u\left(\sigma' \oplus \eta |_{E(u)}\right) = c$); if $\sigma'$ appears after reversing the two solid edges, then $\sigma'$ agrees with $\sigma_2 \oplus \eta$ on them ($f_u\left(\sigma' |_{E(u)}\right) = c$) and $\sigma' \oplus \eta$ agrees with $\sigma_2$ on them ($f_u\left(\sigma' \oplus \eta |_{E(u)}\right) = a$).
%On the two solid edges
%$\sigma'$ either agrees with  $\sigma_2$  or  $\sigma_2 \oplus \eta$,
%and $\sigma' \oplus \eta$ is its reversal and agrees with the other.
For every vertex $u \in U$ that is degree-4 in $\eta$, $\mathfrak{w}(\varphi(u))$ takes the same value in $\mathfrak{W}(\sigma_2, \sigma_2 \oplus \eta, \varphi)$ and $\mathfrak{W}(\sigma', \sigma' \oplus \eta, \varphi)$ as the weight only 
depends on $\varphi(u)$, the signed pairing at $u$.

\captionsetup[subfigure]{labelformat=parens}
\begin{figure}[h!]
\centering
\begin{subfigure}[b]{0.15\linewidth}
\centering\includegraphics[width=\linewidth]{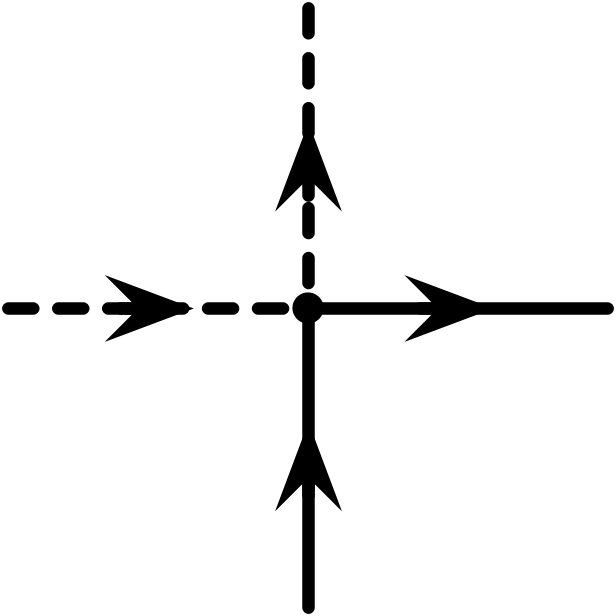}\caption{$\sigma_2$.}\label{fig:xor_degree2_before}
\end{subfigure}
\hspace{0.2\linewidth}
\begin{subfigure}[b]{0.15\linewidth}
\centering\includegraphics[width=\linewidth]{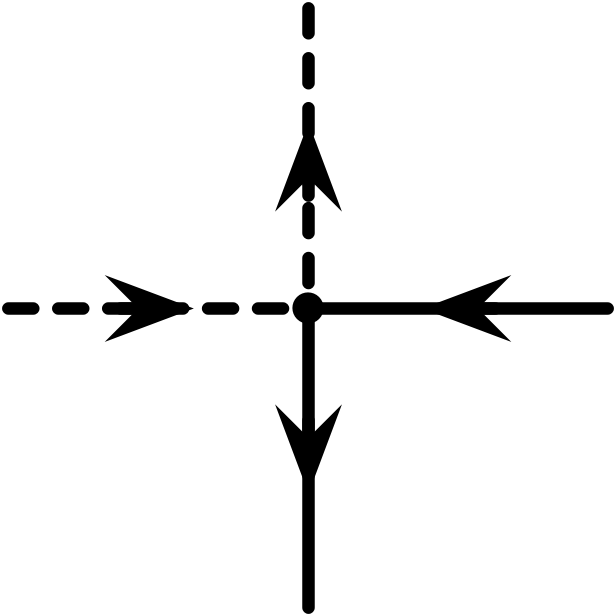}\caption{$\sigma_2 \oplus \eta$.}\label{fig:xor_degree2_after}
\end{subfigure}
\caption{}\label{fig:xor_degree2}
\end{figure}

By the above argument, we established that $\mathfrak{W}(\sigma_2, \sigma_2 \oplus \eta, \varphi) = \mathfrak{W}(\sigma', \sigma' \oplus \eta, \varphi)$.
Therefore, the congestion of $\mathfrak{F}$ can be bounded by
\begin{align*}
%& \max_{\sigma' \in \Omega}\frac{O({\color{magenta} n})}{\mathcal{W}(\sigma') \mathcal{Z}(\Omega)}
%\sum_{\sigma_2 \in \Omega_2}\sum_{\eta \in H_{\sigma'}}
%\sum_{\varphi \in \Phi_\eta}
%\mathfrak{W}(\sigma', \sigma' \oplus \eta, \varphi) \\
& \max_{\sigma' \in \Omega}\frac{O(n)}{\mathcal{W}(\sigma') \mathcal{Z}(\Omega)}
\sum_{\eta \in H_{\sigma'}}
\sum_{\varphi \in \Phi_\eta}
\sum_{\sigma_2 \in \widetilde{\Omega}_2}
\mathfrak{W}(\sigma', \sigma' \oplus \eta, \varphi) \\
\le & \max_{\sigma' \in \Omega}\frac{O(n^2)}{\mathcal{W}(\sigma') \mathcal{Z}(\Omega)}
\sum_{\eta \in H_{\sigma'}}
\sum_{\varphi \in \Phi_\eta}
\mathfrak{W}(\sigma', \sigma' \oplus \eta, \varphi) \\
\le & \max_{\sigma' \in \Omega}\frac{O(n^2)}{\mathcal{W}(\sigma') \mathcal{Z}(\Omega)}
\sum_{\eta \in H_{\sigma'}}
\mathcal{W}(\sigma')\mathcal{W}(\sigma' \oplus \eta) \\
= & \max_{\sigma' \in \Omega}\frac{O(n^2)}{\mathcal{Z}(\Omega)}
\sum_{\eta \in H_{\sigma'}}\mathcal{W}(\sigma' \oplus \eta) \\
%\le & \max_{\sigma' \in \Omega}\frac{O({\color{magenta} n^2})}{\mathcal{Z}(\Omega)}
\le & \frac{O(n^2)}{\mathcal{Z}(\Omega)}
\sum_{\sigma \in \Omega'}\mathcal{W}(\sigma) \\
%= & \max_{\sigma' \in \Omega}O({\color{magenta} n^2})\frac{\mathcal{Z}(\Omega')}{\mathcal{Z}(\Omega)}.
= & O(n^2)\frac{\mathcal{Z}(\Omega')}{\mathcal{Z}(\Omega)}.
\end{align*}
By a standard argument as in \cite{doi:10.1137/0218077, Mihail1996, DBLP:journals/corr/abs-1301-2880}, $\frac{\mathcal{Z}(\Omega_4)}{\mathcal{Z}(\Omega_2)} \le \frac{\mathcal{Z}(\Omega_2)}{\mathcal{Z}(\Omega_0)}$. Therefore, the congestion is bounded by $O(n^2)\frac{\mathcal{Z}(\Omega_2)}{\mathcal{Z}(\Omega_0)}$.
%Note that in each path, no edge is flipped more than once, so the length is $O(n)$.  
%%% JYC (x+y+z)/(y+z) <= (x+y)/(y+z) + 1 <= y/z + 1, if x/y <= y/z 
% here x=\Omega_4, y = \O_2, z = \O_0 
\end{proof}
\begin{remark}
\lemref{lem:congestion} can be alternatively derived using the notion of ``windability''~\cite{DBLP:journals/corr/abs-1301-2880}.
%However, this alternative derivation does not yield a proof of \thmref{thm:fpras};
%we still require the results from Section~\ref{sec:properties} to show that.
\end{remark}

%% file: hardness.tex
\documentclass[paper]{subfiles}

\begin{theorem}
	If $(a, b, c, d) \in \mathcal{F}_>$, then $Z(a,b,c,d)$ does not have an FPRAS unless RP=NP.
\end{theorem}

\begin{remark}
For any $(a, b, c, d) \in \mathcal{F}_>$, there are at least two nonzero numbers among $a$, $b$, $c$, and $d$.
The case $d = 0$ and $a, b, c >0$ was proved in ~\cite{DBLP:journals/corr/abs-1712-05880}.
The case $d = 0$ and one of $a, b, c$ is zero can be proved by a reduction from computing the partition function of the anti-ferromagnetic Ising model on 3-regular graphs;
 we postpone this proof to an expanded  version of this paper.
In this section, we prove the theorem when $d > 0$
and at least one of $a, b, c$ is positive.
\end{remark}
\begin{remark}
The construction in our proof for the cases when 
%$a$, $b$, or $c$ is large
$a > b+c+d$, or $b > a+c+d$, or $c > a+b+d$,
 is in fact a bipartite graph.
This means that approximating $Z(a,b,c,d)$ in those cases is NP-hard even for bipartite graphs.
\end{remark}
\begin{proof}
Let 3-MAX CUT denote the NP-hard problem of computing the cardinality of a maximum cut in a 3-regular graph~\cite{DBLP:conf/stoc/Yannakakis78}. We reduce 3-MAX CUT to approximating $Z(a,b,c,d)$.
We first prove the case when $a > b + c + d$, then adapt our proof to the case when $d > a + b + c$.
Since the proof of NP-hardness for $Z(a,b,c,d)$ is for general (i.e., not necessarily planar) graphs, we can permute the parameters $a, b, c$. Thus the proof for $b > a + c + d$ and $c > a + b + d$ is symmetric to the first case.

Before proving the theorem we briefly state our idea.
Denote an instance of 3-MAX CUT by $G = (V, E)$.
Given $V_+ \subseteq V$ and $V_- = V \setminus V_+$, an edge $\{u, v\} \in E$ is in the cut between $V_+$ and $V_-$ if and only if $(u \in V_+, v \in V_-)$ or $(u \in V_-, v \in V_+)$.
The maximum cut problem favors the partition of $V$ into $V_+$ and $V_-$ so that there are as many edges in $V_+ \times V_-$ as possible.
We want to encode this local preference on each edge by a local fragment of a graph $G'$ in terms of configurations in the eight-vertex model.

\captionsetup[subfigure]{labelformat=parens}
\begin{figure}[h!]
\centering
\begin{subfigure}[b]{0.32\linewidth}
\centering\includegraphics[width=0.7\linewidth]{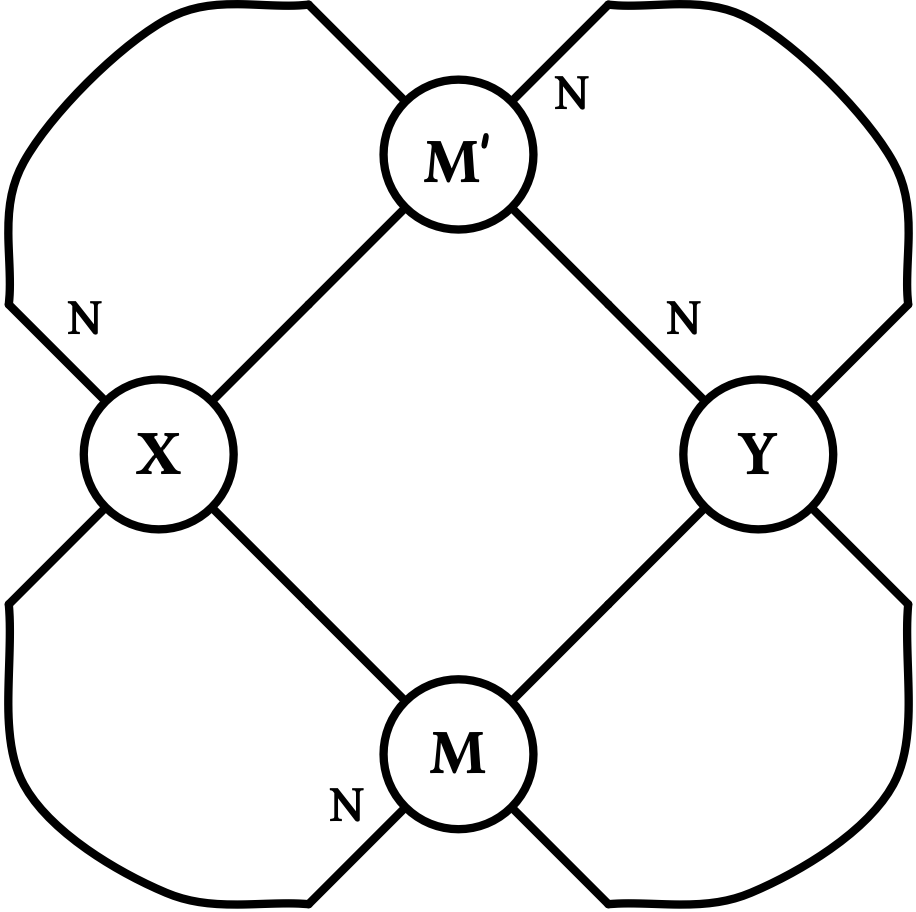}\caption{}\label{fig:hardness_single_a}
\end{subfigure}
\begin{subfigure}[b]{0.32\linewidth}
\centering\includegraphics[width=0.77\linewidth]{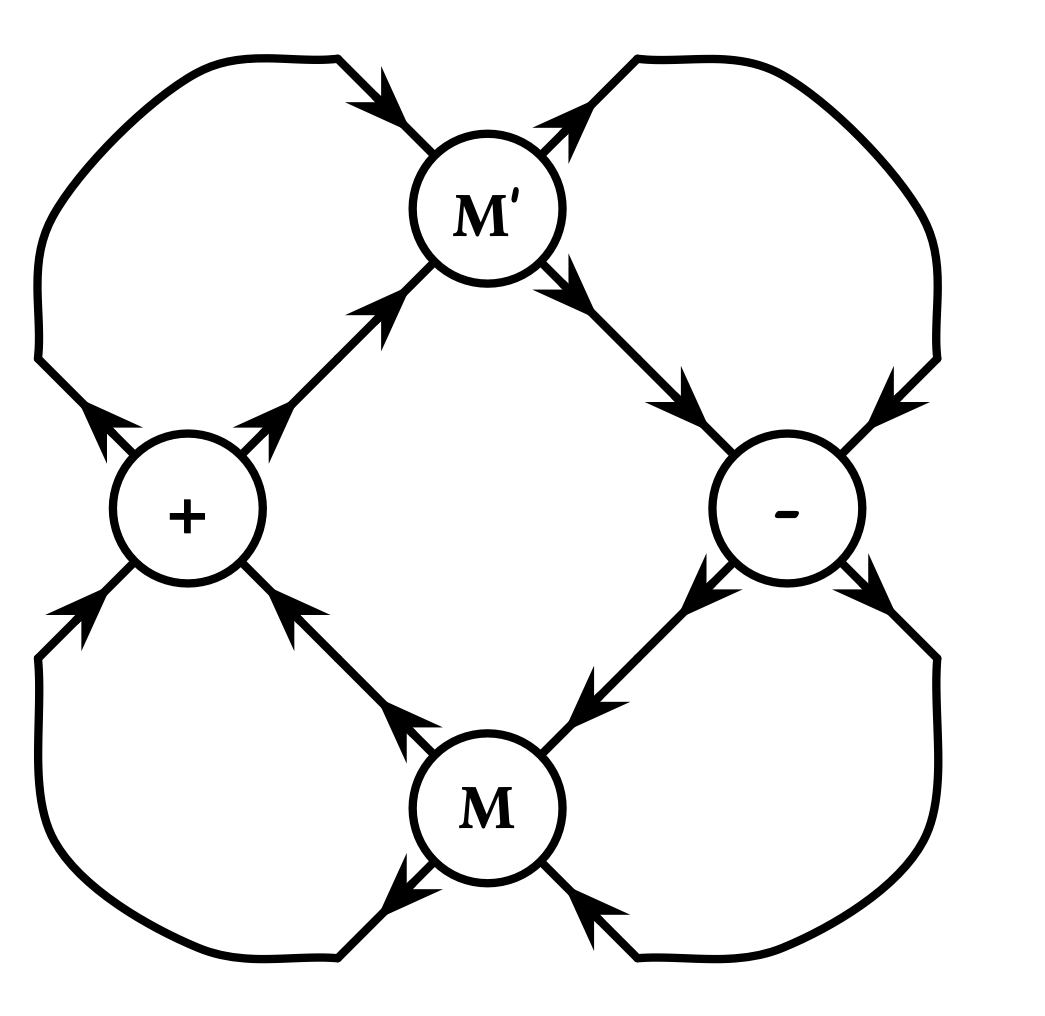}\caption{}\label{fig:hardness_single_a_+-}
\end{subfigure}
\begin{subfigure}[b]{0.32\linewidth}
\centering\includegraphics[width=0.77\linewidth]{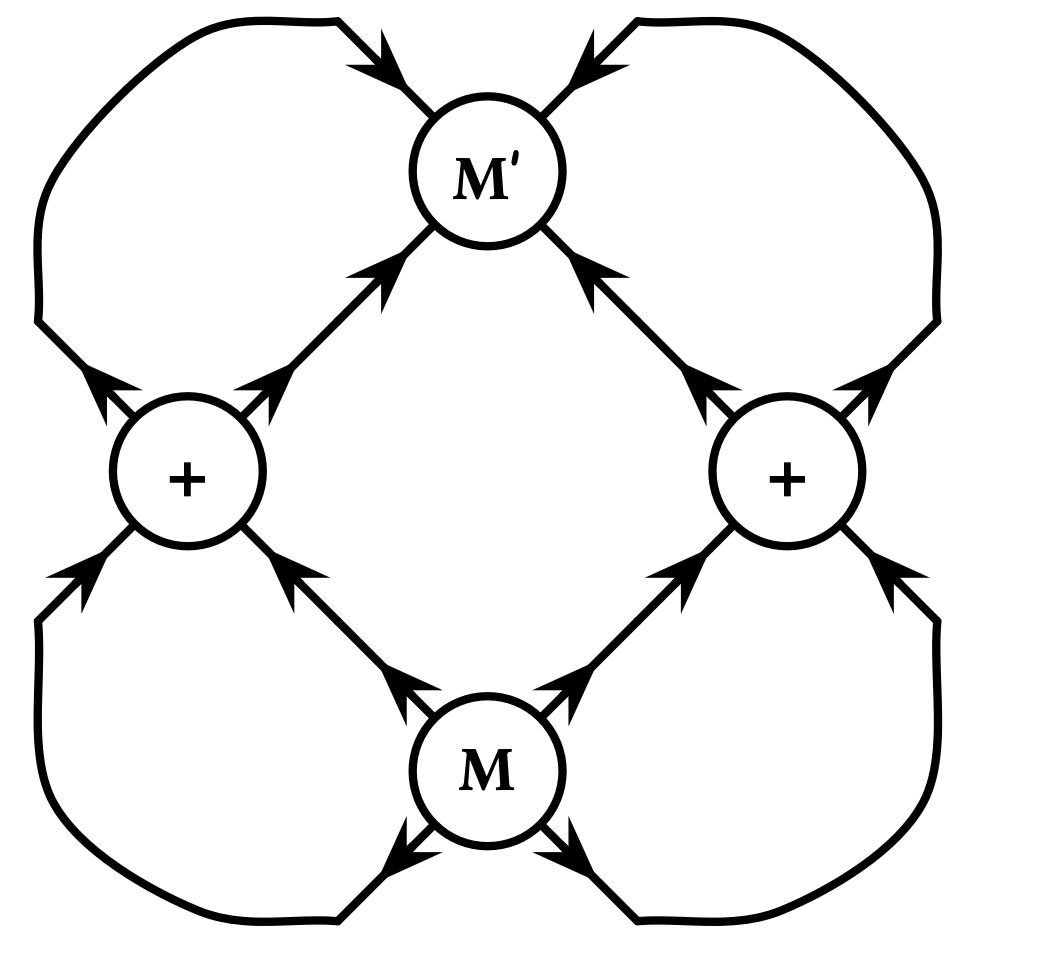}\caption{}\label{fig:hardness_single_a_++}
\end{subfigure}
\caption{A four-way connection implementing a single edge in 3-MAX CUT.}
\end{figure}

Let us start with the case when $a > b+c+d$. Recall that we require $d > 0$.
First we show how to implement a toy example{\textemdash}a single edge $\{u, v\}${\textemdash}by a construction in the eight-vertex model.
Suppose there are four vertices $X, Y, M, M'$ connected as in \figref{fig:hardness_single_a} shows.
The order of the 4 edges at each vertex is aligned to Figure~\ref{fig:orientations} by a rotation so that the edge marked by ``N'' corresponds to the north edge in \figref{fig:orientations}.
Let us impose the virtual constraint on $X$ and $Y$ so that the parameter setting on each of them is $\check{a} > \check{b} = \check{c} = \check{d} = 0$. (We will show how to implement this virtual constraint in the sense of approximation later.) In other words, the four edges incident on $X$ can only be in two possible configurations, \figref{fig:orientations_1} or \figref{fig:orientations_2}.
The same is true for $Y$.
We say $X$ (and similarly $Y$) is in state $+$ if its local configuration is in \figref{fig:orientations_1} (with the ``top'' two edges going out and the ``bottom'' two edges coming in); it is in state $-$ if its local configuration is in \figref{fig:orientations_2} (with the ``top'' two edges coming in and the ``bottom'' two edges going out).
Hence there are a total of 4 valid configurations given the virtual constraints.
When $(X, Y)$ is in state $(+, -)$ (or $(-,+)$), $M$ and $M'$ have 
local configurations both being  \figref{fig:orientations_1} (or 
both being  \figref{fig:orientations_2}), 
with weight $a$ (\figref{fig:hardness_single_a_+-});
when $(X, Y)$ is in state $(+,+)$ (or $(-,-)$),
$M$ and $M'$ have 
local configurations both being 
 \figref{fig:orientations_7} or \figref{fig:orientations_8},
%%%
% JYC don't want to be precise. actually it is (M,M') = (8,7) for (+,+) case and 
% (M,M') = (7,8) for (-,-) case.
 with weight $d < a$ (\figref{fig:hardness_single_a_++}).
This models how two adjacent vertices interact in 3-MAX CUT.
We will call the connection pattern 
described in \figref{fig:hardness_single_a}
between the set of 4 external edges incident to $X$ and the set of 
4 external edges incident to $Y$ (each with two on ``top'' and two
on ``bottom'')
a \emph{four-way connection}.

\captionsetup[subfigure]{labelformat=parens}
\begin{figure}[h!]
\centering
\begin{subfigure}[b]{0.32\linewidth}
\centering\includegraphics[width=0.8\linewidth]{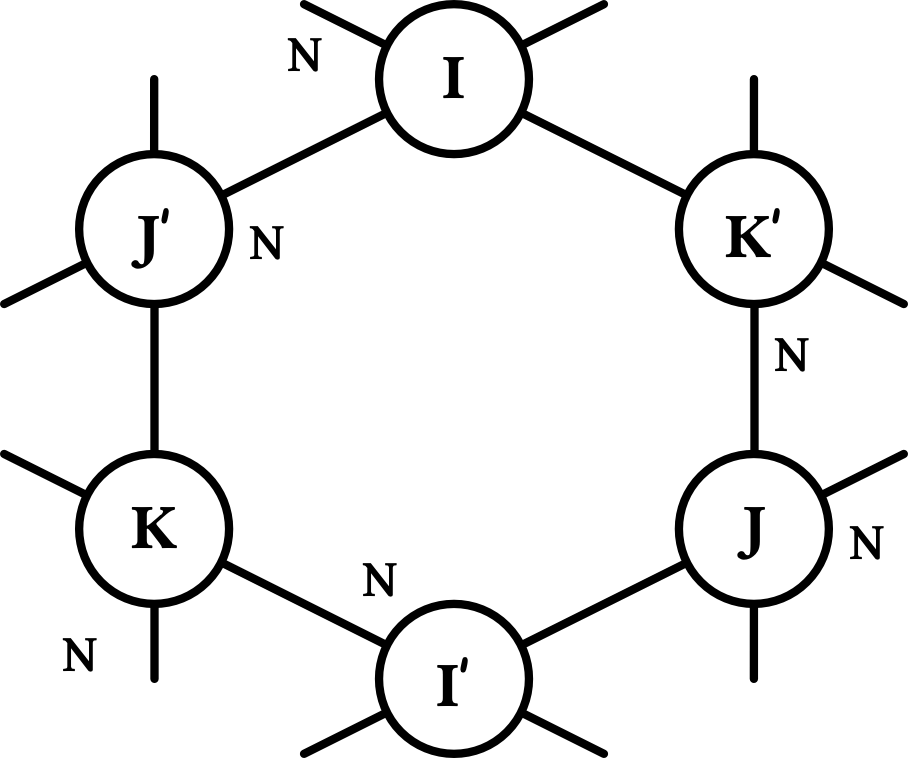}\caption{}\label{fig:hardness_multi}
\end{subfigure}
\begin{subfigure}[b]{0.32\linewidth}
\centering\includegraphics[width=0.95\linewidth]{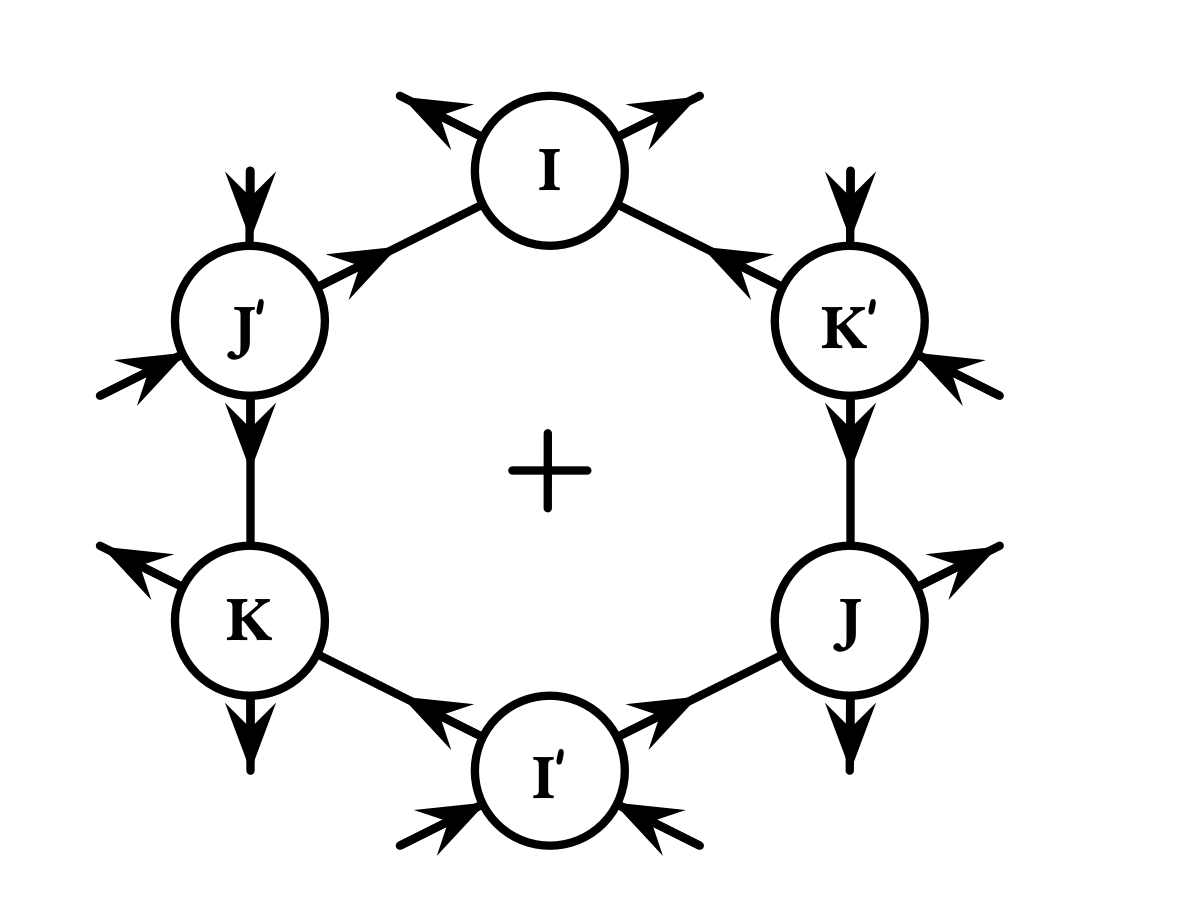}\caption{}\label{fig:hardness_multi_+}
\end{subfigure}
\begin{subfigure}[b]{0.32\linewidth}
\centering\includegraphics[width=0.95\linewidth]{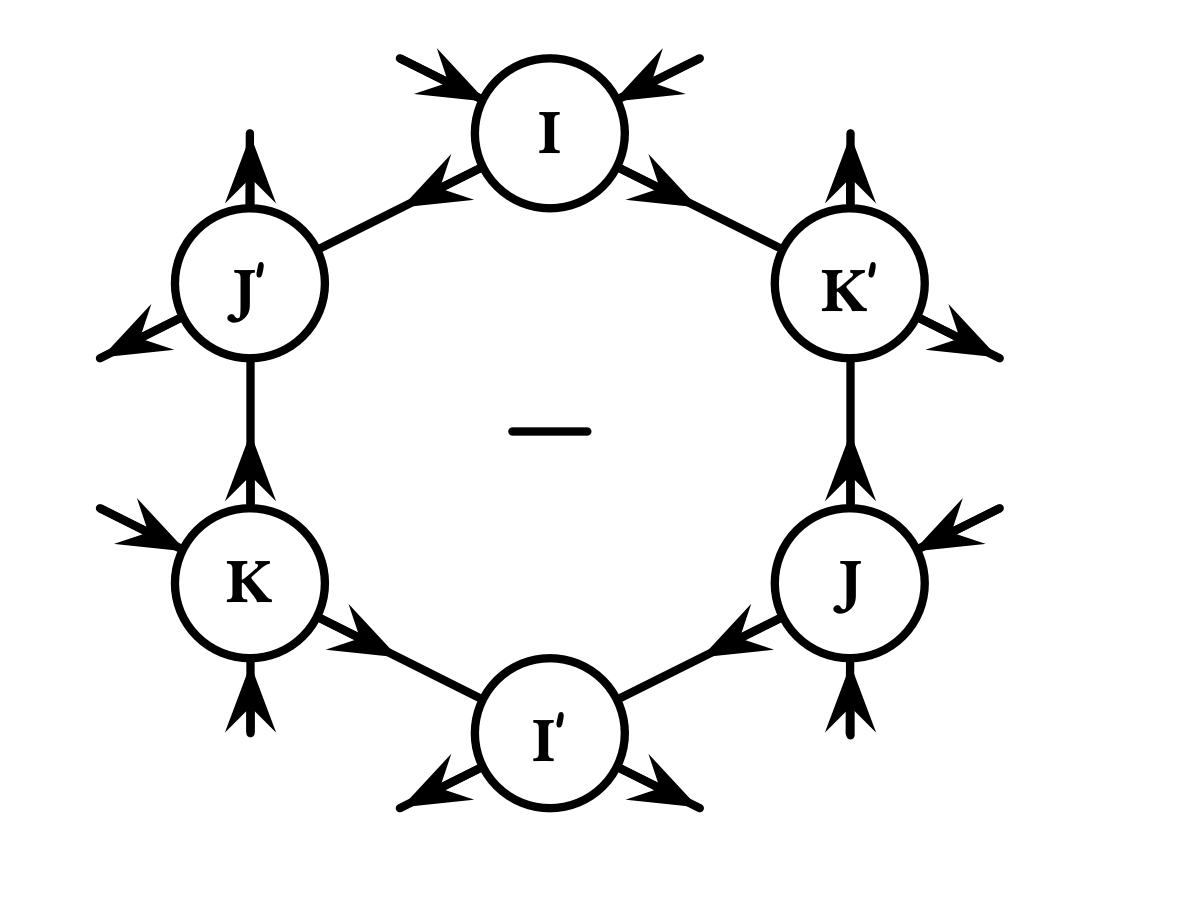}\caption{}\label{fig:hardness_multi_-}
\end{subfigure}
\caption{A locking device implementing a vertex of degree 3 in 3-MAX CUT.}
\end{figure}

To model a vertex of degree 3 in a 3-MAX CUT instance, we use the \emph{locking device} in \figref{fig:hardness_multi}.
Let us assume we have the virtual constraint that each of $I, I', J, J', K, K'$ can only be in two local configurations, \figref{fig:orientations_1} or \figref{fig:orientations_2}.
In fact, each locking device has two states, one shown in \figref{fig:hardness_multi_+} with every node in configuration \figref{fig:orientations_1} (called the $+$ state) and the other shown in \figref{fig:hardness_multi_-} with every node in configuration \figref{fig:orientations_2} (called the $-$ state).
If we think of the external edges incident to $I, J, K$ to serve as the ``top'' edges
(with ``N'' aligned with the ``N'' at $X$ or $Y$ in \figref{fig:hardness_single_a}), and the edges incident to $I', J', K'$ as the ``bottom'' edges there, then we 
 simulate the $\pm$ state of
 a degree 3 vertex as follows: (1) top edges are going out and bottom edges are coming in if the device is in $+$ state, and top edges are coming in and bottom edges are going out if the device is in $-$ state; and (2) the top edges on $I, J, K$ are going out or coming in at the same time.

\begin{figure}[h!]
\centering
\includegraphics[width=0.4\linewidth]{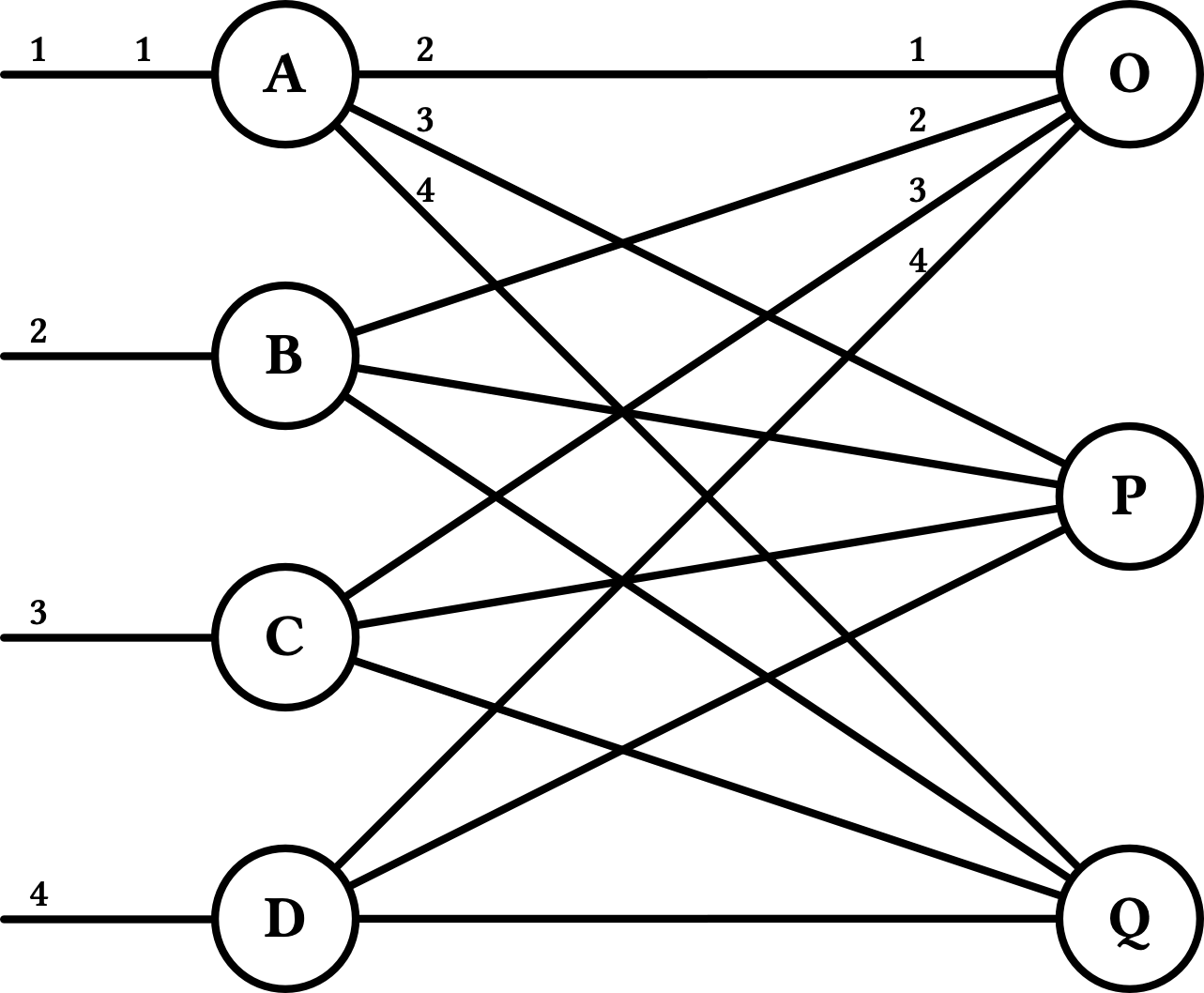}
\caption{A 4-ary construction that amplifies the maximum among $a, b, c, d$.}
\label{fig:hardness_bipartite}
\end{figure}

Next we show how to enforce the virtual constraint in \figref{fig:hardness_multi} that each vertex has two contrary configurations, in the sense of approximation. 
The idea is to implement an \emph{amplifier}  as a 4-ary construction with 
parameter $(\hat{a}, \hat{b}, \hat{c}, \hat{d})$ such that $\hat{a} \gg \hat{b} + \hat{c} + \hat{d}$ using polynomially many vertices in the eight-vertex model. We obtain such an amplifier by an iteration of $\Gamma$ shown in \figref{fig:hardness_bipartite} (input edges of $B, C, D$ labeled similarly as for $A$ and those of $P, 
Q$ labeled similarly as for $O$). Starting with $(a, b, c, d)$ on every vertex in $\Gamma$ (where $a > b+c+d$), 
the parameter setting $(a', b', c', d')$ of $\Gamma$ is
$\left\{\begin{smallmatrix}
a'=\Lambda(a, b,c,d)\\
b'=\Lambda(b, c,d,a)\\
c'=\Lambda(c, d,a,b)\\
d'=\Lambda(d, a,b,c)\\
\end{smallmatrix}\right.$,
where
\begin{multline}\label{eqn:Lambda}
\Lambda(\xi, x,y,z) = \xi^7 + (3 x^4 + 3 y^4 + 3 z^4 + 4 x^2 y^2 + 4 x^2 z^2 + 4 y^2 z^2) \xi^3 \\
+ (2x^4 y^2 + 2 x^4 z^2 + 2 x^2 y^4 + 2 y^4 z^2 + 2 x^2 z^4 + 2 y^2 z^4 + 30 x^2 y^2 z^2) \xi.
\end{multline}
%Jing Yu: It is not surprising that $a_1,b_1,c_1$ are symmetric. But how can $a_1,b_1,c_1,d_1$ be symmetric? If we flip the following eight edges: $e_3, e_4$ of $P,Q$ and $e_1,e_2$ of $R,T$. This is an natural bijection from the configurations of $f_{0000}$ to the configurations of $f_{1100}$ by exchanging $a$ and $d$, $b$ and $c$. By parity we can get the symmetry of $a_1,b_1,c_1,d_1$.

This construction uses  $7$ vertices and is called a \emph{$1$-amplifier}.
We obtain $(a_1, b_1, c_1, d_1) = (a', b', c', d')$ which amplifies the relative weight of configurations in \figref{fig:orientations_1} or \figref{fig:orientations_2}.
If we plug in the amplifier $\Gamma$ into each vertex of $\Gamma$ itself (called a \emph{$2$-amplifier}), we can obtain $(a_2, b_2, c_2, d_2)$ using $7^2$ vertices.
Iteratively, we can construct a series of constraint functions
with parameters $(a_k, b_k, c_k, d_k) \ (k \ge 1)$ such that $\left\{\begin{smallmatrix}
a_{k+1}=\Lambda(a_k, b_k,c_k,d_k)\\
b_{k+1}=\Lambda(b_k, c_k,d_k,a_k)\\
c_{k+1}=\Lambda(c_k, d_k,a_k,b_k)\\
d_{k+1}=\Lambda(d_k, a_k,b_k,c_k)\\
\end{smallmatrix}\right.$,
using $7^k$ vertices for each $k$ (called a \emph{$k$-amplifier}).
\lemref{lem:growth_rate} shows that the asymptotic growth rate is exponential in the number of vertices used.

To reduce the problem 3-MAX CUT to approximating $Z(a,b,c,d)$, 
let  $\kappa > \lambda \ge 1$ be two constants that will be fixed later.
For each 3-MAX CUT instance $G = (V, E)$ with $|V| = n$ and $|E| = m$, we construct a graph $G'$ where a device in \figref{fig:hardness_multi} is created for each $v \in V$, and a four-way connection is made for every $\{u, v\} \in E(G)$,
on the external edges 
corresponding to $\{u, v\}$ as in \figref{fig:hardness_single_a}.
For each 4-way connection in \figref{fig:hardness_single_a},
each of the nodes $M, M'$ is replaced by a $(\lambda \log n)$-amplifier
to boost the ratio of the configurations in
\figref{fig:orientations_1} or \figref{fig:orientations_2} over other configurations.
For each device in \figref{fig:hardness_multi},
each of the nodes  $I, I', J, J', K, K'$ is replaced by a $(\kappa \log n)$-amplifier to lock in the configurations \figref{fig:hardness_multi_+} or \figref{fig:hardness_multi_-}.
 
Next we argue that the maximum size $s$ of all cuts
 in $G$ can be recovered from an approximate solution to
$Z(G'; a, b, c, d)$.

Given a cut $(V_+,  V_-)$ of size $s$ in $G$, we show there is a valid configuration
(at the granularity of nodes
and edges shown in \figref{fig:hardness_multi}) of weight
 $\ge \left(a_{\kappa \log n}\right)^{6n} \left(a_{\lambda \log n}\right)^{2s} \left(d_{\lambda \log n}\right)^{2(m-s)}$.
For every vertex $u \in V_+$ and every $v \in V_-$  we set 
the corresponding locking devices in the $+$ state (\figref{fig:hardness_multi_+}) 
 and $-$ state (\figref{fig:hardness_multi_-}) respectively.
%(\figref{fig:hardness_multi_+}) and for every vertex $v \in V_-$ we set its corresponding locking device in the $+$ state (\figref{fig:hardness_multi_-}).
Consequently, for each edge $\{u, v\}$, the two nodes $M$ and $M'$ in the 4-way connection between the external edges from $u$ and $v$ (two from each) are both in \figref{fig:orientations_1} or \figref{fig:orientations_2} if $\{u, v\}$ is in the cut; they are both in \figref{fig:orientations_7} or \figref{fig:orientations_8} if $\{u, v\}$ is not in the cut.
We have defined a valid configuration, and it has weight  
$\ge \prod_{v \in V} \left({a_{\kappa \log n}}\right)^6 \prod_{e \in V_+ \times V_-} 
\left(a_{\lambda \log n}\right)^2 \prod_{e \not\in V_+ \times V_-} 
\left(d_{\lambda \log n}\right)^2 = \left(a_{\kappa \log n}\right)^{6n} \left(a_{\lambda \log n}\right)^{2s} \left(d_{\lambda \log n}\right)^{2(m-s)}$, where 
the exponent $6$ comes from the 6 nodes $I, I', J, J', K, K'$ in each locking device and $2$ comes from the two nodes $M, M'$ in each four-way connection.

We also show that the weighted sum of all configurations 
is  $<\frac{1}{2}\left(a_{\kappa \log n}\right)^{6n} \left(a_{\lambda \log n}\right)^{2(s+1)} \left(d_{\lambda \log n}\right)^{2(m-(s+1))}$,
where $s$ is the maximum size of cuts in $G$.
First we bound $W_{\rm lock}$, the sum of weights for configurations
where all nodes labeled $I, I', J, J', K, K'$ are locked.
It follows that
\[ W_{\rm lock} \le 2^{n} \left(a_{\kappa \log n}\right)^{6n} \sum_{i=0}^s\binom{m}{i} \left(a_{\lambda \log n}\right)^{2i} \left(d_{\lambda \log n}\right)^{2(m-i)} \le 2^{n+m} \left(a_{\kappa \log n}\right)^{6n} \left(a_{\lambda \log n}\right)^{2s} \left(d_{\lambda \log n}\right)^{2(m-s)}, \]
where each locking device has 2 possible states each with weight $\left(a_{\kappa \log n}\right)^{6}$,
and given a particular $\pm$ assignment of $n$ devices, there can be at most $s$ four-way connections that are between a $+$ device and a $-$ device.
Hence $W_{\rm lock} < \frac{1}{4}\left(a_{\kappa \log n}\right)^{6n} \left(a_{\lambda \log n}\right)^{2(s+1)} \left(d_{\lambda \log n}\right)^{2(m-(s+1))}$ when $\lambda \ge 1$ is large.

It remains to upper-bound the weighted sum of configurations where there is at least one device with some lock broken. This quantity is bounded by
\begin{align*}
& 8^{6n} \sum_{i = 0}^{6n - 1}\binom{6n}{i} \left(a_{\kappa \log n}\right)^{i} \left(b_{\kappa \log n} + c_{\kappa \log n} + d_{\kappa \log n}\right)^{(6n-i)} \left[2\left(a_{\lambda \log n} + b_{\lambda \log n} + c_{\lambda \log n} + d_{\lambda \log n}\right)\right]^{2m} \\
\le \; & 2^{24n + 2m} \left(a_{\kappa \log n}\right)^{6n} \left(\frac{b_{\kappa \log n} + c_{\kappa \log n} + d_{\kappa \log n}}{a_{\kappa \log n}}\right) \left(a_{\lambda \log n} + b_{\lambda \log n} + c_{\lambda \log n} + d_{\lambda \log n}\right)^{2m} \\
\le \; & 2^{24n + 6m} \left(a_{\kappa \log n}\right)^{6n}  \left(\frac{b_{\kappa \log n} + c_{\kappa \log n} + d_{\kappa \log n}}{a_{\kappa \log n}}\right)  \left(a_{\lambda \log n}\right)^{2m} \\
\le \; & 2^{24n + 6m} \left[\Theta(1)\right]^{m n^{6\lambda}} \left(a_{\kappa \log n}\right)^{6n} \frac{1}{\beta^{n^\kappa}},
\end{align*}
where we use the fact that $a_{\lambda \log n} \le a^{64^{{\lambda \log n}}} = a^{n^{6 \lambda}}$ because there are in total $64=2^6$ terms in (\ref{eqn:Lambda}) and $\beta > 1$ by \lemref{lem:growth_rate}.
This quantity is $< \frac{1}{4}\left(a_{\kappa \log n}\right)^{6n}$ when $\kappa > \lambda \ge 1$ is sufficiently large.

\medskip
We have finished the proof for $a > b+c+d$. The case when $b > a + c + d$ or $c > a + b + d$ can be similarly proved.
We now adapt our proof to the case when $d > a + b + c$. Since not all $a, b, c =0$ by our assumption, let us assume without loss of generality that $a > 0$. The amplifier remains exactly the same thanks to its symmetry. For the locking device, we can still lock into two states (with the help of amplifiers on each node in the device): the $+$ state where $I, J, K$ are sources and $I', J', K'$ are sinks; the $-$ state where $I, J, K$ are sinks and $I', J', K'$ are sources. The only difference in the construction is the way that four-way connections are set up (\figref{fig:hardness_single_d}). This time, a node marked with $I, J, K$ in one locking device need to be connected to a node marked with $I', J', K'$ (instead of still $I, J, K$ as is the case when $a > b + c + d$) in another locking device to make sure locking devices in contrary states are favored (by setting $M$ and $M'$ as sinks and sources).

\captionsetup[subfigure]{labelformat=parens}
\begin{figure}[h!]
\centering
\begin{subfigure}[b]{0.32\linewidth}
\centering\includegraphics[width=0.7\linewidth]{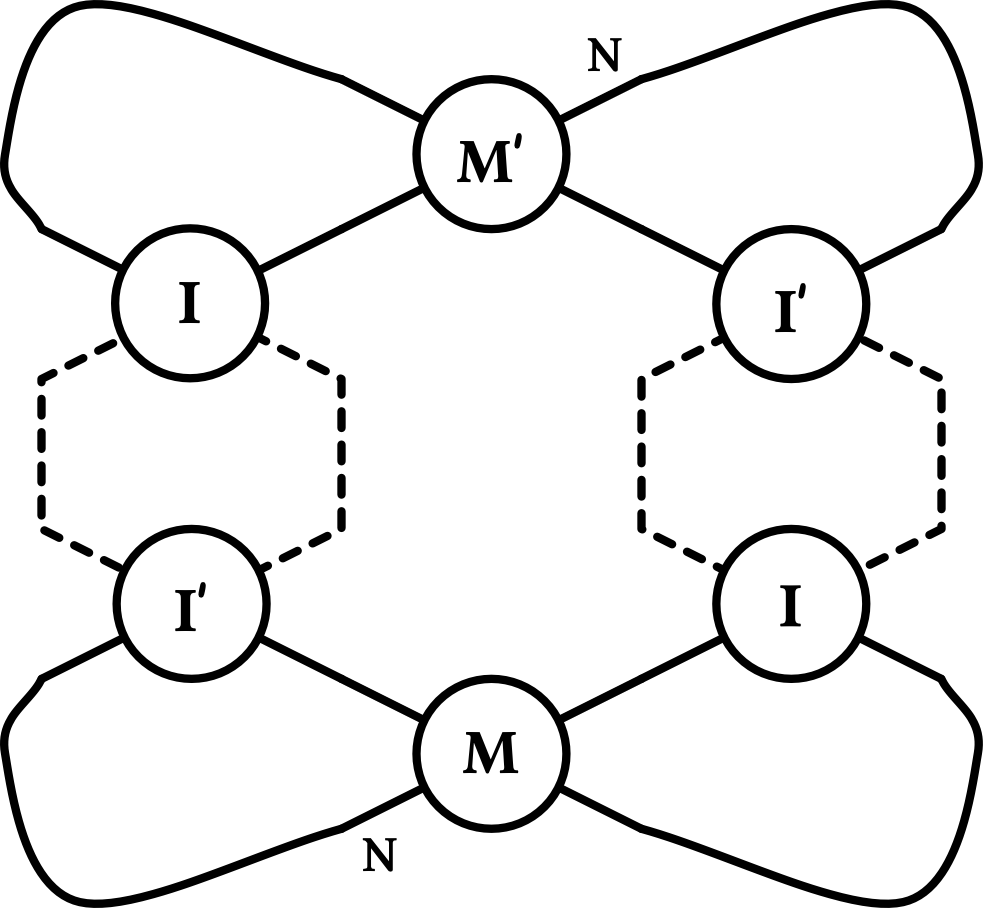}\caption{}\label{fig:hardness_single_d}
\end{subfigure}
\begin{subfigure}[b]{0.32\linewidth}
\centering\includegraphics[width=0.77\linewidth]{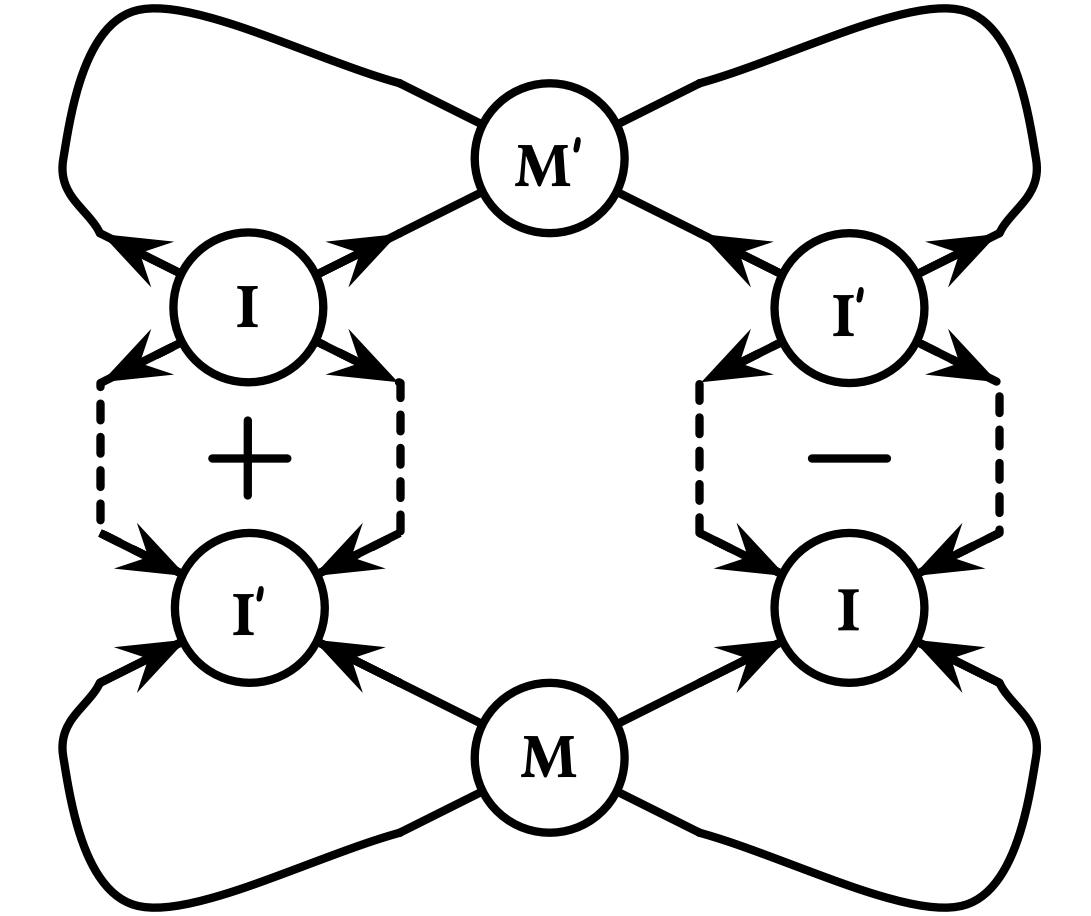}\caption{}\label{fig:hardness_single_d_+-}
\end{subfigure}
\begin{subfigure}[b]{0.32\linewidth}
\centering\includegraphics[width=0.83\linewidth]{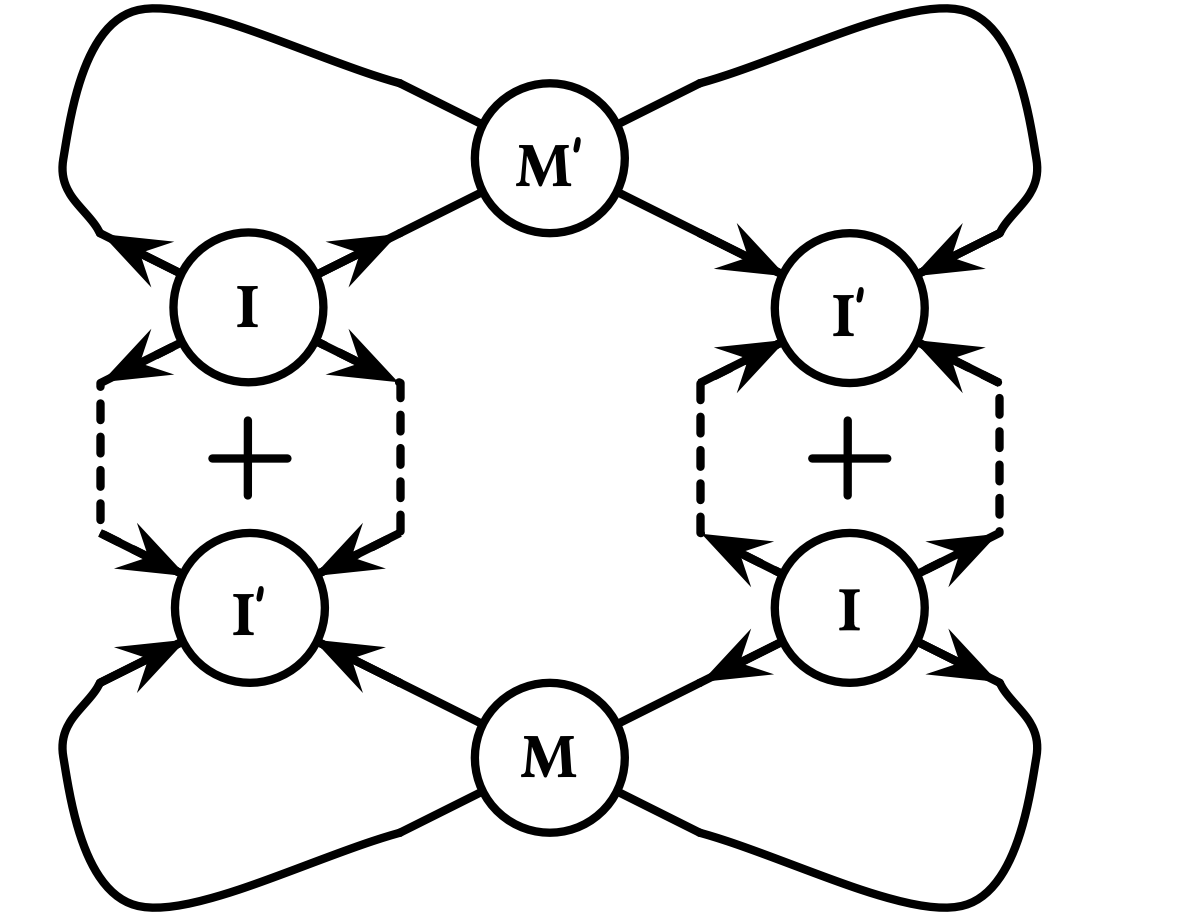}\caption{}\label{fig:hardness_single_d_++}
\end{subfigure}
\caption{Modifying the four-way connection for the case when $d > a + b + c$.}
\end{figure}

Note that in the case when $a > b + c + d$ (or symmetrically $b > a + c + d$, $c > a + b + d$), the construction $G'$ in the eight-vertex model is bipartite for any (not necessarily bipartite) 3-MAX CUT instance $G$. To see this, just check that (1) the amplifiers are bipartite and (2) the four way connections and the locking devices are bipartite by setting the nodes marked with $M, I, J, K$ on one side and the nodes marked with $M', I', J', K'$ on the other side.
Therefore, approximately computing $Z(a,b,c,d)$ in these cases is NP-hard even on bipartite graphs.
We remark this is no longer true for the construction of $G'$ in the case when $d > a + b + c$.
\end{proof}

\begin{lemma}\label{lem:growth_rate}
Let $(a_k, b_k, c_k, d_k) = \Lambda^{(k)} (a, b, c, d)$ given by (\ref{eqn:Lambda}). Assuming $a_0>b_0+c_0+d_0$, $a_0, d_0 > 0$, and $b_0, c_0 \ge 0$, there exists some constants $\alpha > 0, \beta > 1$ depending only on $a_0,b_0,c_0,d_0$ such that for all $k \ge 1$,
$\frac{a_k}{b_k+c_k+d_k} \ge \alpha \beta^{2^k}$.
\end{lemma}
\begin{proof}
Let $(a', b', c', d') = \Lambda(a, b, c, d)$ for any $a, b, c, d$ such that $a, d > 0$ and $b, c \ge 0$. We have
% $a', b', c', d' > 0$.
$a' > 0, b', c', d' \ge 0$ and $b' + c' + d'>0$.
 One can check that $\frac{a'}{b' + c' + d'} > \frac{a}{b + c + d}$ for any such $(a, b, c, d)$.
One can check that
\[ \frac{a'}{b' + c' + d'} - {\left(\frac{a}{b+c+d}\right)}^2 = \frac{a (a - (b + c + d)) F}{{(b + c + d)}^2(b' + c' + d')}, \]
where $F = (b^2 + c^2 + d^2 + 2bc + 2bd + 2cd) a^5 + F_4 a^4 + F_3 a^3 + F_2 a^2 + F_1 a + F_0$ and $F_i \ (0 \le i \le 4)$ are polynomials only in $b, c, d$
(with not necessarily positive coefficients).
%%% a^5 with >=d^2 >0 coeff. rest are lower deg in a.

Therefore, when $\frac{a}{b+c+d}$ is sufficiently large, $\frac{a'}{b' + c' + d'} - {\left(\frac{a}{b+c+d}\right)}^2 > 0$. This indicates that the series $\left\{ \frac{a_k}{b_k+c_k+d_k} \right\}_{k \ge 1}$ has a growth rate of $2^k$ after some finite $j$ such that $\frac{a_j}{b_j+c_j+d_j}$ is sufficiently large.
If we normalize $a = 1$, then $\Lambda$ takes a tuple $(\tilde{b}, \tilde{c}, \tilde{d})$ with $\tilde{b}, \tilde{c}, \tilde{d} \ge 0$ and $0 < \tilde{b}+\tilde{c}+\tilde{d} \le 1$ and maps it to another tuple $(\tilde{b'}, \tilde{c'}, \tilde{d'})$ with $\tilde{b'}+\tilde{c'}+\tilde{d'} < \tilde{b}+\tilde{c}+\tilde{d}$.
The existence of such a point $j$ is proved in \lemref{lem:lower_bound}. This completes the proof by setting $\beta = \frac{a_j}{b_j+c_j+d_j} > 1$.
\end{proof}

\begin{lemma}\label{lem:lower_bound}
Let $\Delta = \{ (b, c, d) \ |\ b, c, d \ge 0, b+c+d \le 1\}$. Let $s: \Delta \rightarrow \mathbb{R}_+$ be the summation function $s((b,c,d)) = b + c + d$. Suppose $g: \Delta 
\rightarrow \Delta$ is a continuous function such that for any $x \neq \mathbf{0}$, $s(g(x)) < s(x)$.
Then for any  $\epsilon > 0$, there exists $N \in \mathbb{Z}_+$ such that 
for any $x \in \Delta$, $s(g^{(N)}(x)) < \epsilon$.
\end{lemma}
\begin{proof}
For any $\epsilon > 0$, let   $\Delta_{\epsilon} 
 = \{ (b, c, d) \in \Delta \ |\ b+c+d \ge  \epsilon \}$,
and consider the continuous function 
$h(x) = s(x) - s(g(x))$ on $\Delta_{\epsilon}$.
Since $\Delta_{\epsilon}$ is \emph{compact}, $h$ 
reaches its minimum at some $x_0$ on $\Delta_{\epsilon}$.
Since $h(x) > 0$ for any $x \in \Delta_{\epsilon}$, we have $h(x_0) > 0$.
Let $N = \lceil \frac{1}{h(x_0)} \rceil$. 
Starting  from any $x \in \Delta$,
 we claim that $s(g^{(N)}(x)) < \epsilon$.
If not, then  $s(g^{(N)}(x)) \ge \epsilon$,
and by monotonicity, $g^{(n)}(x) \in \Delta_{\epsilon}$ for all $0 \le n \le N$.
But then $s(g^{(N)}(x)) \le s(x) - N h(x_0) \le 0$, a contradiction.
% \le 1 - N h(y_0) \le 0$, a contradiction.
\end{proof}
%\begin{lemma}\label{lem:lower_bound}
%Let $\Delta = \{ (b, c, d) \ |\ b, c, d \ge 0, a+b+c \le 1\}$. Let $s: \Delta \mapsto \mathbb{R}_+$ be the function that $s((b,c,d)) = b + c + d$. Suppose $g: \Delta \mapsto \Delta$ is a continuous function such that for any $x \neq \vec{0}$, $s(g(x)) < s(x)$.
%Then for any $x_0 \in \Delta$ and any $\epsilon > 0$, there exists $N \in \mathbb{Z}_+$ such that $s(g^{(N)}(x_0)) < \epsilon$.
%\end{lemma}
%\begin{proof}
%Fix $x_0 \in \Delta$. Let $\epsilon_0 = \inf_n \left\{ s(g^{n}(x_0)) \right\} \ge 0$.
%If $\epsilon_0 = 0$, there must exist an $N$ such that $s(g^{N}(x_0)) \le \frac{\epsilon}{2} < \epsilon$.
%Suppose $\epsilon_0 > 0$.
%Let $\Delta_\epsilon = \{ (b, c, d) \ |\ b, c, d \ge 0, \epsilon \le a+b+c \le 1\}$ for any $\epsilon \le 1$.
%Consider the continuous function $h(x) = s(x) - s(g(x))$ on $\Delta_{\epsilon_0}$.
%Since $\Delta_{\epsilon_0}$ is \emph{compact}, $h$ reaches its minimum at some $y_0$ on $\Delta_{\epsilon_0}$.
%Since $h(y) > 0$ for any $y \in \Delta_{\epsilon_0}$, we have $h(y_0) > 0$.
%Take $N = \lceil \frac{1}{h(y_0)} \rceil$, we have $s(g^{(N)}(x_0)) < \epsilon_0$. A contradiction.
%\end{proof}